\documentclass[11pt,draftclsnofoot,onecolumn]{IEEEtran}

\usepackage{mathtools}

\usepackage{amsmath,amsfonts,amssymb,amsthm}
\usepackage{mathrsfs}
\usepackage{comment,blkarray}
\usepackage{multirow,bigdelim}
\usepackage{cite}
\usepackage[ruled,commentsnumbered, vlined]{algorithm2e}
\usepackage{tikz}
\usepackage[latin1]{inputenc}
\usetikzlibrary{arrows,automata}
\usepackage{verbatim}
\usepackage{graphicx}
\usepackage{cases}
\usepackage{booktabs}
\usepackage{url}
\usepackage{subfigure}
\usepackage{multicol}
\usepackage[justification=centering]{caption}
\usepackage{longtable}	
\usepackage{ifpdf}															%
\ifpdf																		%
	\usepackage{hyperref}
\else																		%
\fi	
\usetikzlibrary{arrows,decorations.pathmorphing,decorations.footprints,fadings,calc,
trees,mindmap,shadows,decorations.text,patterns,positioning,shapes,matrix,fit}

\makeatletter
\newcommand{\rmnum}[1]{\romannumeral #1}
\newcommand{\Rmnum}[1]{\expandafter\@slowromancap\romannumeral #1@}

\newif\if@borderstar
\def\bordermatrix{\@ifnextchar*{%
  \@borderstartrue\@bordermatrix@i}{\@borderstarfalse\@bordermatrix@i*}%
}
\def\@bordermatrix@i*{\@ifnextchar[{\@bordermatrix@ii}{\@bordermatrix@ii[()]}}
\def\@bordermatrix@ii[#1]#2{%
\begingroup
  \m@th\@tempdima8.75\p@\setbox\z@\vbox{%
    \def\cr{\crcr\noalign{\kern 2\p@\global\let\cr\endline }}%
    \ialign {$##$\hfil\kern 2\p@\kern\@tempdima & \thinspace %
    \hfil $##$\hfil && \quad\hfil $##$\hfil\crcr\omit\strut %
    \hfil\crcr\noalign{\kern -\baselineskip}#2\crcr\omit %
    \strut\cr}}%
  \setbox\tw@\vbox{\unvcopy\z@\global\setbox\@ne\lastbox}%
  \setbox\tw@\hbox{\unhbox\@ne\unskip\global\setbox\@ne\lastbox}%
  \setbox\tw@\hbox{%
    $\kern\wd\@ne\kern -\@tempdima\left\@firstoftwo#1%
    \if@borderstar\kern2pt\else\kern -\wd\@ne\fi%
    \global\setbox\@ne\vbox{\box\@ne\if@borderstar\else\kern 2\p@\fi}%
    \vcenter{\if@borderstar\else\kern -\ht\@ne\fi%
    \unvbox\z@\kern -\if@borderstar2\fi\baselineskip}%
    \if@borderstar\kern -2\@tempdima\kern2\p@\else\,\fi\right\@secondoftwo#1 $%
  }\null \;\vbox{\kern\ht\@ne\box\tw@}%
\endgroup
}
\makeatother

\newtheorem{theorem}{Theorem}
\newtheorem{lemma}{Lemma}
\newtheorem{cor}{Corollary}

\newtheorem{example}{Example}
\newtheorem{definition}{Definition}

\allowdisplaybreaks

\renewcommand{\arraystretch}{0.93}

\newcommand{\Rank}{{\mathrm{Rank}}}

\SetKw{KwInitialization}{Initialization:}
\allowdisplaybreaks

\excludecomment{NEC-revisit}

\begin{document}

\title{Distributed Source Coding\\ for Compressing Vector-Linear Functions}

\author{Xuan~Guang,~~Xiufang~Sun,~~and~~Ruze~Zhang}

\maketitle

{
\begin{abstract}
Inspired by mobile satellite communication systems and the important and prevalent applications of computational tasks, we consider a distributed source coding model for compressing vector-linear functions, which consists of multiple sources, multiple encoders and a decoder linked to all the encoders. In the model, each encoder has access to a certain subset of the sources and the decoder is required to compute with zero error a vector-linear function, corresponding to a matrix~$T$, of the source information. The connectivity state between the sources and the encoders and the vector-linear function are all arbitrary.
From the information-theoretic point of view, we are interested in the function-compression capacity, which is defined by the minimum average number of times that the system is used for computing with zero error the vector-linear function once. This notion  measures the efficiency of using the system. For the nontrivial models with $1\!<\!\textup{Rank}(T)\!<\!s$, the explicit characterization of the function-compression capacity  in general is overwhelmingly difficult.
In the current paper, we first present a general lower bound on the function-compression capacity applicable to arbitrary connectivity states and vector-linear functions. Next, we confine to the nontrivial models with only three sources and no more than three encoders. We prove that all the $3\times2$ column-full-rank matrices $T$ can be divided into two types $T_1$ and~$T_2$, for which  the function-compression capacities  are identical if the matrices $T$ have the same type. We further introduce model isomorphism and prove that isomorphic models are of the same function-compression capacity. We explicitly characterize the capacities for two most nontrivial models associated with $T_2$ by a novel approach of both upper bounding and lower bounding the size of image sets of encoding functions. This shows that the lower bound thus obtained  is not always tight. Rather, by completely characterizing their function-compression capacities, the lower bound is tight for all the models associated with $T_1$ and  all the models associated with $T_2$ except for the two most nontrivial models. Furthermore, we apply the function-compression capacities for the two most nontrivial models to network function computation, and show that the best known upper bound proved by Guang \emph{et. al.} (2019) on computing capacity in network function computation is in general not tight for computing vector-linear functions which answers the open problem that whether this bound is always tight.
\end{abstract}

\section{Introduction}

In a generic mobile satellite communication system, a transmitter (source) can broadcast information to all the satellites (encoders) within its line of sight simultaneously; a satellite can combine and encode the information it receives from all the transmitters it covers and then broadcast the encoded information; and a receiver (decoder) can decode the encoded information it receives from all the satellites within the line of sight (cf. Fig.\,\ref{satellite}).
\begin{figure}[t]
    \centering
    \includegraphics[width=0.55\textwidth]{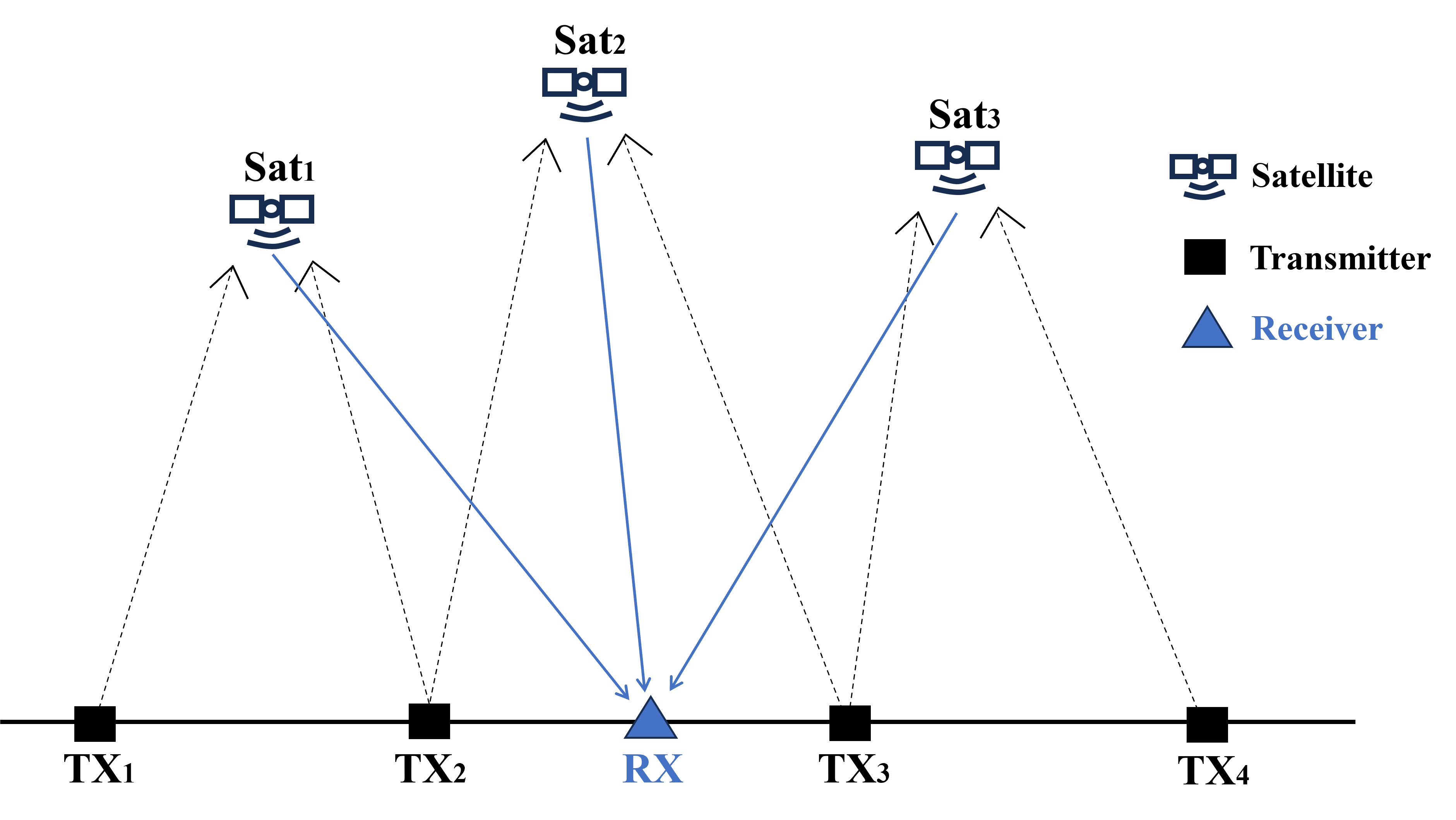}
    \caption{The satellite communication system.}
    \label{satellite}
    \vspace{-2em}
\end{figure}
Inspired by the system and the important and prevalent applications of computational tasks in the era of big data and artificial intelligence, we consider a new problem called the \emph{distributed source coding problem for function compression}.
A distributed source coding model associated with a single decoder for function compression consists of multiple sources, multiple encoders and a decoder linked to all the encoders, where each encoder has access to a certain subset of the sources, and the decoder is required to compute a target function of the source information.
Here, the bandwidth constraints of upload links between the sources and the encoders and the download links between the encoders and the decoder are asymmetric. We assume that the bandwidth of the upload links is unlimited and the bandwidth of the download links is limited because, for instance, transmitters can be equipped with high-power uplink amplifiers, and satellite platforms are in general power-limited.


\vspace{-0.5em}

\subsection{Related Works}

Inspired by the mobile satellite communication systems, Yeung and Zhang \cite{Yeung_distributed_source_coding} put forward the distributed source coding problem with multiple sources, multiple encoders and multiple decoders, where each decoder is required to reconstruct a certain subset of the sources. This work, together with the class of multilevel diversity coding problems \cite{Roche_symmetrical,Yeung_distortion,Yeung_symmetrical}, was subsequently generalized to the class of \emph{network coding} problems, e.g., \cite{Network_information_flow,linear,alg,Zhang-book,yeung08b},
which launches a new direction in multiterminal source coding.
Another related line of research is function computation/compression, e.g., \cite{Korner-Marton-IT73,Doshi_fun_comp_graph_color_sch,Feizi-Medard,
Orlitsky-Roche_general_side_inf_model_rat_reg,Witsenhausen-IT-76,Alon-Orlitsky_source_cod_graph_entropies,
Guang_Zhang_Arithmetic_sum_TIT,Appuswamy11,Huang_Comment_cut_set_bound,Guang_Improved_upper_bound,
Appuswamy-lin-func-lin-code,Li_Xu_vector_linear_diamond,Guang_Zhang_Arithmetic_sum_Sel_Areas}.
To our knowledge, the first non-identity function compression problem is the K\"{o}rner-Marton problem \cite{Korner-Marton-IT73}, in which the decoder requires to compute the modulo $2$ sum of two correlated sources. Doshi \emph{et al.} \cite{Doshi_fun_comp_graph_color_sch} generalized the K\"{o}rner-Marton model by requiring the decoder to compress with asymptotically zero error an arbitrary function of two correlated sources. Feizi and M{\'e}dard \cite{Feizi-Medard} further investigated the function compression over a tree network. Orlitsy and Roche \cite{Orlitsky-Roche_general_side_inf_model_rat_reg} considered an asymptotically-zero-error function compression problem with side information at the decoder.
The zero-error source coding problem with side information at the decoder was considered by Witsenhausen \cite{Witsenhausen-IT-76} and further developed by Alon and Orlitsky \cite{Alon-Orlitsky_source_cod_graph_entropies}. Guang and Zhang \cite{Guang_Zhang_Arithmetic_sum_TIT} investigated the zero-error distributed compression problem of binary arithmetic sum.
In the network function computation,
a single sink node is required to compute a function of the source messages generated by multiple source nodes over a directed acyclic network.
Appuswamy \emph{et al.} \cite{Appuswamy11} investigated the fundamental \emph{computing capacity}.
Motivated by the cut-set based upper bound obtained in \cite{Appuswamy11}, Huang \emph{et al.} \cite{Huang_Comment_cut_set_bound} obtained a general upper bound on the computing capacity, where ``general'' means that the upper bound can be applied for arbitrary functions and arbitrary network topologies. Subsequently, Guang \emph{et al.} \cite{Guang_Improved_upper_bound} proved an improved general upper bound by using a novel approach of the cut-set strong partition, which is not only a strict improvement over the previous upper bounds but also tight for all the considered network function computation problems previous to~\cite{Guang_Improved_upper_bound} whose computing capacities are known.

\vspace{-0.5em}

\subsection{Contributions and Organization of the Paper}
In this paper, we confine our discussion to distributed source coding model for compressing vector-linear functions denoted by $(s,m,\Omega,T)$, where $s$ and $m$ are the number of sources and the number of encoders, respectively; $\Omega$ is an arbitrary connectivity state between the sources and the encoders; and $T$ is the matrix corresponding to an arbitrary vector-linear function over a finite field.
From the information-theoretic point of view, we are interested in the \emph{function-compression capacity} for the model $(s,m,\Omega,T)$, which is defined by the minimum average number of times that the system is used for computing with zero error the target function once. This notion measures the efficiency of using the system,
rather than the notion of compression capacity considered in most previously studied multiterminal source coding models in which how to efficiently establish a system is investigated, e.g., lossless source coding models \cite{Korner-Marton-IT73,Orlitsky-Roche_general_side_inf_model_rat_reg,
Doshi_fun_comp_graph_color_sch,Feizi-Medard,Slepian-Wolf-IT73}, zero-error source coding models \cite{Witsenhausen-IT-76,Alon-Orlitsky_source_cod_graph_entropies,Koulgi_zero-error-cod_cor_inf_sour}, and lossy source coding models \cite{Wyner-Ziv_rat_distortion_fun_sid_inf,Yamamoto_rat_distortion_gener_fun_sid_inf,
Berger-Yeung_multi_source_cod_dist_cri}. 
The main contributions and organization of the paper are given as follows.
\begin{itemize}
\item  In Section \ref{DSC-FC}, we formally present the distributed source coding model for compressing vector-linear functions and define the function-compression capacity for the model from the viewpoint of  system usage. Further, we show that each specified model can be transformed into an equivalent model of network function computation. This equivalence is useful to characterize the function-compression capacity in the rest of the paper.
\item  Some preparatory results are given in Section \ref{pre-result}. We first characterize the function-compression capacities for two special classes of models $(s,m,\Omega,T)$ with $\textup{Rank}(T)=1$ and $\textup{Rank}(T)=s$, which correspond to compressing the scalar-linear function and the identity function, respectively. For the other nontrivial models $(s,m,\Omega,T)$ with $1\!<\!\textup{Rank}(T)\!<\!s$, we present a lower bound on the function-compression capacity by applying the equivalence to the model of network function computation and  the best known upper bound proved by Guang~\emph{et al.} \cite{Guang_Improved_upper_bound} on the computing capacity in network function computation. Next, we focus on the nontrivial models $(3,m,\Omega,T)$ with only three sources (i.e., $s=3$) considered in the rest of the paper and divide all the $3\times2$ column-full-rank matrices $T$ into two types $T_1$ and $T_2$, for which we prove that the function-compression capacities for the model $(3,m,\Omega,T)$ are identical if the matrices $T$ have the same type. We further introduce a notion of model isomorphism, and then prove that isomorphic models are of the same function-compression capacity.
\item  Section \ref{section-not-tight} is devoted to the capacity characterization for two most nontrivial models associated with~$T_2$, which shows that the lower bound thus obtained on the function-compression capacity is not always tight. Rather, the  lower bound is tight for each model $\big(3,m,\Omega,T_1\big)$ with  arbitrary connectivity states~$\Omega$ for $1\!\leq\!m\!\leq\!3$ and each model $\big(3,m,\Omega,T_2\big)$ with  arbitrary connectivity states~$\Omega$ for $1\!\leq\!m\!\leq\!3$ except for the two most nontrivial models, which are discussed in Sections~\ref{section-3-m-T1} and \ref{section-3-m-T2}, respectively.
    Following an intuitive explanation why the lower bound is not tight, we characterize the function-compression capacities for the two models, in which we not only upper bound but also lower bound the size of image sets of encoding functions by a novel approach to obtain an improved converse proof. An important application of the function-compression capacity for the two models is in the tightness of the best known upper bound on the computing capacity in network function computation, where whether this upper bound is in general tight or not was given as an open problem in \cite{Guang_Improved_upper_bound}. The function-compression capacity for the two models implies that the best known upper bound is in general not tight for computing vector-linear functions, rather than computing scalar-linear functions of which the computing capacities over arbitrary network topologies are  characterized by this upper bound.
\item  In Sections \ref{section-3-m-T1} and \ref{section-3-m-T2}, we characterize the function-compression capacities, respectively, for all the models $(3,m,\Omega,T_1)$ with arbitrary connectivity states~$\Omega$ for $1\leq m\leq3$ and all the models $(3,m,\Omega,T_2)$ with arbitrary connectivity states~$\Omega$ for $1\leq m\leq3$ except for the above two most nontrivial models. We first specify the obtained lower bound for each model therein and then prove that the specified lower bound is  tight  by designing optimal source codes for isomorphic models.
\item  In Section \ref{conclusion}, we conclude with a summary of our results and a remark on future research.
\end{itemize}
}

\section{Distributed Source Coding for Function Compression}\label{DSC-FC}

\subsection{Model}\label{Model}

We consider a \emph{distributed source coding model for function compression} as depicted in Fig.\,\ref{DSCM-FC}, in which there are $s$ \emph{sources} $\sigma_1, \sigma_2, \cdots, \sigma_s$, $m$ \emph{encoders} $v_1, v_2, \cdots, v_m$, and a single \emph{decoder} $\rho$ linked to all the~$m$ encoders. We further let $S=\{\sigma_1, \sigma_2, \cdots, \sigma_s\}$ and $V=\{v_1,v_2,\cdots,v_m\}$. Each source $\sigma_i\in S$ generates a sequence of symbols in a finite alphabet $\mathcal{A}$ and transmits the sequence to the encoders in a given subset~$\Gamma_{\sigma_i}$ of $V$. Here, we use $\Gamma_{\sigma_i}$ to denote the set of the encoders that are able to receive the sequence of symbols generated by $\sigma_i$, i.e.,
\begin{equation}\label{Gamma-sigma-i}
\Gamma_{\sigma_i}=\big\{ v\in V:~ \sigma_i\rightarrow v \big\},
\end{equation}
where we use $\sigma_i\rightarrow v$ to represent that the encoder $v$ can receive the sequence of symbols generated by~$\sigma_i$. Further, we let $\Omega\triangleq\big(\Gamma_{\sigma_1},\Gamma_{\sigma_2},\cdots,\Gamma_{\sigma_s}\big)$, which can be regarded as the \emph{state of connectivity} between the sources and the encoders. Here, we assume without loss of generality that $\Gamma_{\sigma_i}\neq\emptyset$ for each $\sigma_i\in S$ and $\bigcup\limits_{i=1}^s \Gamma_{\sigma_i}=V$. Dually, for each $v_j\in V$, $1\leq j\leq m$, we let
\begin{equation*}\label{Lambda-j}
\Theta(v_j)\triangleq\big\{ \sigma\in S:~ \sigma\rightarrow v_j \big\}=\big\{ \sigma\in S:~ v_j\in\Gamma_{\sigma} \big\}, \footnote{Here, the assumption $\bigcup_{i=1}^s \Gamma_{\sigma_i}=V$ implies that $\Theta(v_j)\neq\emptyset$ for each $v_j\in V$.}
\end{equation*}
which is the set of the sources whose sequences can be received by the encoder $v$. We can readily see that $\big(\Gamma_{\sigma_1},\Gamma_{\sigma_2},\cdots,\Gamma_{\sigma_s}\big)$ and $\big(\Theta(v_1),\Theta(v_2),\cdots,\Theta(v_m)\big)$ are one-to-one corresponding, and thus in the rest of the paper, we write either $\Omega=\big(\Gamma_{\sigma_1},\Gamma_{\sigma_2},\cdots,\Gamma_{\sigma_s}\big)$ or $\Omega=\big(\Theta(v_1),\Theta(v_2),\cdots,\Theta(v_m)\big)$ according to the convenience of discussion. Furthermore, for each use of the link $e_j\triangleq(v_j,\rho)$ connecting an encoder~$v_j$ to the decoder $\rho$, a symbol in $\mathcal{A}$ can be reliably transmitted from $v_j$ to $\rho$, i.e., we take the capacity of each link to be $1$ with respect to the alphabet $\mathcal{A}$. Next, consider a nonconstant function $f:~ \mathcal{A}^s\rightarrow \textup{Im}\,f$,\footnote{In this paper, we use $\textup{Im}\,f$ to denote the image set of a function $f$.} called the \emph{target function}, that is needed to be computed with zero error at the decoder $\rho$. Without loss of generality, we assume that the $i$th argument of the target function $f$ is generated at the~$i$th source $\sigma_i$ for $1\leq i\leq s$. We have completed the specification of the distributed source coding model for function compression, which is denoted by $(s,m,\Omega,f)$.
\begin{figure}[http]
\vspace{-0.5em}
	\centering
\tikzstyle{source}=[draw,circle,fill=gray!20, minimum size=6pt, inner sep=0pt]
\tikzstyle{encoder}=[draw,circle,fill=gray!20, minimum size=6pt, inner sep=0pt]
\tikzstyle{user}=[draw,circle,fill=gray!20, minimum size=6pt, inner sep=0pt]
\tikzstyle{blackdot}=[circle, fill, inner sep=1pt]

	\begin{tikzpicture}


        \node at (-5.8,6) {sources};
		\node[source](a1)at(-4,6){};
        \node at (-4,6.4) {$\sigma_1$};
		\node[source](a2)at(-1.6,6){};
        \node at (-1.6,6.4) {$\sigma_2$};
        \node[blackdot](a3) at (-0.3,6){};
        \node[blackdot](a4) at (0,6){};
        \node[blackdot](a5) at (0.3,6){};
		\node[source](a6)at(1.8,6){};
        \node at (1.8,6.4) {$\sigma_s$};

        \draw[->,>=latex,thick] (a1) -- (-4.5,4.8) node {};
        \draw (-3.95,5.2) node  {$\cdots$};
        \draw[->,>=latex,thick] (a1) -- (-3.5,4.8) node {};
        \node at (-4,4.5) {$\Gamma_{\sigma_1}$};

        \draw[->,>=latex,thick] (a2) -- (-2.1,4.8) node {};
        \draw (-1.55,5.2) node  {$\cdots$};
        \draw[->,>=latex,thick] (a2) -- (-1.1,4.8) node {};
        \node at (-1.6,4.5) {$\Gamma_{\sigma_2}$};

        \draw[->,>=latex,thick] (a6) -- (1.3,4.8) node {};
        \draw (1.85,5.2) node  {$\cdots$};
        \draw[->,>=latex,thick] (a6) -- (2.3,4.8) node {};
        \node at (1.8,4.5) {$\Gamma_{\sigma_s}$};

        \node at (-5.8,3.7) {encoders};
		\node[encoder](r1)at(-4,3.7){};
        \node at (-4.4,3.7) {$v_1$};
		\node[encoder](r2)at(-1.6,3.7){};
        \node at (-2,3.7) {$v_2$};
        \node[blackdot](r3) at (-0.3,3.7){};
        \node[blackdot](r4) at (0,3.7){};
        \node[blackdot](r5) at (0.3,3.7){};
		\node[encoder](rm)at(1.8,3.7){};
        \node at (2.2,3.7) {$v_m$};

		\node at (-5.8,1.5) {decoder};
        \node[user](p)at(-1,1.5){};
        \node at (-1,1.1) {$\rho$};

		\draw[->,>=latex,thick](r1)--(p)  node[midway, auto,swap, left=0mm] {};
		\draw[->,>=latex,thick](r2)--(p)  node[midway, auto,swap, right=0mm] {};
		\draw[->,>=latex,thick](rm)--(p)  node[midway, auto,swap, right=0mm] {};
	\end{tikzpicture}
\vspace{-0.5em}
	\caption{The distributed source coding model for function compression.}
\label{DSCM-FC}
\vspace{-1em}
\end{figure}
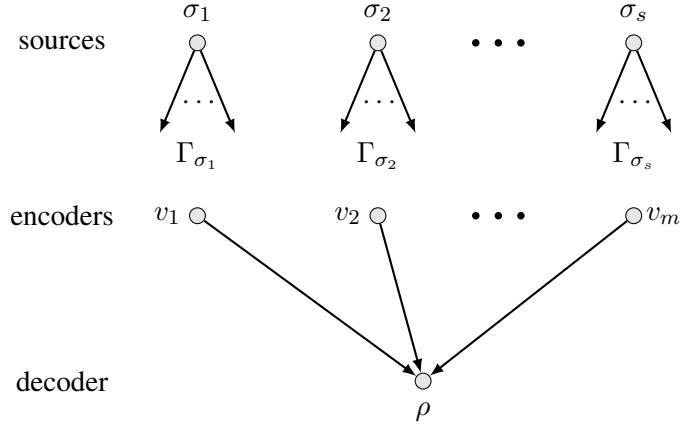

We consider computing the target function $f$ multiple times over the model $(s,m,\Omega,f)$. To be specific, let $k$ be a positive integer. Each source $\sigma_i$, $1\leq i\leq s$ generates a sequence of $k$ symbols $\boldsymbol{x}_i\triangleq(x_{i,1}, x_{i,2}, \cdots, x_{i,k})\in\mathcal{A}^k$, which is called the \emph{source message} generated by $\sigma_i$. At the single decoder~$\rho$, the $k$ values of the target function $f$
\[
f(\boldsymbol{x}_1,\boldsymbol{x}_2,\cdots,\boldsymbol{x}_s)\triangleq\big( f(x_{1,\ell}, x_{2,\ell}, \cdots, x_{s,\ell}):~ \ell=1,2,\cdots,k \big)
\]
are required to be computed with zero error. Toward this end, we define a \emph{$k$-shot (function-compression) source code} $\mathbf{C}$ for $(s,m,\Omega,f)$, which consists of
\begin{itemize}
\item an \emph{encoding function} for each encoder $v_j$, $1\leq j\leq m:$
\begin{equation*}\label{enc-fun}
\varphi_j:~ \mathcal{A}^{k\cdot|\Theta(v_j)|}\rightarrow \textup{Im}\, \varphi_j,
\end{equation*}
which is used to compress the source messages received by $v_j$;

\item a \emph{decoding function} at the decoder $\rho:$
\begin{equation*}\label{dec-fun}
\psi:~ \prod\limits_{j=1}^m \textup{Im}\, \varphi_j \rightarrow (\textup{Im}\,f)^k,
\end{equation*}
which is used to compute $k$ values of $f$ with zero error at $\rho$.
\end{itemize}
Such a $k$-shot source code $\mathbf{C}=\{\varphi_j:1\leq j\leq m;~ \psi\}$ for the model $(s,m,\Omega,f)$ is \emph{admissible} if the~$k$ values of the target function $f$ can be computed with zero error at the decoder $\rho$ for all the source messages $\boldsymbol{x}_i\in\mathcal{A}^k$, $1\leq i\leq s$, namely that
\[
\psi\Big( \varphi_j\big(\boldsymbol{x}_{\Theta(v_j)}\big):~ 1\leq j\leq m \Big)=f(\boldsymbol{x}_1,\boldsymbol{x}_2,\cdots,\boldsymbol{x}_s),\quad \forall~ \boldsymbol{x}_i\in \mathcal{A}^k,~1\leq i\leq s,
\]
where we let $\boldsymbol{x}_{\Theta(v_j)}\triangleq\big(\boldsymbol{x}_i:~ \sigma_i\in\Theta(v_j)\big)$. For such an admissible $k$-shot source code $\mathbf{C}=\big\{\varphi_j:1\leq j\leq m;~ \psi\big\}$, we let
\[
n_j(\mathbf{C})\triangleq \left\lceil \log_{|\mathcal{A}|} |\textup{Im}\, \varphi_j| \right\rceil
\]
for each $1\leq j\leq m$, which is the number of times that the link $(v_j,\rho)$ is used to transmit the encoded message $\varphi_j\big(\boldsymbol{x}_{\Theta(v_j)}\big)$ using the code $\mathbf{C}$. We further let
\[
R_j(\mathbf{C})\triangleq \frac{~n_j(\mathbf{C})~}{~k~},
\]
which is the average number of times that the link $(v_j,\rho)$ is used to compute $f$ once by using the code~$\mathbf{C}$. Then $R_j(\mathbf{C})$ can be regarded as the \emph{compression rate} of the encoder $v_j$. The \emph{coding rate} of the code~$\mathbf{C}$, denoted by $R(\mathbf{C})$, is defined as the maximum (worst) compression rate for all the encoders, i.e.,
\[
R(\mathbf{C})\triangleq\max_{1\leq j\leq m} R_j(\mathbf{C})=\frac{~\max\limits_{1\leq j\leq m} n_j(\mathbf{C})~}{~k~}=\frac{~n(\mathbf{C})~}{~k~},
\]
where we let $n(\mathbf{C})\triangleq\max\limits_{1\leq j\leq m} n_j(\mathbf{C})$. This coding rate $R(\mathbf{C})$ can be regarded as the average ``cost'' of the code $\mathbf{C}$ for computing $f$ once on the model $(s,m,\Omega,f)$.

Next, we say that a nonnegative real number $R$ is \emph{achievable} for the model $(s,m,\Omega,f)$ if $\forall~ \epsilon>0$, there exists an admissible $k$-shot code $\mathbf{C}$ for some positive integer $k$ such that $R(\mathbf{C})<R+\epsilon$. Consequently, the \emph{function-compression capacity} for $(s,m,\Omega,f)$ is defined as
\begin{equation}\label{def-capacity}
\mathcal{C}(s,m,\Omega,f)\triangleq \inf \big\{ R:~ R~ \textup{is achievable for}~ (s,m,\Omega,f) \big\}.
\end{equation}

In this paper, we consider the target function to be a \emph{vector-linear function} over a finite field $\mathbb{F}_q$, where~$q$ is a prime power. More precisely, let
\begin{equation}\label{vec-lin-func}
f(x_1,x_2,\cdots,x_s)=(x_1,x_2,\cdots,x_s)\cdot T,\quad \forall~ x_i\in\mathbb{F}_q,~ 1\leq i\leq s,
\end{equation}
where $T$ is an $\mathbb{F}_q$-valued column-full-rank matrix of size $s\times r$, i.e., $\textup{Rank}(T)=r$ (which implies $r\leq s$); and we assume without loss of generality that $T$ has no all-zero rows. In the rest of the paper, we will write the model $(s,m,\Omega,f)$ as $(s,m,\Omega,T)$ for computing such a vector-linear function $f$ in \eqref{vec-lin-func}.

\vspace{-0.5em}

\subsection{Equivalence to the Model of Network Function Computation}\label{NFC}

In this subsection, we will show that each specified model $(s,m,\Omega,T)$ can be transformed into an equivalent model of \emph{network function computation}. For a specified model $(s,m,\Omega,T)$, we let $\mathcal{G}=(\mathcal{V},\mathcal{E})$ be a directed acyclic graph with the node set $\mathcal{V}=S\cup V\cup \{\rho\}$ and the edge set $\mathcal{E}$ which will be clear later. For each \emph{source node} $\sigma_i$, $1\leq i\leq s$, we set $\ell$ parallel edges from $\sigma_i$ to each $v_j\in\Gamma_{\sigma_i}$, denoted by $d_{i,j}^{(1)},\,d_{i,j}^{(2)},\cdots,d_{i,j}^{(\ell)}$, where
\begin{equation}\label{def-ell}
\ell\triangleq \left\lceil \frac{m}{r} \right\rceil.
\end{equation}
We let $\mathcal{E}(\sigma_i,v_j)\triangleq\big\{d_{i,j}^{(1)},\,d_{i,j}^{(2)},\cdots,d_{i,j}^{(\ell)}\big\}$ for notational simplicity. Accordingly, we let
\[
\mathcal{E}\triangleq \quad\quad \bigcup\limits_{ \mathclap{\qquad\quad~ \textup{all pairs}\, (\sigma_i,v_j)\, \textup{with}\, \sigma_i\rightarrow v_j } } \quad\mathcal{E}(\sigma_i,v_j)~~ \bigcup~~ \{e_1,e_2,\cdots,e_m\},
\]
where we recall that $e_j=(v_j,\rho)$, the edge from the intermediate node $v_j\in V$ to $\rho$, $1\leq j\leq m$.
We further assume that each edge in $\mathcal{E}$ has the unit capacity with respect to the finite field $\mathbb{F}_q$. For an edge $e\in\mathcal{E}$, the \emph{tail} node and \emph{head} node of $e$ are denoted by tail$(e)$ and head$(e)$, respectively. For a node~$u$ in $\mathcal{V}$, we let $\textup{In}(u)=\{ e \in \mathcal{E}:\textup{head}(e)=u \}$ and $\textup{Out}(u)=\{ e \in \mathcal{E}:\textup{tail}(e)=u \}$, the set of \emph{input edges} of $u$ and the set of \emph{output edges} of $u$, respectively. The graph $\mathcal{G}$, together with $S$ and $\rho$, forms a \emph{network}~$\mathcal{N}$, i.e., $\mathcal{N}=(\mathcal{G},S,\rho)$. On the network $\mathcal{N}$, we consider computing with zero error the target function $f(x_S)=x_S\cdot T$ multiple times, where we let $x_S\triangleq(x_1,x_2,\cdots,x_s)$ for notational simplicity. We have specified the network function computation model induced by $(s,m,\Omega,T)$, and denote the model by $(\mathcal{N},T)$.

For a positive integer $k$, a \emph{$k$-shot (function-computing) network code} $\widehat{\mathbf{C}}=\big\{\theta_e:e\in\mathcal{E};~ \phi\big\}$ for $(\mathcal{N},T)$ consists of
\begin{itemize}
\item a \emph{local encoding function} $\theta_e$ for each edge $e\in\mathcal{E}$ such that
\begin{equation}\label{def:loc-encod-fun}
  \theta_{e}:
  \begin{cases}
    \qquad \mathbb{F}_q^k \rightarrow \textup{Im}\,\theta_e, & \textup{if}~\textup{tail}(e)=\sigma_{i}~\textup{for some}~ 1\leq i\leq s, \\
    \prod\limits_{d\in \textup{In}(v_j)} \textup{Im}\,\theta_d \rightarrow
    \textup{Im}\,\theta_e, & \textup{if}~\textup{tail}(e)=v_j~\textup{for some}~ 1\leq j\leq m;
  \end{cases}
\end{equation}

\item a \emph{decoding function} at the \emph{sink node} $\rho$ given by
\[
\phi:~\prod_{j=1}^m \textup{Im}\,\theta_{e_j} \rightarrow \mathbb{F}_q^{k\times r},
\]
which is used to compute the $k$ function values $f(\boldsymbol{x}_S)$ with zero error at $\rho$, where
\[
f(\boldsymbol{x}_S)=\boldsymbol{x}_S\cdot T=(\boldsymbol{x}_1,\boldsymbol{x}_2,\cdots,\boldsymbol{x}_s)\cdot T=
    \begin{bmatrix} (x_{1,1},x_{2,1},\cdots,x_{s,1})\cdot T\\ (x_{1,2},x_{2,2},\cdots,x_{s,2})\cdot T \vspace{-0.5em}\\ \vdots \vspace{-0.7em} \\ (x_{1,k},x_{2,k},\cdots,x_{s,k})\cdot T
    \end{bmatrix}
\]
with $\boldsymbol{x}_S=(\boldsymbol{x}_1,\boldsymbol{x}_2,\cdots,\boldsymbol{x}_s)$ and $\boldsymbol{x}_i=(x_{i,1},x_{i,2},\cdots,x_{i,k})^\top$ being the source message generated by~$\sigma_i$ for $1\leq i\leq s$.
\end{itemize}

With the causality of the encoding mechanism as specified in \eqref{def:loc-encod-fun}, we see that the message transmitted on the edge $e_j,1\leq j\leq m$ is a function of $\boldsymbol{x}_{\Theta(v_j)}$ (where $\boldsymbol{x}_{\Theta(v_j)}=\big(\boldsymbol{x}_i:\,\sigma_i\in\Theta(v_j)\big)$), denoted by $\widehat{\theta}_{e_j}\big(\boldsymbol{x}_{\Theta(v_j)}\big)$. More precisely,
\[
\widehat{\theta}_{e_j}\big(\boldsymbol{x}_{\Theta(v_j)}\big)=\theta_{e_j}\big( \theta_{\mathcal{E}(\sigma_i,v_j)}(\boldsymbol{x}_i):\, \sigma_i\in\Theta(v_j) \big),
\]
where $\theta_{\mathcal{E}(\sigma_i,v_j)}(\boldsymbol{x}_i)\triangleq\big( \theta_{d_{i,j}^{(1)}}(\boldsymbol{x}_i),\, \theta_{d_{i,j}^{(2)}}(\boldsymbol{x}_i),\cdots,\theta_{d_{i,j}^{(\ell)}}(\boldsymbol{x}_i) \big)$, and we call $\widehat{\theta}_{e_j}$ the \emph{global encoding function} for the edge $e_j$.

A $k$-shot network code $\widehat{\mathbf{C}}=\big\{\theta_e:e\in\mathcal{E};~ \phi\big\}$ for the model $(\mathcal{N},T)$ is called \emph{admissible} if the $k$ values $f(\boldsymbol{x}_S)=\boldsymbol{x}_S\cdot T$ can be computed with zero error at the sink node $\rho$ for all possible source messages $\boldsymbol{x}_S=(\boldsymbol{x}_1,\boldsymbol{x}_2,\cdots,\boldsymbol{x}_s)\in\mathbb{F}_q^{k\times s}$, $1\leq i\leq s$, namely that
\[
\phi\left( \widehat{\theta}_{e_j}\big(\boldsymbol{x}_{\Theta(v_j)}\big):~ 1\leq j\leq m \right)=\boldsymbol{x}_S\cdot T,\quad \forall~ \boldsymbol{x}_i\in \mathbb{F}_q^k,~1\leq i\leq s.
\]
For such an admissible $k$-shot network code $\widehat{\mathbf{C}}$, we let
\[
n_e\big(\widehat{\mathbf{C}}\big)\triangleq\Big\lceil\log_{q}|\textup{Im}\,\theta_{e}|\Big\rceil
\]
for each $e\in\mathcal{E}$, which is the number of times that the edge $e$ is used to transmit $\widehat{\theta}_{e_j}\big(\boldsymbol{x}_{\Theta(v_j)}\big)$ by using the code $\widehat{\mathbf{C}}$. The \emph{computing rate} of the code $\widehat{\mathbf{C}}$, denoted by $\mathcal{R}\big(\widehat{\mathbf{C}}\big)$, is defined as the average number of times that the vector-linear function of $T$ can be computed with zero error for one use of the network~$\mathcal{N}$ by using the code $\mathbf{C}$, i.e.,
\[
\mathcal{R}\big(\widehat{\mathbf{C}}\big)\triangleq \frac{~k~}{~\max\limits_{e\in\mathcal{E}}\,n_e\big(\widehat{\mathbf{C}}\big)~}
=\frac{~k~}{~n\big(\widehat{\mathbf{C}}\big)~},
\]
where $n\big(\widehat{\mathbf{C}}\big)\triangleq\max\limits_{e\in\mathcal{E}}\,n_e\big(\widehat{\mathbf{C}}\big)$ can be regarded as the number of times that the network $\mathcal{N}$ is used to compute $f(x_S)=x_S\cdot T$, the target function, $k$ times by using the code $\widehat{\mathbf{C}}$. The \emph{computing capacity} for $(\mathcal{N},T)$ is defined as
\begin{align}\label{def-computing-capacity}
\mathcal{C}(\mathcal{N},T)\triangleq  \sup \Big\{ \mathcal{R}\big(\widehat{\mathbf{C}}\big):~ \textup{all admissible}~ k\textup{-shot network codes}~ \widehat{\mathbf{C}}~ \textup{for}~ (\mathcal{N},T) \Big\}.
\end{align}

For the above models $(s,m,\Omega,T)$ and $(\mathcal{N},T)$, we claim that
\begin{equation}\label{correlation}
\mathcal{C}(s,m,\Omega,T)=\frac{~1~}{~\mathcal{C}(\mathcal{N},T)~}.
\end{equation}
We first show that for the model $(\mathcal{N},T)$, each admissible $k$-shot network code $\widehat{\mathbf{C}}=\big\{\theta_e:e\in\mathcal{E};~ \phi\big\}$ can be transformed to another admissible $k$-shot network code such that each $v_j,1\leq j\leq m$ receives all the source messages $\boldsymbol{x}_i=(x_{i,1},x_{i,2},\cdots,x_{i,k})^\top$ for $\sigma_i\in\Theta(v_j)$, where the latter is also an admissible $k$-shot source code for $(s,m,\Omega,T)$. To see this, for the $k$-shot network code $\widehat{\mathbf{C}}=\big\{\theta_e:e\in\mathcal{E};~ \phi\big\}$, we consider the edge subset $\textup{In}(\rho)=\{e_j:\,1\leq j\leq m\}$, which is a cut set separating $\rho$ from all the source nodes $\sigma_i,1\leq i\leq s$.\footnote{We say that $\textup{In}(\rho)$ is a cut set separating $\rho$ from all the source nodes $\sigma_i,1\leq i\leq s$ if there exists no directed path from each source node $\sigma_i$ for $1\leq i\leq s$ to $\rho$ upon deleting the edges in $\textup{In}(\rho)$ from $\mathcal{E}$.} Thus we have
\begin{equation}\label{code-equiva-geq}
\begin{split}
q^{m\cdot n}\geq \prod_{j=1}^m \big|\textup{Im}\,\theta_{e_j}\big| &\geq \#\left\{ \left( \widehat{\theta}_{e_j}\big(\boldsymbol{x}_{\Theta(v_j)}\big):\,1\leq j\leq m \right):~ \forall~ \boldsymbol{x}_i\in\mathbb{F}_q^k,\,1\leq i\leq s \right\}\geq q^{k\cdot r},
\end{split}
\end{equation}
where we let $n\triangleq n\big(\widehat{\mathbf{C}}\big)$ and $r=\textup{Rank}(T)$ for notational simplicity. Here, the first inequality in \eqref{code-equiva-geq} follows from the fact that $q^n\geq\big|\textup{Im}\,\theta_{e_j}\big|$ for each $1\leq j\leq m$, and the last inequality in \eqref{code-equiva-geq} holds because the code $\widehat{\mathbf{C}}$, which can compute the $k$ values of the vector-linear function, has to distinguish all the images $\boldsymbol{x}_S\cdot T\in\big(\mathbb{F}_q^r\big){^k}=\mathbb{F}_q^{k\times r}$ on the edge subset $\textup{In}(\rho)$. By \eqref{code-equiva-geq}, we immediately obtain that
\begin{equation}\label{n-geq-k/ell}
n\geq \frac{~kr~}{~m~}\geq \frac{~k~}{~\ell~},
\end{equation}
where the second inequality in \eqref{n-geq-k/ell} follows from $\ell=\lceil \frac{m}{r} \rceil$ by \eqref{def-ell}. We have thus obtained that for any admissible $k$-shot network code, at least $k/\ell$ symbols in $\mathbb{F}_q$ can be transmitted on each edge. In particular, we consider $\ell$ parallel edges $d_{i,j}^{(1)},\,d_{i,j}^{(2)},\cdots,d_{i,j}^{(\ell)}$ from $\sigma_i$ to $v_j$ for each pair $(\sigma_i,v_j)$ satisfying $\sigma_i\rightarrow v_j$. Through the $\ell$ edges, at least $k$ symbols in $\mathbb{F}_q$ can be transmitted from $\sigma_i$ to $v_j$. Now, for each pair $(\sigma_i,v_j)$ with $\sigma_i\rightarrow v_j$, we are able to modify $\theta_{d_{i,j}^{(1)}},\,\theta_{d_{i,j}^{(2)}},\cdots,\theta_{d_{i,j}^{(\ell)}}$ such that $\sigma_i$ directly transmits all the $k$ source symbols $\boldsymbol{x}_i$ to $v_j$ through $d_{i,j}^{(1)},\,d_{i,j}^{(2)},\cdots,d_{i,j}^{(\ell)}$, and then modify $\theta_{e_j}$ to $\theta_j\triangleq\theta_{e_j}\circ\big( \theta_{\mathcal{E}(\sigma_i,v_j)}:\,\sigma_i\in\Theta(v_j) \big)$ for $1\leq j\leq m$, namely that
\[
\theta_j\big( \boldsymbol{x}_{\Theta(v_j)} \big)=\theta_{e_j}\big( \theta_{\mathcal{E}(\sigma_i,v_j)}(\boldsymbol{x}_i):\, \sigma_i\in\Theta(v_j) \big),~ 1\leq j\leq m.
\]
We can readily see that the above modified code for $(\mathcal{N},T)$ is still admissible and has the same computing rate $k/n$ as the original code $\widehat{\mathbf{C}}$, and thus it is an admissible $k$-shot source node for the model $(s,m,\Omega,T)$ with the coding rate $n/k$.
On the other hand, it is easy to see that an admissible $k$-shot source code for $(s,m,\Omega,T)$ with the coding rate $n/k$ is also an admissible $k$-shot network code for $(\mathcal{N},T)$ with the computing rate $k/n$. Hence, by the definitions of \eqref{def-capacity} and \eqref{def-computing-capacity}, we immediately prove \eqref{correlation}.

\section{Preparatory Results}\label{pre-result}

\subsection{Capacity Characterization for $\textup{Rank}(T)=1$ and $\textup{Rank}(T)=s$}\label{two-simple cases}
In this subsection, we characterize the capacities for two simple cases of the model $(s,m,\Omega,T)$ with $\textup{Rank}(T)=1$ and $\textup{Rank}(T)=s$. For the case of $\textup{Rank}(T)=1$, we need to compute a scalar-linear function over a finite field $\mathbb{F}_q$:
\begin{equation}\label{scalar}
f(x_S)=\sum_{i=1}^s a_i\cdot x_i,\quad  a_i\in\mathbb{F}_q \text{ for all } 1\leq i\leq s,
\end{equation}
where we let $x_S=(x_1,x_2,\cdots,x_s)$. Without loss of generality, it suffices to consider the scalar-linear function $f(x_S)$ in \eqref{scalar} with $a_i\neq 0$ for all $1\leq i\leq s$. Further, we let $y_i=a_i\cdot x_i$ for all $1\leq i\leq s$, and it is equivalent to considering the algebraic sum $g(y_1,y_2,\cdots,y_s)=\sum\limits_{i=1}^s y_i$ over $\mathbb{F}_q$. Thus, for the case of $\textup{Rank}(T)=1$, it is sufficient to consider the algebraic sum
\begin{equation*}\label{alg}
f(x_S)=\sum_{i=1}^s x_i=x_S\cdot T_{\textup{sum}},
\end{equation*}
where $T_{\textup{sum}}$ is an all-one column $s$-vector. For the case of $\textup{Rank}(T)=s$ (namely that $T$ is an $s\times s$ invertible matrix over $\mathbb{F}_q$), we readily see that computing $x_S\cdot T$ is equivalent to decoding the original source messages $x_S$, or equivalently, computing the identity function $x_S\cdot T_{\textup{id}}$, where $T_{\textup{id}}$ is an $s\times s$ identity matrix. This case is in fact a model of multi-source single-sink network coding.

In order to characterize the capacities for the above two cases $(s,m,\Omega,T_{\textup{sum}})$ and $(s,m,\Omega,T_{\textup{id}})$, by~\eqref{correlation} it suffices to characterize the computing capacities for the associated models $(\mathcal{N},T_{\textup{sum}})$ and $(\mathcal{N},T_{\textup{id}})$ of network function computation, and the computing capacities can be characterized by the existing results in the literature (cf.\,\cite{Appuswamy11,Huang_Comment_cut_set_bound,Guang_Improved_upper_bound}). Before specifying the capacities, we first present some graph-theoretic notations as follows. For two nodes $u$ and $v$ in $\mathcal{V}$, if there exists no directed path from $u$ to $v$, we say that $v$ is \emph{separated} from $u$. Given a set of edges $C\subseteq\mathcal{E}$, we define two subsets of the source nodes
\begin{align*}
K_C&=\big\{ \sigma\in S:~ \exists~ e\in C~ \textup{s.t. there exists a directed path from}~ \sigma~ \textup{to}~ \textup{tail}(e) \big\}, \notag \\
I_C&=\big\{ \sigma\in S:~ \rho~ \textup{is separated from}~ \sigma~ \textup{upon deleting the edges in}~ C~ \textup{from}~ \mathcal{E}  \big\}. 
\end{align*}
Further, an edge set $C$ is said to be a \emph{cut set} if $I_C\neq\emptyset$, and we let $\Lambda(\mathcal{N})$ be the family of all the cut sets in the network $\mathcal{N}$, i.e.,
$\Lambda(\mathcal{N})=\big\{C\subseteq\mathcal{E}:~ I_C\neq\emptyset\big\}$.

We now present the computing capacities $\mathcal{C}(\mathcal{N},T_{\textup{sum}})$ and $\mathcal{C}(\mathcal{N},T_{\textup{id}})$ in the following lemma (cf. the part of Special Target Functions in \cite[Section III.B]{Huang_Comment_cut_set_bound} and \cite[Theorems~III.1 and III.2]{Appuswamy11}).

\begin{lemma}\label{cap-alg-id}
The computing capacities for $(\mathcal{N}, T_{\textup{sum}})$ and $(\mathcal{N},T_{\textup{id}})$ are given by
\[
\mathcal{C}(\mathcal{N}, T_{\textup{sum}})=\min\limits_{C\in\Lambda(\mathcal{N})} |C| \quad \textup{and}~ \quad \mathcal{C}(\mathcal{N}, T_{\textup{id}})=\min\limits_{C\in\Lambda(\mathcal{N})} \frac{~|C|~}{~|I_C|~}.
\]
\end{lemma}

Together with \eqref{correlation}, we specify the capacities for the models $(s,m,\Omega,T_{\textup{sum}})$ and $(s,m,\Omega,T_{\textup{id}})$ below.
\begin{theorem}\label{cap-r=1-r=s}
Consider the models $(s,m,\Omega,T_{\textup{sum}})$ and $(s,m,\Omega,T_{\textup{id}})$. Then,
\[
\mathcal{C}(s,m,\Omega,T_{\textup{sum}})=\max_{\sigma_i\in S} \frac{~1~}{~|\Gamma_{\sigma_i}|~} \quad \textup{and}~ \quad \mathcal{C}(s,m,\Omega,T_{\textup{id}})=\max\limits_{\Gamma\subseteq V} \frac{~|I_{\Gamma}|~}{~\big|\Gamma\big|~}, \footnotemark
\]
\footnotetext{We assume $\Gamma\neq\emptyset$ when it appears in the denominator throughout the paper.}
where for a subset of encoders $\Gamma\subseteq V$, let
\begin{equation}\label{id-def-I-Gamma}
I_{\Gamma}\triangleq\big\{ \sigma\in S: \rho~ \text{is separated from}~ \sigma~ \text{upon deleting the edges}~ (v,\rho)~ \textup{for all}~ v\in\Gamma  \big\}. \footnotemark
\end{equation}
\end{theorem}
\footnotetext{We let $|I_{\Gamma}|=0$ if $I_{\Gamma}=\emptyset$.}
\begin{proof}
We first consider the model $(s,m,\Omega,T_{\textup{sum}})$. By Lemma \ref{cap-alg-id}, we have
\begin{equation}\label{cap-alg-simplify}
\mathcal{C}(\mathcal{N}, T_{\textup{sum}})=\min\limits_{C\in\Lambda(\mathcal{N})} |C|=\min\limits_{\sigma_i\in S} \textup{mincut}(\sigma_i, \rho),
\end{equation}
where $\textup{mincut}(\sigma_i, \rho)$ stands for the minimum cut capacity separating $\rho$ from $\sigma_i$ in $\mathcal{N}$. For the last equality in \eqref{cap-alg-simplify}, since each cut set $C\in\Lambda(\mathcal{N})$ separates $\rho$ from at least one source node $\sigma_i$, we immediately have $\min_{C\in\Lambda(\mathcal{N})} |C|\geq \min_{\sigma_i\in S}\textup{mincut}(\sigma_i, \rho)$; and on the other hand, each minimum cut separating $\rho$ from a source node $\sigma_i$ is a cut set in $\Lambda(\mathcal{N})$, and thus $\min_{C\in\Lambda(\mathcal{N})} |C|\leq \min_{\sigma_i\in S}\textup{mincut}(\sigma_i, \rho)$. Furthermore, we have
\begin{equation*}\label{alg-C-3}
\textup{mincut}(\sigma_i, \rho)=\big| \big\{e_j:\,v_j\in\Gamma_{\sigma_i}\big\} \big|=\big|\Gamma_{\sigma_i}\big|,
\end{equation*}
because the edge subset $\big\{e_j:~v_j\in\Gamma_{\sigma_i}\big\}$ (cf.~\eqref{Gamma-sigma-i} for $\Gamma_{\sigma_i}$) is a minimum cut separating $\rho$ from $\sigma_i$. Combining \eqref{cap-alg-simplify} and \eqref{correlation}, we obtain that
\[
\mathcal{C}(s,m,\Omega,T_{\textup{sum}})=\frac{~1~}{~\mathcal{C}(\mathcal{N}, T_{\textup{sum}})~}=\frac{~1~}{~\min\limits_{\sigma_i\in S} |\Gamma_{\sigma_i}|~}=\max\limits_{\sigma_i\in S} \frac{~1~}{~|\Gamma_{\sigma_i}|~}.
\]

Next, we consider the model $(s,m,\Omega,T_{\textup{id}})$. It follows from \eqref{correlation} and Lemma \ref{cap-alg-id} that
\begin{align}\label{s-m-Omega-T-id-C}
\mathcal{C}(s,m,\Omega,T_{\textup{id}})&=\frac{~1~}{~\mathcal{C}(\mathcal{N}, T_{\textup{id}})~}=\max_{C\in\Lambda(\mathcal{N})} \frac{~|I_C|~}{~|C|~}.
\end{align}
We first claim that
\begin{align}\label{id-max-I-C-C-eq-max-In-rho-I-C-C}
\max\limits_{C\in\Lambda(\mathcal{N})} \frac{~|I_C|~}{~|C|~}
=\max\limits_{C\in\Lambda(\mathcal{N})\,\textup{s.t.}\,C\subseteq\textup{In}(\rho)} \frac{~|I_C|~}{~|C|~}.
\end{align}
It is easy to see that
\begin{align}\label{id-max-I-C-C-geq-max-In-rho-I-C-C}
\max\limits_{C\in\Lambda(\mathcal{N})} \frac{~|I_C|~}{~|C|~}\geq
\max\limits_{C\in\Lambda(\mathcal{N})\,\textup{s.t.}\,C\subseteq\textup{In}(\rho)} \frac{~|I_C|~}{~|C|~}.
\end{align}
In order to prove this claim, it suffices to prove the other direction. Consider an arbitrary cut set $C\in\Lambda(\mathcal{N})$ and take the following operations to $C$. For each edge $e\in C$, if $\textup{head}(e)\in V$, say $v$, then replace $e$ by the edge from~$v$ to $\rho$, i.e., the edge $(v,\rho)$; and otherwise, keep $e$ unchanged. We denote by $C'$ the new edge subset thus obtained. We can readily see that $C'\subseteq\textup{In}(\rho)$, $|C'|\leq|C|$ and $C'$ is a cut separating all the edges in $C$ to $\rho$, namely that no path exists from each edge in $C$ to $\rho$ upon deleting all the edges in~$C'$. This implies that $C'$ is also a cut set in $\Lambda(\mathcal{N})$ and $I_{C'}\supseteq  I_{C}$, and hence $|I_{C'}|\geq|I_{C}|$. With this, we immediately have
\begin{align}\label{id-I-C-C-leq-max-In-rho-I-C-C}
\frac{~|I_C|~}{~|C|~}\leq\frac{~|I_{C'}|~}{~|C'|~}\leq
\max\limits_{C\in\Lambda(\mathcal{N})\,\textup{s.t.}\,C\subseteq\textup{In}(\rho)} \frac{~|I_C|~}{~|C|~}.
\end{align}
Note that the inequality \eqref{id-I-C-C-leq-max-In-rho-I-C-C} is true for each cut set $C\in\Lambda(\mathcal{N})$. Thus, we have proved that
\begin{align}\label{id-max-I-C-C-leq-max-In-rho-I-C-C}
\max\limits_{C\in\Lambda(\mathcal{N})} \frac{~|I_C|~}{~|C|~}\leq
\max\limits_{C\in\Lambda(\mathcal{N})\,\textup{s.t.}\,C\subseteq\textup{In}(\rho)} \frac{~|I_C|~}{~|C|~},
\end{align}
Combining \eqref{id-max-I-C-C-leq-max-In-rho-I-C-C} and \eqref{id-max-I-C-C-geq-max-In-rho-I-C-C}, the claim \eqref{id-max-I-C-C-eq-max-In-rho-I-C-C} is proved.

We further consider
\begin{align}\label{id-max-In-rho-I-C-C-eq}
\max\limits_{C\in\Lambda(\mathcal{N})\,\textup{s.t.}\,C\subseteq\textup{In}(\rho)} \frac{~|I_C|~}{~|C|~}=\max\limits_{C\in\Lambda(\mathcal{N})\,\textup{s.t.}\,C\subseteq\textup{In}(\rho)} \frac{~|I_C|~}{~\big|\{\textup{tail}(e):e\in C\}\big|~}=\max\limits_{\Gamma\subseteq V} \frac{~|I_{\Gamma}|~}{~\big|\Gamma\big|~},
\end{align}
where by the definition of $I_{\Gamma}$ in \eqref{id-def-I-Gamma}, we have $I_{\Gamma}=I_C$ if $C=\big\{(v,\rho):v\in\Gamma\big\}$. Combining \eqref{s-m-Omega-T-id-C}, \eqref{id-max-I-C-C-eq-max-In-rho-I-C-C} and \eqref{id-max-In-rho-I-C-C-eq}, we have proved that
\[
\mathcal{C}(s,m,\Omega,T_{\textup{id}})=\max_{C\in\Lambda(\mathcal{N})} \frac{~|I_C|~}{~|C|~}=\max\limits_{C\in\Lambda(\mathcal{N})\,\textup{s.t.}\,C\subseteq\textup{In}(\rho)} \frac{~|I_C|~}{~|C|~}=\max\limits_{\Gamma\subseteq V} \frac{~|I_{\Gamma}|~}{~\big|\Gamma\big|~}.
\]
The theorem is proved.
\end{proof}

\vspace{-0.5em}

\subsection{The Best Known Lower Bound on the Capacity for the Nontrivial Model $(s,m,\Omega,T)$}\label{best-known-lower-bound}

In the rest of the paper, we will consider the nontrivial cases of the model $(s,m,\Omega,T)$ with $1\!\!<\!\!\Rank(T)\!\!<\!\!s$. Generally speaking, the characterization of the function-compression capacity for a nontrivial case is difficult. In fact, its dual problem to characterize the computing capacity for the corresponding model $(\mathcal{N},T)$ of network function computation is difficult, e.g., \cite{Guang_Improved_upper_bound,Appuswamy-lin-func-lin-code,Li_Xu_vector_linear_diamond}. Appuswamy and Franceschetti~\cite{Appuswamy-lin-func-lin-code} investigated the solvability (rate-$1$ achievability) of linear (function-computing) network codes when the single sink node is required to compute a vector-linear function of the source messages over a network, where the used technique is rather complicated and relies on the use of some advanced algebraic tools. Consequently, Guang \textit{et al.} \cite{Guang_Improved_upper_bound} enhanced their results by applying an improved upper bound on the computing capacity which is obtained by using a novel cut-set strong partition approach. For the computing capacity characterization problem of $(\mathcal{N},T)$ considered here, only the computing capacity for a vector-linear function over the diamond network has been completely characterized \cite{Li_Xu_vector_linear_diamond}. So in the paper we focus on the simplest nontrivial cases of the model $(s,m,\Omega,T)$ with three sources (i.e., $s=3$) so that $\textup{Rank}(T)=2$, the number of encoders $m\leq3$ and arbitrary connectivity states $\Omega$. We will see in the rest of the paper that the capacity characterization even for the simplest cases are nontrivial.

We recall the equivalence of the model $\big(3,m,\Omega,T\big)$ and the corresponding model of network function computation as discussed in Section \ref{NFC}. By~\eqref{correlation} we can lower bound the function-compression capacities for all the models by applying upper bounds on the computing capacity in network function computation. In network function computation, several upper bounds on the computing capacity have been obtained \cite{Appuswamy11,Huang_Comment_cut_set_bound,Guang_Improved_upper_bound}, which are applicable to arbitrary network and arbitrary target function. Here, the best known upper bound is the one proved by  Guang~\emph{et al.}~\cite{Guang_Improved_upper_bound} in using the approach of the cut-set strong partition. With this, we can obtain a lower bound on the function-compression capacity for each model $(3,m,\Omega,T)$. To be specific, the model $(3,m,\Omega,T)$ can be transformed to a model of network function computation, which we denote by $(\mathcal{N},T)$. Before specifying the best known upper bound proved by Guang~\emph{et~al.}~\cite{Guang_Improved_upper_bound} on the computing capacity for $\big(\mathcal{N},T\big)$, we present the definition of a strong partition of a cut set.

\begin{definition}[\!\!{\cite[Definition~2]{Guang_Improved_upper_bound} and \cite[Definition~3]{Guang_Zhang_Arithmetic_sum_TIT}}]\label{def:strong_partition}
	Let $C\in\Lambda(\mathcal{N})$ be a cut set and $\mathcal{P}_{C}=\{C_{1}, C_{2},\cdots,C_{t}\}$
	be a partition of the cut set $C$. The partition $\mathcal{P}_{C}$ is said to be a strong
	partition of $C$ if the following two conditions are satisfied:
	\begin{enumerate}
		\item $I_{C_{i}}\neq\emptyset$,\quad $\forall~ 1\leq i\leq t;$
		\item $I_{C_{i}}\cap K_{C_{j}}=\emptyset$,\quad$\forall~ 1\leq i,j\leq t $ and $i\neq j$.\footnote{There is a typo in the original definition of strong partition {\cite[Definition~2]{Guang_Improved_upper_bound}}, where in 2), ``$I_{C_i}\cap I_{C_j}=\emptyset$'' in~{\cite[Definition~2]{Guang_Improved_upper_bound}} should be ``$I_{C_{i}}\cap K_{C_{j}}=\emptyset$'' as stated in \cite[Definition~3]{Guang_Zhang_Arithmetic_sum_TIT}.}
	\end{enumerate}
\end{definition}
We note that a cut set $C\in\Lambda(\mathcal{N})$ is a trivial strong partition of itself. We now specify the upper bound in \cite{Guang_Improved_upper_bound} on the computing capacity for $\big(\mathcal{N},T\big)$ as follows:
\begin{align*}
\mathcal{C}(\mathcal{N},T)\leq \min_{ C\in\Lambda(\mathcal{N})}\, \min_{ \substack{\textup{all strong partitions}\\ \mathcal{P}_C\,\textup{of}\,C} } \frac{|C|}{\textbf{\textup{rank}}_{\mathcal{P}_C}(T)},
\end{align*}
where for a strong partition $\mathcal{P}_{C}\triangleq\{C_{1}, C_{2},\cdots,C_{t}\}$ of $C$, we define
\begin{align*}
\textbf{\textup{rank}}_{\mathcal{P}_C}(T)\triangleq \sum_{i=1}^t \textup{Rank}\big(T[I_{C_i}]\big)+\textup{Rank}\big(T[I_{C}]\big)-\textup{Rank}\big(T[\cup_{i=1}^t I_{C_i}]\big)
\end{align*}
with $T[I]$ for a source subset $I\subseteq S$ representing the submatrix of $T$ containing the $i$th row if $\sigma_i\in I$.
Together with \eqref{correlation}, we obtain that
\begin{align}\label{3-m-Omega-T-geq-max-best-known-upper-bound}
\mathcal{C}\big(3,m,\Omega,T\big)=\frac{~1~}{~\mathcal{C}\big(\mathcal{N}, T\big)~}  \geq \max\limits_{C\in\Lambda(\mathcal{N})}\, \max_{ \substack{\textup{all strong partitions}\\ \mathcal{P}_C\,\textup{of}\,C} }  \frac{~\textbf{\textup{rank}}_{\mathcal{P}_C}(T)~}{~|C|~}. 
\end{align}
We remark that in fact, the above lower bound holds for all possible models $(s,m,\Omega,T)$ with arbitrary number of sources $s$, arbitrary number of encoders $m$, arbitrary connectivity states $\Omega$ and arbitrary matrices $T$ with $1<\Rank(T)<s$.

\vspace{-0.5em}

\subsection{Classification of Target Functions for the Nontrivial Model $(3,m,\Omega,T)$ with $\textup{Rank}(T)=2$}\label{nontrivial-cases}

In this subsection, we prove the following theorem which implies that all the $\mathbb{F}_q$-valued $3\times2$ column-full-rank matrices $T$ can be divided into two types, and the capacities for the model $(3,m,\Omega,T)$ are identical if the matrices $T$ have the same type.

\begin{theorem}\label{Type-I-II}
Let $T$ be an $\mathbb{F}_q$-valued $3\times2$ column-full-rank matrix.\footnote{All matrices $T$ in the rest of the paper are assumed to have no all-zero row.} Then
\begin{equation*}
\mathcal{C}(3,m,\Omega,T)=
  \begin{cases}
 \mathcal{C}(3,m,\Omega,T_1),  & \text{if any two rows in $T$ are linearly independent;} \bigskip\\
 \mathcal{C}(3,m,\Omega,T_2),  & \parbox{0.45\columnwidth}{ \text{if there exist two rows in $T$ to be linearly dependent,} \raggedright \text{assuming WLOG the first two rows}; }
  \end{cases}
\end{equation*}
where
\begin{equation}\label{T1-T2}
    T_1=\setlength{\arraycolsep}{3pt}
    \renewcommand{\arraystretch}{0.7}
    \begin{bmatrix} 1&0\\ 0&1\\ 1&1 \end{bmatrix}\quad \textup{and} \quad
     T_2=\setlength{\arraycolsep}{3pt}
    \renewcommand{\arraystretch}{0.7}
    \begin{bmatrix} 1&0\\ 1&0\\ 0&1 \end{bmatrix}.
    \end{equation}
\end{theorem}

To prove Theorem \ref{Type-I-II}, we first present the lemma below.

\begin{lemma}\label{cap-eq-Q}
Consider an $\mathbb{F}_q$-valued $3\times2$ column-full-rank matrix $T$ and an $\mathbb{F}_q$-valued $2\times2$ invertible matrix $Q$. Then
\[
\mathcal{C}(3,m,\Omega,T)=\mathcal{C}(3,m,\Omega,TQ).
\]
\end{lemma}
\begin{IEEEproof}
An admissible $k$-shot source code for $(3,m,\Omega,T)$ is also admissible for $(3,m,\Omega,TQ)$, and vice versa, because the decoder $\rho$ can compute with zero error the $k$ function values $\boldsymbol{x}_S\cdot T$ if and only if the $k$ values $\boldsymbol{x}_S\cdot TQ$. This immediately implies that $\mathcal{C}(3,m,\Omega,T)=\mathcal{C}(3,m,\Omega,TQ)$.
\end{IEEEproof}

\begin{proof}[Proof of Theorem \ref{Type-I-II}]
We first consider the matrix $T$ in which any two rows of $T$ are linearly independent. It is not difficult to see that there exists an $\mathbb{F}_q$-valued $2\times 2$ invertible matrix $Q$ such that $T'\triangleq TQ$ of form $\left[\begin{smallmatrix} a&0\\ 0&b\\ c&c \end{smallmatrix}\right]$, where $a,\,b$ and $c$ are three nonzero elements in $\mathbb{F}_q$. Together with Lemma \ref{cap-eq-Q}, we have
\begin{equation}\label{cap-eq-Q-TypeI}
\mathcal{C}(3,m,\Omega,T)=\mathcal{C}\big(3,m,\Omega,T'\big).
\end{equation}

Next, we let $y_1=a x_1$, $y_2=b x_2$ and $y_3=c x_3$. Then
\[
(x_1,x_2,x_3)\cdot T'=(ax_1+bx_2,\,bx_2+cx_3)=(y_1+y_2,\,y_2+y_3)=(y_1,y_2,y_3)\cdot T_1.
\]
Since $a$, $b$ and $c$ are nonzero, an admissible $k$-shot source code for $\big(3,m,\Omega,T'\big)$ can be readily modified to an admissible $k$-shot source code for $(3,m,\Omega,T_1)$ by setting $y_1=a x_1$, $y_2=b x_2$, $y_3=c x_3$, and vice versa. Thus, we have
\begin{equation*}\label{cap-eq-T'-T1}
\mathcal{C}\big(3,m,\Omega,T'\big)=\mathcal{C}(3,m,\Omega,T_1).
\end{equation*}
Together with \eqref{cap-eq-Q-TypeI}, we have proved that $\mathcal{C}(3,m,\Omega,T)=\mathcal{C}(3,m,\Omega,T_1)$.

Next, we consider the other case that $T$ is an $\mathbb{F}_q$-valued $3\times 2$ column-full-rank matrix such that two rows of $T$ are linearly dependent, say, the first and second rows. We see that there exists an $\mathbb{F}_q$-valued $2\times 2$ invertible matrix $M$ such that $T''\triangleq TM$ of the form $\left[\begin{smallmatrix} a&0\\ b&0\\ 0&c \end{smallmatrix}\right]$, where $a$, $b$ and $c$ are also nonzero in~$\mathbb{F}_q$. By Lemma \ref{cap-eq-Q}, we have $\mathcal{C}(3,m,\Omega,T)=\mathcal{C}\big(3,m,\Omega,T''\big)$.
Using the same argument as discussed below~\eqref{cap-eq-Q-TypeI}, by setting $y_1=ax_1$, $y_2=b x_2$, $y_3=c x_3$, we can obtain that $\mathcal{C}\big(3,m,\Omega,T''\big)=\mathcal{C}(3,m,\Omega,T_2)$, and thus $\mathcal{C}(3,m,\Omega,T)=\mathcal{C}(3,m,\Omega,T_2)$.
\end{proof}

\vspace{-0.5em}

\subsection{Model Isomorphism}

We consider two connectivity states $\Omega$ and $\Omega'$ with the set of sources $S=\{\sigma_1,\sigma_2,\cdots,\sigma_s\}$ and the set of encoders $V=\{v_1,v_2,\cdots,v_m\}$. We say that $\Omega$ and $\Omega'$ are \emph{(topologically) isomorphic} if there exists a pair $(\pi,\tau)$ of a permutation $\pi$ on $[s]$ for $S$ and a permutation $\tau$ on $[m]$ for $V$ such that\footnote{Here, for a positive integer $z$, we use $[z]$ to denote $\{1,2,\cdots,z\}$ for notational simplicity.}
\[
\sigma_i\rightarrow v_j~ \textup{in}~ \Omega~ \textup{if and only if}~ \sigma_{\pi(i)}\rightarrow v_{\tau(j)}~ \textup{in}~ \Omega', \quad \forall~ i\in[s]~ \textup{and}~ j\in[m].
\]
We further write $\Omega'=\Omega\circ(\pi,\tau)$ and hence we have $\Omega=\Omega'\circ(\pi^{-1},\tau^{-1})$.

Next, we consider an arbitrary target function $f(x_S)=f(x_1,x_2,\cdots,x_s)$ from $\mathcal{A}^s$ to $\textup{Im}\,f$. With the permutation $\pi$ on $[s]$, we write
\[
f\circ\pi\,(x_S)=f\big(\pi(x_S)\big)=f\big(x_{\pi(1)},x_{\pi(2)},\cdots,x_{\pi(s)}\big).
\]
Here, we use $\pi(x_S)$ ($=\pi(x_1,x_2,\cdots,x_s)$) to represent $\big(x_{\pi(1)},x_{\pi(2)},\cdots,x_{\pi(s)}\big)$, and in the sequel use $\pi(S)$ to represent $\big\{\sigma_{\pi(1)},\sigma_{\pi(2)},\cdots,\sigma_{\pi(s)}\big\}$. This abuse of notation should cause no ambiguity and would greatly simplify the notation.

With the above discussion, we say two models $(s,m,\Omega,f)$ and $(s,m,\Omega',f')$ are \emph{isomorphic} if there exists such a pair $(\pi,\tau)$ of a permutation $\pi$ on $[s]$ for $S$ and a permutation $\tau$ on $[m]$ for $V$ such that $\Omega'=\Omega\circ(\pi,\tau)$ and $f'=f\circ\pi$. Now, we can present the following lemma which shows that the capacities of two isomorphic models are identical.

\begin{lemma}\label{lemma-cap=permutation}
Consider two isomorphic models $(s,m,\Omega,f)$ and $\big(s,\,m,\,\Omega\circ(\pi,\tau),\,f\circ\pi\big)$. Then
\begin{equation*}\label{eq-cap=permutation}
\mathcal{C}(s,m,\Omega,f)=\mathcal{C}\big(s,\,m,\,\Omega\circ(\pi,\tau),\,f\circ\pi\big).
\end{equation*}
\end{lemma}
\begin{proof}
Consider an arbitrary admissible $k$-shot source code $\mathbf{C}=\{\varphi_j:1\leq j\leq m;~ \psi\}$ for $(s,m,\Omega,f)$. Then for each $\boldsymbol{x}_S=(\boldsymbol{x}_1,\boldsymbol{x}_2,\cdots,\boldsymbol{x}_s)\in(\mathcal{A}^k)^s$, we have
\begin{align}\label{psi(varphi_j)}
f(\boldsymbol{x}_S)=\psi\big(\varphi_j(\boldsymbol{x}_{\Theta(v_j)}):~ 1\leq j\leq m\big)
=\psi\big(\varphi_j(\boldsymbol{x}_i:\sigma_i\rightarrow v_j~\textup{in}~\Omega):~ 1\leq j\leq m\big).
\end{align}

Next, we consider the other model $\big(s,\,m,\,\Omega\circ(\pi,\tau),\,f\circ\pi\big)$, in which we note that for each $1\leq j\leq m$, the encoder $v_{\tau(j)}$ receives $\boldsymbol{x}_{\pi(i)}$ if and only if $\sigma_i\rightarrow v_j~\textup{in}~\Omega$. We regard $\varphi_j$ as the encoding function for the encoder $v_{\tau(j)}$, denoted by $\phi_{\tau(j)}$, i.e., $\phi_{\tau(j)}=\varphi_j$; and still regard $\psi$ as the decoding function at $\rho$. Then for each $\boldsymbol{x}_S\in(\mathcal{A}^k)^s$, we have
\begin{align}
&\psi\big(\phi_{\tau(j)}\big(\boldsymbol{x}_{\pi(i)}:~ \sigma_{\pi(i)}\rightarrow v_{\tau(j)}~\textup{in}~ \Omega\circ(\pi,\tau) \big):~ 1\leq j\leq m\big) \notag \\
=&\psi\big(\varphi_j(\boldsymbol{x}_{\pi(i)}:~ \sigma_i\rightarrow v_j~\textup{in}~\Omega ):~ 1\leq j\leq m\big) \label{psi-varphi-in-Omega}\\
=&f\big(\boldsymbol{x}_{\pi(1)},\boldsymbol{x}_{\pi(2)},\cdots,\boldsymbol{x}_{\pi(s)}\big) \label{f-x-pi-1-pi-s}\\
=&f\circ\pi\,(\boldsymbol{x}_S), \notag
\end{align}
where the equality \eqref{psi-varphi-in-Omega} follows because $\Omega$ and $\Omega\circ(\pi,\tau)$ are isomorphic associated with the pair of permutations $(\pi,\tau)$, and the equality \eqref{f-x-pi-1-pi-s} follows from \eqref{psi(varphi_j)}. This thus implies that $\mathbf{C}'\triangleq\big\{\phi_{\tau(j)}:1\leq j\leq m;~ \psi\big\}$ is an admissible $k$-shot source code for $\big(s,\,m,\,\Omega\circ(\pi,\tau),\,f\circ\pi\big)$ preserving the same coding rate as the code $\mathbf{C}$ for $\big(s,m,\Omega,f\big)$.

On the other hand, an admissible $k$-shot source code for $\big(s,\,m,\,\Omega\circ(\pi,\tau),\,f\circ\pi\big)$ can be transformed to an admissible $k$-shot source code for $(s,m,\Omega,f)$ with the same coding rate in the same way by performing the pair of permutations $(\pi^{-1},\tau^{-1})$. We have thus proved the lemma.
\end{proof}

\medskip

For the considered model $(3,m,\Omega,T)$ with Rank$(T)=2$, the following corollary is a straightforward consequence of Lemma \ref{lemma-cap=permutation}.

\begin{cor}\label{cor-T-S=3-cap=permutation}
Consider two isomorphic models $(3,m,\Omega,T)$ and $\big(3,\,m,\,\Omega\circ(\pi,\tau),\,T\circ\pi\big)$. Then,
\begin{equation*}\label{T-eq-cap=permutation}
\mathcal{C}(3,m,\Omega,T)=\mathcal{C}\big(3,\,m,\,\Omega\circ(\pi,\tau),\,T\circ\pi\big),
\end{equation*}
where it is written as $x_S\cdot T\circ\pi=\pi(x_S)\cdot T$, or more precisely,
\[
(x_1,x_2,x_3)\cdot T\circ\pi=\big(x_{\pi(1)},x_{\pi(2)},x_{\pi(3)}\big)\cdot T. \footnotemark
\]
\end{cor}
\footnotetext{Also, $T\circ\pi$ can be regarded as the row permutation of $T$ by using $\pi$, e.g., $T\circ\pi=\setlength{\arraycolsep}{3pt}\renewcommand{\arraystretch}{0.7}
\begin{bmatrix} \textup{row}_2\\ \textup{row}_3\\ \textup{row}_1 \end{bmatrix}$ if $T=\setlength{\arraycolsep}{3pt}\renewcommand{\arraystretch}{0.7}
\begin{bmatrix} \textup{row}_1\\ \textup{row}_2\\ \textup{row}_3 \end{bmatrix}$ and $\pi=\setlength{\arraycolsep}{3pt}\renewcommand{\arraystretch}{0.7}
\begin{pmatrix} 1&2&3\\ 3&1&2\end{pmatrix}$, where we use $\textup{row}_i$ to denote the $i$th row vector of $T$ for $1\leq i\leq3$.}

The following example is given to illustrate the isomorphism of two models.

\begin{example}
Consider the model $(3,3,\Omega,T_1)$ as depicted in Fig.\,\ref{Omega}, and the two permutations $\pi$ and~$\tau$ for $S$ and $V$, respectively, as follows:
\[
\pi=\setlength{\arraycolsep}{3pt}\renewcommand{\arraystretch}{0.7}
\begin{pmatrix} 1&2&3\\ 3&1&2\end{pmatrix}
\quad\textup{and}\quad
\tau=\setlength{\arraycolsep}{3pt}\renewcommand{\arraystretch}{0.7}
\begin{pmatrix} 1&2&3\\ 2&3&1\end{pmatrix}.
\]
With the pair of the permutations $(\pi,\tau)$, we can readily depict the isomorphic connectivity state $\Omega\circ(\pi,\tau)$ (see Fig.\,\ref{Omega'}). Furthermore, for $T_1=
\left[\begin{smallmatrix} 1&0\\ 0&1\\ 1&1 \end{smallmatrix}\right]$ (cf.\,\eqref{T1-T2}), we have $T_1\circ\pi=\left[\begin{smallmatrix} 0&1\\ 1&1\\ 1&0 \end{smallmatrix}\right]$, or equivalently,
\[
(x_1,x_2,x_3)\cdot T_1\circ\pi=\big(x_{\pi(1)},x_{\pi(2)},x_{\pi(3)}\big)\cdot T_1=(x_3,x_1,x_2)\cdot T_1=(x_3+x_2,x_1+x_2).
\]
Comparing the two isomorphic models $(3,3,\Omega,T_1)$ and $\big(3,3,\Omega\circ(\pi,\tau),T_1\circ\pi\big)$ depicted in Fig.\,\ref{Omega} and Fig.\,\ref{Omega'}, respectively, we can readily see that the two models are the same in essence, and clearly $\mathcal{C}(3,3,\Omega,T_1)=\mathcal{C}\big(3,\,3,\,\Omega\circ(\pi,\tau),\,T_1\circ\pi\big)$.

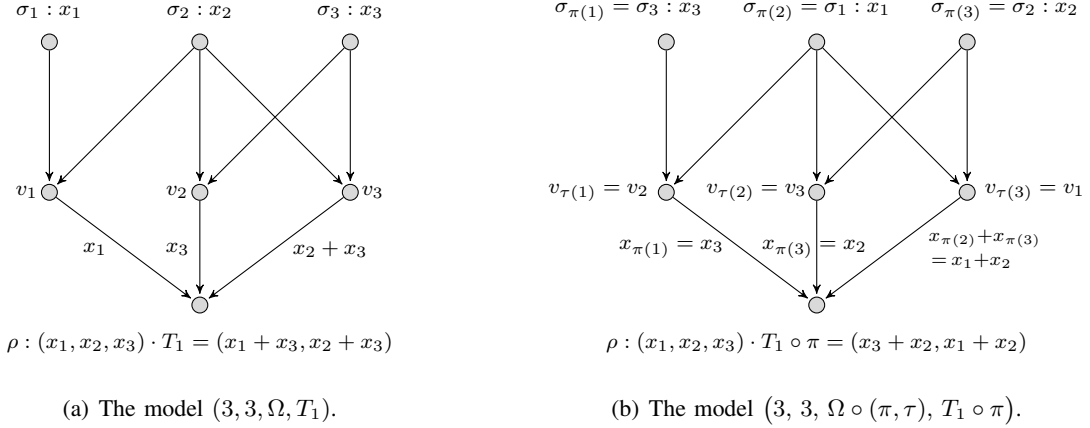
\begin{figure}[!h]
\vspace{-1em}
\centering
\subfigure[The model $(3,3,\Omega,T_1)$.]{
\begin{tikzpicture}[->,>=stealth',shorten >=1pt,auto,node distance=1.2cm]
  \tikzstyle{every state}=[fill=gray!30,draw=black,text=black,minimum size=6pt,inner sep=0pt]
  \tikzstyle{inode}=[draw,circle,fill=gray!30,minimum size=6pt, inner sep=0pt]

  \node[inode]         (a1) at (0,4) {};
  \node at (0,4.4) {\scriptsize $\sigma_1:x_1$};
  \node[inode]         (a2) at (2,4) {};
  \node at (2,4.4) {\scriptsize $\sigma_2:x_2$};
  \node[inode]         (a3) at (4,4) {};
  \node at (4,4.4) {\scriptsize $\sigma_3:x_3$};

  \node[inode]         (r1) at (0,2) {};
  \node at (-0.3,2) {\scriptsize $v_1$};
  \node[inode]         (r2) at (2,2) {};
  \node at (1.7,2) {\scriptsize $v_2$};
  \node[inode]         (r3) at (4,2) {};
  \node at (4.3,2) {\scriptsize $v_3$};

  \node[inode]         (r) at (2,0.5) {};
  \node at (2,0) {\scriptsize $\rho:(x_1,x_2,x_3)\cdot T_1=(x_1+x_3,x_2+x_3)$};

  \path
  (a1) edge (r1)
  (a2) edge (r1)
       edge (r2)
       edge (r3)
  (a3) edge (r2)
       edge (r3)
  (r1) edge node[pos=0.5, left=1mm] {\scriptsize $x_1$} (r)
  (r2) edge node[pos=0.5, left=0mm] {\scriptsize $x_3$} (r)
  (r3) edge node[pos=0.5, right=1mm] {\scriptsize $x_2+x_3$} (r);
\end{tikzpicture}
\label{Omega}
}
\hspace{0.5in}
\subfigure[The model $\big(3,\,3,\,\Omega\circ(\pi,\tau),\,T_1\circ\pi\big)$.]{
\begin{tikzpicture}[->,>=stealth',shorten >=1pt,auto,node distance=1.2cm]
  \tikzstyle{every state}=[fill=gray!30,draw=black,text=black,minimum size=6pt,inner sep=0pt]
  \tikzstyle{inode}=[draw,circle,fill=gray!30,minimum size=6pt, inner sep=0pt]

  \node[inode]         (a1) at (0,4) {};
  \node at (-0.5,4.4) {\scriptsize $\sigma_{\pi(1)}=\sigma_{3}:x_3$};
  \node[inode]         (a2) at (2,4) {};
  \node at (2,4.4) {\scriptsize $\sigma_{\pi(2)}=\sigma_{1}:x_1$};
  \node[inode]         (a3) at (4,4) {};
  \node at (4.5,4.4) {\scriptsize $\sigma_{\pi(3)}=\sigma_{2}:x_2$};

  \node[inode]         (r1) at (0,2) {};
  \node at (-0.9,2) {\scriptsize $v_{\tau(1)}=v_{2}$};
  \node[inode]         (r2) at (2,2) {};
  \node at (1.2,2) {\scriptsize $v_{\tau(2)}=v_{3}$};
  \node[inode]         (r3) at (4,2) {};
  \node at (4.9,2) {\scriptsize $v_{\tau(3)}=v_{1}$};

  \node[inode]         (r) at (2,0.5) {};
  \node at (2,0) {\scriptsize $\rho:(x_1,x_2,x_3)\cdot T_1\circ\pi=(x_3+x_2,x_1+x_2)$};

  \path
  (a1) edge (r1)
  (a2) edge (r1)
       edge (r2)
       edge (r3)
  (a3) edge (r2)
       edge (r3)
  (r1) edge node[pos=0.5, left=1mm] {\scriptsize $x_{\pi(1)}=x_3$} (r)
  (r2) edge node[pos=0.5, left=-8mm] {\scriptsize $x_{\pi(3)}=x_2$} (r)
  (r3) edge node[pos=0.5, right=0mm] {\small $\substack{\quad x_{\pi(2)}+x_{\pi(3)}\vspace{0.5mm}\\=\,x_1+x_2}$} (r);
\end{tikzpicture}
\label{Omega'}
}
\captionsetup{font=small}
\vspace{-1em}
\caption{Two isomorphic models.}
\label{omega-tau-omega-pi-Omega}
\vspace{-2em}
\end{figure}
\end{example}

We end this subsection by introducing a partial order ``$\preceq$'' on connectivity states. Continue to consider $S=\{\sigma_1,\sigma_2,\cdots,\sigma_s\}$ and $V=\{v_1,v_2,\cdots,v_m\}$. For two connectivity states $\Omega=(\Gamma_{\sigma_1}, \Gamma_{\sigma_2}, \cdots, \Gamma_{\sigma_s})$ and $\Omega'=(\Gamma'_{\sigma_1}, \Gamma'_{\sigma_2}, \cdots, \Gamma'_{\sigma_s})$, we write $\Omega\preceq\Omega'$ if $\Gamma_{\sigma_i}\subseteq \Gamma_{\sigma_i}'$ for each $1\leq i\leq s$. We can easily verify that the binary relation ``$\preceq$'' is a (non-strict) partial order. As such, we immediately obtain the following lemma for two connectivity states $\Omega$ and $\Omega'$ with $\Omega\preceq\Omega'$.

\begin{lemma}\label{lem-cap-omega'-geq-cap-omega}
Consider two models $(s, m, \Omega, f)$ and $(s, m, \Omega', f)$, where $\Omega\preceq\Omega'$. Then,
\[
\mathcal{C}(s, m, \Omega, f)\geq\mathcal{C}(s, m, \Omega', f).
\]
\end{lemma}

\section{Two Most Nontrivial Models}\label{section-not-tight}

In order to characterize the function-compression capacities for all the nontrivial models $(3,m,\Omega,T)$ with $m\leq 3$, $\textup{Rank}(T)=2$, and arbitrary connectivity states $\Omega$, by Theorem \ref{Type-I-II} it suffices to consider two matrices $T_1$ and $T_2$ (cf.~\eqref{T1-T2}).
In Section \ref{best-known-lower-bound}, we have proved lower bounds on the function-compression capacities for all the models $(3,m,\Omega,T)$. Nevertheless, the lower bounds thus obtained are not always tight. To be specific, the obtained lower bounds are tight, or equivalently, identical to the function-compression capacities for all the models $\big(3,m,\Omega,T_1\big)$ with $1\leq m\leq3$ and arbitrary connectivity states~$\Omega$, which will be discussed in the next section. The obtained lower bounds are not tight for the two models $\big(3,3,\Omega_1,T_2\big)$ and $\big(3,3,\Omega_2,T_2\big)$ depicted in Figs.\,\ref{graph-T2-m=3-xi=2/3-not-tight-fig1} and \ref{graph-T2-m=3-xi=2/3-not-tight-fig2}, respectively, where
\begin{align*}
\Omega_1 & =\big(\Gamma_{\sigma_1}=\{v_1,v_2\},~ \Gamma_{\sigma_2}=\{v_1,v_3\},~ \Gamma_{\sigma_3}=\{v_2,v_3\}\big), 
\\
\Omega_2&=\big(\Gamma_{\sigma_1}=\{v_1,v_2\},~ \Gamma_{\sigma_2}=\{v_1,v_3\},~
\Gamma_{\sigma_3}=\{v_1,v_2,v_3\}\big). 
\end{align*}
In this section, we will completely characterize the function-compression capacities for the two most nontrivial models $\big(3,3,\Omega_1,T_2\big)$ and $\big(3,3,\Omega_2,T_2\big)$. For the remaining cases of the model $\big(3,m,\Omega,T_2\big)$, the obtained lower bounds are tight, which will be discussed in the following Section \ref{section-3-m-T2}.

\begin{figure*}[t]
\tikzstyle{vertex}=[draw,circle,fill=gray!30,minimum size=6pt, inner sep=0pt]
\tikzstyle{vertex1}=[draw,circle,fill=gray!80,minimum size=6pt, inner sep=0pt]
\centering
\begin{minipage}[b]{0.5\textwidth}
\centering
{
 \begin{tikzpicture}[x=0.6cm]
     \node[draw,circle,fill=gray!30,minimum size=6pt, inner sep=0pt](a1)at(1,4){};
        \node at (1,4.4) {$\sigma_1$};
        \node[draw,circle,fill=gray!30,minimum size=6pt, inner sep=0pt](a2)at(4,4){};
        \node at (4,4.4) {$\sigma_2$};
        \node[draw,circle,fill=gray!30,minimum size=6pt, inner sep=0pt](a3)at(7,4){};
        \node at (7,4.4) {$\sigma_3$};

        \node[draw,circle,fill=gray!30,minimum size=6pt, inner sep=0pt](r1)at(1,2.5){};
        \node[draw,circle,fill=gray!30,minimum size=6pt, inner sep=0pt](r2)at(4,2.5){};
         \node[draw,circle,fill=gray!30,minimum size=6pt, inner sep=0pt](r3)at(7,2.5){};
        \node at (0.4,2.4) {$v_1$};
        \node at (3.4,2.4) {$v_2$};
        \node at (7.6,2.4) {$v_3$};

        \node[draw,circle,fill=gray!30,minimum size=6pt, inner sep=0pt](r)at(4,1){};
        \node at (4,0.6) {$\rho$};

        \draw[->,>=latex](a1)--(r1);
        \draw[->,>=latex](a1)--(r2);
        \draw[->,>=latex](a2)--(r1);
        \draw[->,>=latex](a2)--(r3);
        \draw[->,>=latex](a3)--(r2);
        \draw[->,>=latex](a3)--(r3);
        \draw[->,>=latex](r1)--(r) node[midway, auto,swap, left=0mm] {$\varphi_1$};
        \draw[->,>=latex](r2)--(r) node[midway, auto,swap, left=0mm] {$\varphi_2$};
        \draw[->,>=latex](r3)--(r) node[midway, auto,swap, right=0mm] {$\varphi_3$};
    \end{tikzpicture}
}
 \vspace{-1em}
  \caption{The model $(3,3,\Omega_1,T_2)$.}
        \label{graph-T2-m=3-xi=2/3-not-tight-fig1}
\end{minipage}%
\centering
\begin{minipage}[b]{0.5\textwidth}
\centering
 \begin{tikzpicture}[x=0.6cm]
   \node[draw,circle,fill=gray!30,minimum size=6pt, inner sep=0pt](a1)at(1,4){};
        \node at (1,4.4) {$\sigma_1$};
        \node[draw,circle,fill=gray!30,minimum size=6pt, inner sep=0pt](a2)at(4,4){};
        \node at (4,4.4) {$\sigma_2$};
        \node[draw,circle,fill=gray!30,minimum size=6pt, inner sep=0pt](a3)at(7,4){};
        \node at (7,4.4) {$\sigma_3$};

        \node[draw,circle,fill=gray!30,minimum size=6pt, inner sep=0pt](r1)at(1,2.5){};
        \node[draw,circle,fill=gray!30,minimum size=6pt, inner sep=0pt](r2)at(4,2.5){};
         \node[draw,circle,fill=gray!30,minimum size=6pt, inner sep=0pt](r3)at(7,2.5){};
        \node at (0.4,2.4) {$v_1$};
        \node at (3.4,2.4) {$v_2$};
        \node at (7.6,2.4) {$v_3$};

        \node[draw,circle,fill=gray!30,minimum size=6pt, inner sep=0pt](r)at(4,1){};
        \node at (4,0.6) {$\rho$};

         \draw[->,>=latex](a1)--(r1);
        \draw[->,>=latex](a1)--(r2);
        \draw[->,>=latex](a2)--(r1);
        \draw[->,>=latex](a2)--(r3);
        \draw[->,>=latex](a3)--(r1);
        \draw[->,>=latex](a3)--(r2);
        \draw[->,>=latex](a3)--(r3);
        \draw[->,>=latex](r1)--(r) node[midway, auto,swap, left=0mm] {$\varphi_1$};
        \draw[->,>=latex](r2)--(r) node[midway, auto,swap, left=0mm] {$\varphi_2$};
        \draw[->,>=latex](r3)--(r) node[midway, auto,swap, right=0mm] {$\varphi_3$};
    \end{tikzpicture}
     \vspace{-1em}
      \caption{The model $(3,3,\Omega_2,T_2)$.}
   \label{graph-T2-m=3-xi=2/3-not-tight-fig2}
\end{minipage}
\vspace{-4em}
\end{figure*}

\vspace{-0.5em}

\subsection{An Intuitive Example}

Consider the two models $\big(3,3,\Omega_1,T_2\big)$ and $\big(3,3,\Omega_2,T_2\big)$ depicted in Figs.\,\ref{graph-T2-m=3-xi=2/3-not-tight-fig1} and \ref{graph-T2-m=3-xi=2/3-not-tight-fig2}, respectively.
We specify the lower bound \eqref{3-m-Omega-T-geq-max-best-known-upper-bound} 
for the models $\big(3,3,\Omega_1,T_2\big)$ and $\big(3,3,\Omega_2,T_2\big)$ and thus obtain the lower bound~$2/3$ on their function-compression capacities, namely that
\[
\mathcal{C}\big(3,3,\Omega_1,T_2\big)\geq\frac{2}{3}\quad\textup{and}
\quad\mathcal{C}\big(3,3,\Omega_2,T_2\big)\geq\frac{2}{3}.
\]
To briefly see this, we take the model $\big(3,3,\Omega_1,T_2\big)$ as an example. Consider an arbitrary $k$-shot (function-compression) source code $\mathbf{C}=\big\{\varphi_1,\varphi_2,\varphi_3;\, \psi\big\}$ for $\big(3,3,\Omega_1,T_2\big)$. The triple of the three encoding functions $(\varphi_1,\varphi_2,\varphi_3)$ has to distinguish all $q^{2k}$ function values, implying that $$|\textup{Im}\,\varphi_1|\cdot|\textup{Im}\,\varphi_2|\cdot|\textup{Im}\,\varphi_3|\geq q^{2k}.$$
Together with the definition of $n\triangleq n(\mathbf{C})$ (cf. the paragraph immediately above the equation \eqref{def-capacity}), i.e., $|\textup{Im}\,\varphi_i|\leq|\mathbb{F}_q|^n$ for $i=1,2,3$, we have
\[
|\mathbb{F}_q|^{3n}\geq |\textup{Im}\,\varphi_1|\cdot|\textup{Im}\,\varphi_2|\cdot|\textup{Im}\,\varphi_3|\geq q^{2k},
\]
or equivalently, $n/k\geq 2/3$ implying $\mathcal{C}\big(3,3,\Omega_1,T_2\big)\geq2/3$.

We now use the model $\big(3,3,\Omega_1,T_2\big)$ to give an intuitive (but not complete) explanation why the lower bound $2/3$ is not tight. Suppose that the lower bound $2/3$ for the model $\big(3,3,\Omega_1,T_2\big)$ is achievable by an admissible $k$-shot source code $\mathbf{C}=\big\{\varphi_1(\boldsymbol{x}_1,\boldsymbol{x}_2),\,
\varphi_2(\boldsymbol{x}_1,\boldsymbol{x}_3),\,\varphi_3(\boldsymbol{x}_2,\boldsymbol{x}_3);~ \psi\big\}$, where $n\triangleq n(\mathbf{C})$ with $n/k=2/3$. Let us for the time being assume that $k=3$ and $n=2$.

Let $\boldsymbol{x}_i\in\mathbb{F}_q^3$ be the source message generated by $\sigma_i,\,i=1,2,3$. Then we write
\begin{equation*}\label{eg-def-varphi}
\varphi(\boldsymbol{x}_1,\boldsymbol{x}_2,\boldsymbol{x}_3)\triangleq
\big( \varphi_1(\boldsymbol{x}_1,\boldsymbol{x}_2),\,\varphi_2(\boldsymbol{x}_1,\boldsymbol{x}_3),\,
\varphi_3(\boldsymbol{x}_2,\boldsymbol{x}_3) \big).
\end{equation*}
Since the code $\mathbf{C}$ can compute at $\rho$ with zero error the vector-linear function $x_S\cdot T_2=(x_1+x_2,x_3)$, it is necessary for it to distinguish all the $q^{2\cdot 3}=q^6$ function values $(\boldsymbol{x}_1+\boldsymbol{x}_2,\boldsymbol{x}_3)$ for all $\boldsymbol{x}_i\in\mathbb{F}_q^3,~i=1,2,3$, implying $\big|\textup{Im}\,\varphi\big|\geq q^6$. On the other hand, by the definition of $n$, we have
\[
\big|\textup{Im}\,\varphi\big|\leq \big|\textup{Im}\,\varphi_1\big|\cdot\big|\textup{Im}\,\varphi_2\big|\cdot\big|\textup{Im}\,\varphi_3\big|
\leq \big|\mathbb{F}_q^2\big| \cdot \big|\mathbb{F}_q^2\big| \cdot \big|\mathbb{F}_q^2\big|=q^6.
\]
Thus, we obtain that $\big|\textup{Im}\,\varphi\big|=q^6$.
This implies that for each function value $(\boldsymbol{a},\boldsymbol{b})\in\mathbb{F}_q^3\times\mathbb{F}_q^3$,
\begin{align}\label{eg-varphi-eq-1}
\#\Big\{ \varphi(\boldsymbol{x}_1,\boldsymbol{x}_2,\boldsymbol{x}_3):~ \textup{all}~ \boldsymbol{x}_1,\boldsymbol{x}_2,\boldsymbol{x}_3\in\mathbb{F}_q^3~ \textup{s.t.}~ (\boldsymbol{x}_1,\boldsymbol{x}_2,\boldsymbol{x}_3)\cdot T_2=(\boldsymbol{a},\boldsymbol{b}) \Big\}=1.
\end{align}

We further consider the set on the LHS of \eqref{eg-varphi-eq-1}:
\begin{align*}
&\#\Big\{ \varphi(\boldsymbol{x}_1,\boldsymbol{x}_2,\boldsymbol{x}_3):~ \textup{all}~ \boldsymbol{x}_1,\boldsymbol{x}_2,\boldsymbol{x}_3\in\mathbb{F}_q^3~ \textup{s.t.}~ (\boldsymbol{x}_1,\boldsymbol{x}_2,\boldsymbol{x}_3)\cdot T_2=(\boldsymbol{a},\boldsymbol{b}) \Big\} \notag  \\
&=\#\Big\{ \varphi(\boldsymbol{x}_1,\boldsymbol{x}_2,\boldsymbol{x}_3):~ \textup{all}~ \boldsymbol{x}_1,\boldsymbol{x}_2,\boldsymbol{x}_3\in\mathbb{F}_q^3~ \textup{s.t.}~ \boldsymbol{x}_1+\boldsymbol{x}_2=\boldsymbol{a}~\textup{and}~\boldsymbol{x}_3=\boldsymbol{b} \Big\} \notag  \\
&=\#\Big\{ \varphi(\boldsymbol{x}_1,\boldsymbol{a}-\boldsymbol{x}_1,\boldsymbol{b}):~ \textup{all}~ \boldsymbol{x}_1\in\mathbb{F}_q^3 \Big\} \notag  \\
&=\#\Big\{\big( \varphi_1(\boldsymbol{x}_1,\boldsymbol{a}-\boldsymbol{x}_1),\,
\varphi_2(\boldsymbol{x}_1,\boldsymbol{b}),\,
\varphi_3(\boldsymbol{a}-\boldsymbol{x}_1,\boldsymbol{b}) \big):~ \textup{all}~ \boldsymbol{x}_1\in\mathbb{F}_q^3 \Big\}. \notag
\end{align*}
Together with \eqref{eg-varphi-eq-1}, this implies that
\begin{align}\label{eg-varphi-2-eq-1}
\#\Big\{ \varphi_2(\boldsymbol{x}_1,\boldsymbol{b}):~  \textup{all}~ \boldsymbol{x}_1\in\mathbb{F}_q^3 \Big\}=1.
\end{align}

Furthermore, we consider an arbitrary $\boldsymbol{a}_2\in\mathbb{F}_q^3$ for $\boldsymbol{x}_2$. By the definition of $n$ and recalling $n=2$, we have
\[
\#\Big\{ \varphi_1(\boldsymbol{x}_1,\boldsymbol{a}_2):~  \textup{all}~ \boldsymbol{x}_1\in\mathbb{F}_q^3 \Big\}\leq \big|\textup{Im}\,\varphi_1\big|\leq \big|\mathbb{F}_q^2\big|=q^2.
\]
This implies that there exist two different source messages $\boldsymbol{a}_1,\boldsymbol{a}_1'\in\mathbb{F}_q^3$ for $\boldsymbol{x}_1$ such that
\begin{align}\label{eg-varphi-1-a1-eq-a1'}
\varphi_1(\boldsymbol{a}_1,\boldsymbol{a}_2)=\varphi_1(\boldsymbol{a}_1',\boldsymbol{a}_2).
\end{align}
Combining the above, we can see that
\begin{align}
\varphi(\boldsymbol{a}_1,\boldsymbol{a}_2,\boldsymbol{b})&=\big( \varphi_1(\boldsymbol{a}_1,\boldsymbol{a}_2),
\,\varphi_2(\boldsymbol{a}_1,\boldsymbol{b}),\,
\varphi_3(\boldsymbol{a}_2,\boldsymbol{b}) \big) \notag \\
&=\big( \varphi_1(\boldsymbol{a}_1',\boldsymbol{a}_2),
\,\varphi_2(\boldsymbol{a}_1',\boldsymbol{b}),\,
\varphi_3(\boldsymbol{a}_2,\boldsymbol{b}) \big) \label{eg-varphi-a1-eq-a1'-a2-b} \\
&=\varphi(\boldsymbol{a}_1',\boldsymbol{a}_2,\boldsymbol{b}), \notag
\end{align}
where the equality \eqref{eg-varphi-a1-eq-a1'-a2-b} follows from $\varphi_1(\boldsymbol{a}_1,\boldsymbol{a}_2)=\varphi_1(\boldsymbol{a}_1',\boldsymbol{a}_2)$ by \eqref{eg-varphi-1-a1-eq-a1'} and $\varphi_2(\boldsymbol{a}_1,\boldsymbol{b})=\varphi_2(\boldsymbol{a}_1',\boldsymbol{b})$ by~\eqref{eg-varphi-2-eq-1}. By the admissibility of the code $\mathbf{C}$, we immediately obtain that
\[
(\boldsymbol{a}_1+\boldsymbol{a}_2,\boldsymbol{b})=\psi\big( \varphi(\boldsymbol{x}_1=\boldsymbol{a}_1,\boldsymbol{x}_2=\boldsymbol{a}_2,\boldsymbol{x}_3=\boldsymbol{b}) \big)=\psi\big( \varphi(\boldsymbol{x}_1=\boldsymbol{a}_1',\boldsymbol{x}_2=\boldsymbol{a}_2,\boldsymbol{x}_3=\boldsymbol{b}) \big)=(\boldsymbol{a}_1'+\boldsymbol{a}_2,\boldsymbol{b}),
\]
a contradiction to the fact that $\boldsymbol{a}_1+\boldsymbol{a}_2\neq\boldsymbol{a}_1'+\boldsymbol{a}_2$. Hence, the lower bound $2/3$ is not tight for the model $(3,3,\Omega_1,T_2)$. This implies that the necessary condition that has been used to obtain the lower bound $2/3$ is not strong enough to be also sufficient. In the next subsection, we will prove that $\mathcal{C}(3,3,\Omega_1,T_2)=3/4$, which is considerably larger than the lower bound $2/3$.

\vspace{-0.5em}

\subsection{Capacity Characterization for Two Models $\big(3,3,\Omega_1,T_2\big)$ and $\big(3,3,\Omega_2,T_2\big)$}

We characterize the function-compression capacities for two models $\big(3,3,\Omega_1,T_2\big)$ and $\big(3,3,\Omega_2,T_2\big)$ in the theorem below.

\begin{theorem}\label{cap-T2-m=3-3/4}
Consider the two models $\big(3,3,\Omega_1,T_2\big)$ and $\big(3,3,\Omega_2,T_2\big)$, where
\begin{align*}
\Omega_1=\big(\Gamma_{\sigma_1}=\{v_1,v_2\},~ \Gamma_{\sigma_2}=\{v_1,v_3\},~
\Gamma_{\sigma_3}=\{v_2,v_3\}\big)
\end{align*}
and
\begin{align*}
\Omega_2=\big(\Gamma_{\sigma_1}=\{v_1,v_2\},~ \Gamma_{\sigma_2}=\{v_1,v_3\},~
\Gamma_{\sigma_3}=\{v_1,v_2,v_3\} \big)
\end{align*}
as depicted in Figs.\,\ref{graph-T2-m=3-xi=2/3-not-tight-fig1} and \ref{graph-T2-m=3-xi=2/3-not-tight-fig2}, respectively.
Then
\[
\mathcal{C}\big(3,3,\Omega_1,T_2\big)=\mathcal{C}\big(3,3,\Omega_2,T_2\big)=\frac{3}{4}.
\]
\end{theorem}
\begin{proof}
We first note that $\Omega_1\preceq\Omega_2$, which implies that $\mathcal{C}\big(3,3,\Omega_1,T_2\big)\geq\mathcal{C}\big(3,3,\Omega_2,T_2\big)$ by Lemma \ref{lem-cap-omega'-geq-cap-omega}. So in order to prove the theorem, we only need to prove $\mathcal{C}\big(3,3,\Omega_2,T_2\big)\geq3/4$ and $\mathcal{C}\big(3,3,\Omega_1,T_2\big)\leq3/4$ corresponding to the converse and the achievability, respectively.

\smallskip

Now, we start to prove that $\mathcal{C}\big(3,3,\Omega_2,T_2\big)\geq3/4$. Consider an arbitrary positive integer $k$ and let $\mathbf{C}=\{\varphi_1, \varphi_2, \varphi_3;~ \psi\}$ be an arbitrary admissible $k$-shot source code for $\big(3,3,\Omega_2,T_2\big)$. For any source messages $\boldsymbol{x}_i\in\mathbb{F}_q^k$ for $i=1,2,3$, we write
\begin{equation*}\label{def-varphi}
\varphi(\boldsymbol{x}_1,\boldsymbol{x}_2,\boldsymbol{x}_3)\triangleq
\big( \varphi_1(\boldsymbol{x}_1,\boldsymbol{x}_2,\boldsymbol{x}_3),\,\varphi_2(\boldsymbol{x}_1,\boldsymbol{x}_3),\,
\varphi_3(\boldsymbol{x}_2,\boldsymbol{x}_3) \big),
\end{equation*}
and clearly,
\begin{equation}\label{varphi-leq-q-3n}
\begin{split}
\#\Big\{ \varphi(\boldsymbol{x}_1,\boldsymbol{x}_2,\boldsymbol{x}_3):~ \textup{all}~ \boldsymbol{x}_1,\boldsymbol{x}_2,\boldsymbol{x}_3\in\mathbb{F}_q^k \Big\}\leq \big|\textup{Im}\,\varphi_1\big|\cdot\big|\textup{Im}\,\varphi_2\big|\cdot\big|\textup{Im}\,\varphi_3\big| \leq q^{3n},
\end{split}
\end{equation}
where we let $n\triangleq n(\mathbf{C})$ for notational simplicity. By $\boldsymbol{x}_S\cdot T_2=(\boldsymbol{x}_1+\boldsymbol{x}_2,\boldsymbol{x}_3)$, all $(\boldsymbol{a},\boldsymbol{b})$ in $\mathbb{F}_q^k\times\mathbb{F}_q^k$ are all the function values and we consider
\begin{align*}
&\Big\{ \varphi(\boldsymbol{x}_1,\boldsymbol{x}_2,\boldsymbol{x}_3):~ \textup{all}~ \boldsymbol{x}_1,\boldsymbol{x}_2,\boldsymbol{x}_3\in\mathbb{F}_q^k \Big\} \notag\\
&= \bigcup\limits_{(\boldsymbol{a},\boldsymbol{b})\in\mathbb{F}_q^k\times\mathbb{F}_q^k
} \Big\{ \varphi(\boldsymbol{x}_1,\boldsymbol{x}_2,\boldsymbol{x}_3):~ \textup{all}~ \boldsymbol{x}_1,\boldsymbol{x}_2,\boldsymbol{x}_3\in\mathbb{F}_q^k~ \textup{s.t.}~ (\boldsymbol{x}_1+\boldsymbol{x}_2,\boldsymbol{x}_3)=(\boldsymbol{a},\boldsymbol{b}) \Big\}.
\end{align*}
This implies that
\begin{align}
&\#\Big\{ \varphi(\boldsymbol{x}_1,\boldsymbol{x}_2,\boldsymbol{x}_3):~ \textup{all}~ \boldsymbol{x}_1,\boldsymbol{x}_2,\boldsymbol{x}_3\in\mathbb{F}_q^k \Big\}\notag\\
&=\sum\limits_{(\boldsymbol{a},\boldsymbol{b})\in\mathbb{F}_q^k\times\mathbb{F}_q^k
}\#\Big\{ \varphi(\boldsymbol{x}_1,\boldsymbol{x}_2,\boldsymbol{x}_3):~ \textup{all}~ \boldsymbol{x}_1,\boldsymbol{x}_2,\boldsymbol{x}_3\in\mathbb{F}_q^k~ \textup{s.t.}~ (\boldsymbol{x}_1+\boldsymbol{x}_2,\boldsymbol{x}_3)=(\boldsymbol{a},\boldsymbol{b}) \Big\}  \label{varphi-supseteq-1} \\
&=\sum\limits_{(\boldsymbol{a},\boldsymbol{b})\in\mathbb{F}_q^k\times\mathbb{F}_q^k
}\#\Big\{ \varphi(\boldsymbol{x}_1,\boldsymbol{x}_2,\boldsymbol{b}):~ \textup{all}~ \boldsymbol{x}_1,\boldsymbol{x}_2\in\mathbb{F}_q^k~ \textup{s.t.}~ \boldsymbol{x}_1+\boldsymbol{x}_2=\boldsymbol{a} \Big\}, \label{varphi-supseteq-x3=b}
\end{align}
where the equality \eqref{varphi-supseteq-1} follows from the admissibility of the code $\mathbf{C}$.


Following from \eqref{varphi-supseteq-x3=b}, for each function value $(\boldsymbol{a},\boldsymbol{b})\in\mathbb{F}_q^k\times\mathbb{F}_q^k$, we continue to consider
\begin{align}
&\#\Big\{ \varphi(\boldsymbol{x}_1,\boldsymbol{x}_2,\boldsymbol{b}):~ \textup{all}~ \boldsymbol{x}_1,\boldsymbol{x}_2\in\mathbb{F}_q^k~ \textup{s.t.}~ \boldsymbol{x}_1+\boldsymbol{x}_2=\boldsymbol{a}\Big\}\nonumber\\
&=\#\Big\{\big( \varphi_1(\boldsymbol{x}_1,\boldsymbol{x}_2,\boldsymbol{b}),\,
\varphi_2(\boldsymbol{x}_1,\boldsymbol{b}),\,
\varphi_3(\boldsymbol{x}_2,\boldsymbol{b}) \big):~ \textup{all}~ \boldsymbol{x}_1,\boldsymbol{x}_2\in\mathbb{F}_q^k~ \textup{s.t.}~ \boldsymbol{x}_1+\boldsymbol{x}_2=\boldsymbol{a}\Big\} \notag \\
&\geq\#\Big\{
\varphi_2(\boldsymbol{x}_1,\boldsymbol{b}):~ \textup{all}~ \boldsymbol{x}_1,\boldsymbol{x}_2\in\mathbb{F}_q^k~ \textup{s.t.}~ \boldsymbol{x}_1+\boldsymbol{x}_2=\boldsymbol{a}\Big\}\nonumber\\
&=\#\Big\{
\varphi_2(\boldsymbol{x}_1,\boldsymbol{b}):~ \textup{all}~ \boldsymbol{x}_1\in\mathbb{F}_q^k\Big\},\label{varphi-z1z2geq-eq2}
\end{align}
where \eqref{varphi-z1z2geq-eq2} holds because for each $\boldsymbol{x}_1\in\mathbb{F}_q^k$, there always exists a source message $\boldsymbol{x}_2\in\mathbb{F}_q^k$ such that $\boldsymbol{x}_1+\boldsymbol{x}_2=\boldsymbol{a}$. Combining \eqref{varphi-supseteq-x3=b} and \eqref{varphi-z1z2geq-eq2}, we obtain that
\begin{align}
\#\Big\{ \varphi(\boldsymbol{x}_1,\boldsymbol{x}_2,\boldsymbol{x}_3):~ \textup{all}~ \boldsymbol{x}_1,\boldsymbol{x}_2,\boldsymbol{x}_3\in\mathbb{F}_q^k \Big\}\geq\sum\limits_{(\boldsymbol{a},\boldsymbol{b})\in\mathbb{F}_q^k\times\mathbb{F}_q^k
}\#\Big\{
\varphi_2(\boldsymbol{x}_1,\boldsymbol{b}):~ \textup{all}~ \boldsymbol{x}_1\in\mathbb{F}_q^k \Big\}. \label{varphi-supseteq-2}
\end{align}

Furthermore, we claim that for each $\boldsymbol{b}\in\mathbb{F}_q^k$,
\begin{align}\label{Im-varphi2-geq-q-k-n}
\#\Big\{
\varphi_2(\boldsymbol{x}_1,\boldsymbol{b}):~ \textup{all}~ \boldsymbol{x}_1\in\mathbb{F}_q^k \Big\}\geq q^{k-n},
\end{align}
which will become clear later. Together with \eqref{varphi-supseteq-2}, we thus obtain that
\begin{align}
&\#\Big\{ \varphi(\boldsymbol{x}_1,\boldsymbol{x}_2,\boldsymbol{x}_3):~ \textup{all}~ \boldsymbol{x}_1,\boldsymbol{x}_2,\boldsymbol{x}_3\in\mathbb{F}_q^k \Big\}\geq\sum\limits_{(\boldsymbol{a},\boldsymbol{b})\in\mathbb{F}_q^k\times\mathbb{F}_q^k
} \#\Big\{
\varphi_2(\boldsymbol{x}_1,\boldsymbol{b}):~ \textup{all}~ \boldsymbol{x}_1\in\mathbb{F}_q^k\Big\} \notag \\
&\geq\sum\limits_{(\boldsymbol{a},\boldsymbol{b})\in\mathbb{F}_q^k\times\mathbb{F}_q^k
}q^{k-n}=q^{3k-n}. \label{varphi-geq-q-k-3n}
\end{align}
Combining \eqref{varphi-leq-q-3n} and \eqref{varphi-geq-q-k-3n}, we have
\begin{align*}
q^{3n}\geq\#\Big\{ \varphi(\boldsymbol{x}_1,\boldsymbol{x}_2,\boldsymbol{x}_3):~ \textup{all}~ \boldsymbol{x}_1,\boldsymbol{x}_2,\boldsymbol{x}_3\in\mathbb{F}_q^k \Big\}\geq q^{3k-n},
\end{align*}
namely that $n/k\geq3/4$. We have thus proved that $R(\mathbf{C})\geq3/4$.
Furthermore, we note that the lower bound $3/4$ on the coding rate is true for each positive integer $k$ and each admissible $k$-shot source code for $\big(3,3,\Omega_2,T_2\big)$, and thus we have proved that
\begin{align*}
\mathcal{C}\big(3,3,\Omega_2,T_2\big)\geq\frac{~3~}{~4~}.
\end{align*}

We now prove the claim \eqref{Im-varphi2-geq-q-k-n}. For any fixed source message $\boldsymbol{c}\in\mathbb{F}_q^k$ for $\boldsymbol{x}_2$, we first consider
\begin{align}
&\#\Big\{ \varphi(\boldsymbol{x}_1,\boldsymbol{c},\boldsymbol{b}):\, \textup{all}~ \boldsymbol{x}_1\in\mathbb{F}_q^k \Big\} \notag \\
&\geq\#\Big\{ (\boldsymbol{x}_1,\boldsymbol{c},\boldsymbol{b})\cdot T_2:\, \textup{all}~ \boldsymbol{x}_1\in\mathbb{F}_q^k \Big\} \label{varphi-c-geq-T2} \\
&=\#\Big\{ (\boldsymbol{x}_1+\boldsymbol{c},\boldsymbol{b}):\, \textup{all}~ \boldsymbol{x}_1\in\mathbb{F}_q^k \Big\}=q^k, \label{varphi-z1z2geq-eq4}
\end{align}
where \eqref{varphi-c-geq-T2} follows from the admissibility of the code $\mathbf{C}$. On the other hand, we write
\begin{align}
&\#\Big\{ \varphi(\boldsymbol{x}_1,\boldsymbol{c},\boldsymbol{b}):~ \textup{all}~ \boldsymbol{x}_1\in\mathbb{F}_q^k \Big\} \notag \\
&=\#\Big\{ \big( \varphi_1(\boldsymbol{x}_1,\boldsymbol{c},\boldsymbol{b}),\, \varphi_2(\boldsymbol{x}_1,\boldsymbol{b}),\, \varphi_3(\boldsymbol{c},\boldsymbol{b}) \big):~ \textup{all}~ \boldsymbol{x}_1\in\mathbb{F}_q^k \Big\} \notag \\
&=\#\Big\{ \big( \varphi_1(\boldsymbol{x}_1,\boldsymbol{c},\boldsymbol{b}),\, \varphi_2(\boldsymbol{x}_1,\boldsymbol{b})\big):~ \textup{all}~ \boldsymbol{x}_1\in\mathbb{F}_q^k \Big\} \notag \\
&\leq\#\Big\{ \varphi_1(\boldsymbol{x}_1,\boldsymbol{c},\boldsymbol{b}):~ \textup{all}~ \boldsymbol{x}_1\in\mathbb{F}_q^k \Big\} \cdot \#\Big\{ \varphi_2(\boldsymbol{x}_1,\boldsymbol{b}):~ \textup{all}~ \boldsymbol{x}_1\in\mathbb{F}_q^k \Big\} \notag \\
&\leq|\textup{Im}\, \varphi_1| \cdot \#\Big\{ \varphi_2(\boldsymbol{x}_1,\boldsymbol{b}):~ \textup{all}~ \boldsymbol{x}_1\in\mathbb{F}_q^k \Big\} \notag \\
&\leq q^n\cdot \#\Big\{ \varphi_2(\boldsymbol{x}_1,\boldsymbol{b}):~ \textup{all}~ \boldsymbol{x}_1\in\mathbb{F}_q^k \Big\}. \label{varphi-z1z2geq-eq5}
\end{align}
Combining \eqref{varphi-z1z2geq-eq4} and \eqref{varphi-z1z2geq-eq5}, we obtain that
\begin{align*}
q^k\leq\#\Big\{ \varphi(\boldsymbol{x}_1,\boldsymbol{c},\boldsymbol{b}):~ \textup{all}~ \boldsymbol{x}_1\in\mathbb{F}_q^k \Big\}\leq  q^n\cdot \#\Big\{ \varphi_2(\boldsymbol{x}_1,\boldsymbol{b}):~ \textup{all}~ \boldsymbol{x}_1\in\mathbb{F}_q^k \Big\},
\end{align*}
implying that
\begin{align*}
\#\Big\{
\varphi_2(\boldsymbol{x}_1,\boldsymbol{b}):~ \textup{all}~ \boldsymbol{x}_1\in\mathbb{F}_q^k\Big\}\geq q^{k-n}.  
\end{align*}
Hence, we have proved the claim \eqref{Im-varphi2-geq-q-k-n}.

\smallskip

Next, we will prove that $\mathcal{C}\big(3,3,\Omega_1,T_2\big)\leq3/4$. Toward this end, we construct an admissible $4$-shot source code $\mathbf{C}=\{\varphi_1,\varphi_2,\varphi_3;\,\psi\}$ for $(3,3,\Omega_1,T_2)$ with $n(\mathbf{C})=3$ as follows, which implies that the coding rate $R(\mathbf{C})=n(\mathbf{C})/k=3/4$. Let $\boldsymbol{x}_i=(x_{i,1},\,x_{i,2},\,x_{i,3},\,x_{i,4})^\top$ for $i=1,2,3$, and
\begin{align*}
&\varphi_1(\boldsymbol{x}_1,\boldsymbol{x}_2)=(x_{1,1}+x_{2,1},~ x_{1,2}+x_{2,2},~ x_{1,3}+x_{2,3}),\\ &\varphi_2(\boldsymbol{x}_1,\boldsymbol{x}_3)=(x_{1,4},~ x_{3,1},~ x_{3,2}), \\ &\varphi_3(\boldsymbol{x}_2,\boldsymbol{x}_3)=(x_{2,4},~ x_{3,3},~ x_{3,4}).
\end{align*}
With the received messages $\varphi_1(\boldsymbol{x}_1,\boldsymbol{x}_2)$, $\varphi_2(\boldsymbol{x}_1,\boldsymbol{x}_3)$ and $\varphi_3(\boldsymbol{x}_2,\boldsymbol{x}_3)$, we can compute at the decoder $\rho$
\[
\boldsymbol{x}_S\cdot T_2=
\begin{bmatrix}
\hspace{-0.1mm}x_{1,1}+x_{2,1} &x_{3,1}\\
\hspace{-0.1mm}x_{1,2}+x_{2,2} &x_{3,2}\\
\hspace{-0.1mm}x_{1,3}+x_{2,3} &x_{3,3}\\
\hspace{-0.1mm}x_{1,4}+x_{2,4} &x_{3,4}
\end{bmatrix}.
\]
This immediately implies that $\mathcal{C}(3,3,\Omega_1,T_2)\leq3/4$ and also $\mathcal{C}(3,3,\Omega_2,T_2)\leq3/4$. Together with the converse part $\mathcal{C}(3,3,\Omega_1,T_2)\geq\mathcal{C}(3,3,\Omega_2,T_2)\geq3/4$, the theorem is proved.
\end{proof}

\medskip

\begin{figure*}[t]
\tikzstyle{vertex}=[draw,circle,fill=gray!30,minimum size=6pt, inner sep=0pt]
\tikzstyle{vertex1}=[draw,circle,fill=gray!80,minimum size=6pt, inner sep=0pt]
\centering
\begin{minipage}[b]{0.5\textwidth}
\centering
{
 \begin{tikzpicture}[x=0.6cm]
     \node[draw,circle,fill=gray!30,minimum size=6pt, inner sep=0pt](a1)at(1,4.5){};
        \node at (1,4.9) {$\sigma_1$};
        \node[draw,circle,fill=gray!30,minimum size=6pt, inner sep=0pt](a2)at(4,4.5){};
        \node at (4,4.9) {$\sigma_2$};
        \node[draw,circle,fill=gray!30,minimum size=6pt, inner sep=0pt](a3)at(7,4.5){};
        \node at (7,4.9) {$\sigma_3$};

        \node[draw,circle,fill=gray!30,minimum size=6pt, inner sep=0pt](r1)at(1,2.5){};
        \node[draw,circle,fill=gray!30,minimum size=6pt, inner sep=0pt](r2)at(4,2.5){};
         \node[draw,circle,fill=gray!30,minimum size=6pt, inner sep=0pt](r3)at(7,2.5){};
        \node at (0.4,2.4) {$v_1$};
        \node at (3.4,2.4) {$v_2$};
        \node at (7.6,2.4) {$v_3$};

        \node[draw,circle,fill=gray!30,minimum size=6pt, inner sep=0pt](r)at(4,1){};
        \node at (4,0.6) {$\rho$};


        \draw[->,>=latex](a1)--(r1);
        \draw[->,>=latex](a1)--(r2);
        \draw[->,>=latex](a2)--(r1);
        \draw[->,>=latex](a2)--(r3);
        \draw[->,>=latex](a3)--(r2);
        \draw[->,>=latex](a3)--(r3);

        \draw[->,>=latex] (a1) to[bend right=15](r1);
        \draw[->,>=latex] (a2) to[bend right=15](r1);
        \draw[->,>=latex] (a3) to[bend right=15](r2);
        \draw[->,>=latex] (a1) to[bend left=15](r2);
        \draw[->,>=latex] (a2) to[bend left=15](r3);
        \draw[->,>=latex] (a3) to[bend left=15](r3);

        \draw[->,>=latex](r1)--(r) node[midway, auto,swap, left=0mm] {$e_1$};
        \draw[->,>=latex](r2)--(r) node[midway, auto,swap, left=0mm] {$e_2$};
        \draw[->,>=latex](r3)--(r) node[midway, auto,swap, right=0mm] {$e_3$};
    \end{tikzpicture}
}
 \vspace{-1em}
  \caption{The network function computation \vspace{-1mm} \\ model $(\mathcal{N}_1,T_2)$.}
  \label{graph-T2-m=3-xi=2/3-not-tight-fig1-NFC}
\end{minipage}%
\centering
\begin{minipage}[b]{0.5\textwidth}
\centering
 \begin{tikzpicture}[x=0.6cm]
   \node[draw,circle,fill=gray!30,minimum size=6pt, inner sep=0pt](a1)at(1,4.5){};
        \node at (1,4.9) {$\sigma_1$};
        \node[draw,circle,fill=gray!30,minimum size=6pt, inner sep=0pt](a2)at(4,4.5){};
        \node at (4,4.9) {$\sigma_2$};
        \node[draw,circle,fill=gray!30,minimum size=6pt, inner sep=0pt](a3)at(7,4.5){};
        \node at (7,4.9) {$\sigma_3$};

        \node[draw,circle,fill=gray!30,minimum size=6pt, inner sep=0pt](r1)at(1,2.5){};
        \node[draw,circle,fill=gray!30,minimum size=6pt, inner sep=0pt](r2)at(4,2.5){};
         \node[draw,circle,fill=gray!30,minimum size=6pt, inner sep=0pt](r3)at(7,2.5){};
        \node at (0.4,2.4) {$v_1$};
        \node at (3.4,2.4) {$v_2$};
        \node at (7.6,2.4) {$v_3$};

        \node[draw,circle,fill=gray!30,minimum size=6pt, inner sep=0pt](r)at(4,1){};
        \node at (4,0.6) {$\rho$};


         \draw[->,>=latex](a1)--(r1);
        \draw[->,>=latex](a1)--(r2);
        \draw[->,>=latex](a2)--(r1);
        \draw[->,>=latex](a2)--(r3);
        \draw[->,>=latex](a3)--(r1);
        \draw[->,>=latex](a3)--(r2);
        \draw[->,>=latex](a3)--(r3);

        \draw[->,>=latex] (a1) to[bend right=15](r1);
        \draw[->,>=latex] (a2) to[bend right=15](r1);
        \draw[->,>=latex] (a3) to[bend right=10](r1);
        \draw[->,>=latex] (a3) to[bend left=15](r2);
        \draw[->,>=latex] (a1) to[bend left=15](r2);
        \draw[->,>=latex] (a2) to[bend left=15](r3);
        \draw[->,>=latex] (a3) to[bend left=15](r3);

        \draw[->,>=latex](r1)--(r) node[midway, auto,swap, left=0mm] {$e_1$};
        \draw[->,>=latex](r2)--(r) node[midway, auto,swap, left=0mm] {$e_2$};
        \draw[->,>=latex](r3)--(r) node[midway, auto,swap, right=0mm] {$e_3$};
    \end{tikzpicture}
     \vspace{-1em}
      \caption{The network function computation \vspace{-1mm} \\ model $(\mathcal{N}_2,T_2)$.}
   \label{graph-T2-m=3-xi=2/3-not-tight-fig2-NFC}
\end{minipage}
  \vspace{-4em}
\end{figure*}

An important application of the function-compression capacity $3/4$ for the models $(3,3,\Omega_1,T_2)$ and $(3,3,\Omega_2,T_2)$ is in the tightness of the best known upper bound on the computing capacity in network function computation \cite{Guang_Improved_upper_bound}, where an open problem that whether this upper bound is in general tight was given. By the equivalence of the distributed source coding model for function compression and the model of network function computation discussed in Section \ref{NFC}, we can transform $(3,3,\Omega_1,T_2)$ and $(3,3,\Omega_2,T_2)$ to two models of network function computation, denoted by  $(\mathcal{N}_1,T_2)$ and $(\mathcal{N}_2,T_2)$ and depicted in Figs. \ref{graph-T2-m=3-xi=2/3-not-tight-fig1-NFC} and \ref{graph-T2-m=3-xi=2/3-not-tight-fig2-NFC}, respectively. More precisely, by \eqref{def-ell} we note that $m=3$ and $r=\textup{Rank}(T_2)=2$, so that  $\ell=\lceil m/r \rceil=2$ and then there are two edges from $\sigma_i$ to $v_j$ in the corresponding $\Gamma_{\sigma_i}$. See Figs.~\ref{graph-T2-m=3-xi=2/3-not-tight-fig1-NFC} and \ref{graph-T2-m=3-xi=2/3-not-tight-fig2-NFC}. By Theorem \ref{cap-T2-m=3-3/4} and \eqref{correlation}, the computing capacities for the two models are
\[
\mathcal{C}(\mathcal{N}_1,T_2)=\mathcal{C}(\mathcal{N}_2,T_2)=\frac{4}{3}.
\]
However, we have known that the best known upper bound proved by Guang \textit{et al.}~\cite{Guang_Improved_upper_bound} on the computing capacities $\mathcal{C}(\mathcal{N}_1,T_2)$ and $\mathcal{C}(\mathcal{N}_2,T_2)$ are identical to $3/2$. This thus implies that the best known upper bound is in general not tight for computing vector-linear functions, rather than computing scalar-linear functions of which the computing capacities over arbitrary network topologies are able to be characterized by this upper bound proved in \cite{Guang_Improved_upper_bound}. This shows that there exists essential difference for computing vector-linear functions and scalar-linear functions over a network.

Furthermore, the binary arithmetic sum is the only target function for all previously considered network function computation problems for which the best known upper bound proved in \cite{Guang_Improved_upper_bound} is not tight for their computing capacities (cf. \cite{Guang_Zhang_Arithmetic_sum_TIT} and \cite{Guang_Zhang_Arithmetic_sum_Sel_Areas}). Here, we have given another target function which for computing over networks the best known upper bound is in general not tight for the computing capacities.

\section{Capacity Characterization for the Model $(3,m,\Omega,T_1)$}\label{section-3-m-T1}

In this section, we characterize the function-compression capacities for all the models $(3,m,\Omega,T_1)$ with $1\leq m\leq3$ and arbitrary connectivity states~$\Omega$.

\begin{theorem}\label{cap-T1}
Consider a model $\big(3,m,\Omega,T_1\big)$, where $1\leq m\leq 3$ and $\Omega$ is an arbitrary connectivity state. Then
\begin{equation*}\label{cap-3-m-T1}
\mathcal{C}\big(3,m,\Omega,T_1\big)=\max\limits_{\Gamma\subseteq V} \frac{~\textup{Rank}\big(T_1[I_{\Gamma}]\big)~}{~|\Gamma|~},
\end{equation*}
where we recall that
\begin{align}\label{T1-def-I-Gamma}
I_{\Gamma} = \big\{ \sigma_i\in S: \rho~ \text{is separated from}~ \sigma_i~ \text{upon deleting the edges}~ (v,\rho)~ \textup{for all}~ v\in\Gamma  \big\}
\end{align}
\textup{(cf. $\eqref{id-def-I-Gamma}$)}, and $T_1[I_{\Gamma}]$ stands for the submatrix of $T_1$ containing the $i$th row if $\sigma_i\in I_{\Gamma}$.\footnote{Here, for a matrix $T$ and a source subset $I\subseteq S$, we let $\textup{Rank}\big(T[I]\big)=0$ if $I=\emptyset$.}
\end{theorem}

We start to prove Theorem \ref{cap-T1}. Consider an arbitrary model $\big(3,m,\Omega,T_1\big)$ with $1\leq m\leq 3$, and let~$(\mathcal{N},T_1)$ be the corresponding model of network function computation. We specify the lower bound~\eqref{3-m-Omega-T-geq-max-best-known-upper-bound} for the model $\big(3,m,\Omega,T_1\big)$ in the following.
We first note that for $T_1=\left[\begin{smallmatrix} 1&0\\ 0&1\\ 1&1 \end{smallmatrix}\right]$
(cf. \eqref{T1-T2}), any two row vectors of $T_1$ are linearly independent. With this, for a cut set $C\in\Lambda(\mathcal{N})$ and a strong partition $\mathcal{P}_{C}=\{C_1,C_2,\cdots,C_t\}$ of $C$ (where $1\leq t\leq 3$),
we have
\begin{align}
\textbf{\textup{rank}}_{\mathcal{P}_{C}}(T_1)&=
\sum_{i=1}^t \textup{Rank}\big(T_1[I_{C_i}]\big)+\textup{Rank}\big(T_1[I_{C}]\big)
-\textup{Rank}\big(T_1[\cup_{i=1}^t I_{C_i}]\big) 
\notag \\
&=
\sum_{i=1}^t \textup{Rank}\big( T_1[I_{C_i}] \big) \label{rank-Gamma-T1-eq-i=1-t-Rank-T1-I-Gamma-i} \\
&\geq\textup{Rank}\big(T_1[I_{C}]\big). \label{rank-Gamma-T1-geq-Rank-T1-I-Gamma}
\end{align}
To see this, we note that the equality \eqref{rank-Gamma-T1-eq-i=1-t-Rank-T1-I-Gamma-i} is true for the trivial strong partition $\mathcal{P}_{C}=\{C\}$ (i.e., $t=1$) of~$C$, and then the inequality \eqref{rank-Gamma-T1-geq-Rank-T1-I-Gamma} holds.\footnote{In fact, for the trivial strong partition $\mathcal{P}_{C}=\{C\}$ of~$C$, the inequality \eqref{rank-Gamma-T1-geq-Rank-T1-I-Gamma} holds with equality.}
Otherwise, for a nontrivial strong partition $\mathcal{P}_{C}=\{C_1,C_2,\cdots,C_t\}$ (i.e., $t=2$ or $3$) of $C$, we can see that
\[
2=\textup{Rank}\big(T_1[I_{C}]\big)=\textup{Rank}\big(T_1[\cup_{i=1}^t I_{C_i}]\big),
\]
immediately implying the equality \eqref{rank-Gamma-T1-eq-i=1-t-Rank-T1-I-Gamma-i} and further the inequality \eqref{rank-Gamma-T1-geq-Rank-T1-I-Gamma}. With \eqref{rank-Gamma-T1-eq-i=1-t-Rank-T1-I-Gamma-i} and \eqref{rank-Gamma-T1-geq-Rank-T1-I-Gamma}, we consider the lower bound \eqref{3-m-Omega-T-geq-max-best-known-upper-bound} and obtain that
\begin{align}
\max\limits_{C\in\Lambda(\mathcal{N})}\, \max_{ \substack{\textup{all strong partitions}\\ \mathcal{P}_C\,\textup{of}\,C} }  \frac{~\textbf{\textup{rank}}_{\mathcal{P}_C}(T_1)~}{~|C|~}&=\max_{C\in\Lambda(\mathcal{N})} \max_{ \substack{\textup{all strong partitions}\\ \mathcal{P}_{C}=\{C_1,C_2,\cdots,C_t\}\, \textup{of}\,C} } \frac{~\sum_{i=1}^t \textup{Rank}\big(T_1[I_{C_i}]\big)~}{~|C|~}
\label{max-Gamma-V-rank-Gamma-T1-Gamma-eq-max-Gamma-V-all-sp-T1-I-Gamma-i} \\
&\geq\max\limits_{C\in\Lambda(\mathcal{N})} \frac{~\textup{Rank}\big(T_1[I_{C}]\big)~}{~|C|~}.\label{equ2}
\end{align}
On the other hand, we continue to consider \eqref{max-Gamma-V-rank-Gamma-T1-Gamma-eq-max-Gamma-V-all-sp-T1-I-Gamma-i} as follows:
\begin{align}
&\max\limits_{C\in\Lambda(\mathcal{N})}\, \max_{ \substack{\textup{all strong partitions}\\ \mathcal{P}_C\,\textup{of}\,C} }  \frac{~\textbf{\textup{rank}}_{\mathcal{P}_C}(T_1)~}{~|C|~}=\max\limits_{C\in\Lambda(\mathcal{N})} \max\limits_{ \substack{\textup{all strong partitions}\\ \mathcal{P}_{C}=\{C_1,C_2,\cdots,C_t\}\, \textup{of}\,C} } \frac{~\sum_{i=1}^t \textup{Rank}\big(T_1[I_{C_i}]\big)~}{~|C|~} \notag \\
&=\max\limits_{C\in\Lambda(\mathcal{N})} \max\limits_{ \substack{\textup{all strong partitions}\\ \mathcal{P}_{C}=\{C_1,C_2,\cdots,C_t\}\, \textup{of}\,C} } \frac{~\sum_{i=1}^t \textup{Rank}\big(T_1[I_{C_i}]\big)~}{~\sum_{i=1}^t |C_i|~} \label{max-Gamma-rank-Gamma-T1-Gamma-eq-max-Gamma-all-sp-i=1-t-Rank-T1-I-Gamma-i-Gamma-i} \\
&\leq\max\limits_{C\in\Lambda(\mathcal{N})} \max\limits_{ \substack{\textup{all strong partitions}\\ \mathcal{P}_{C}=\{C_1,C_2,\cdots,C_t\}\, \textup{of}\,C} }  \max_{1\leq i\leq t} \frac{~\textup{Rank}\big(T_1[I_{C_i}]\big)~}{~|C_i|~}  \label{max-Gamma-rank-Gamma-T1-Gamma-leq-max-Gamma-all-sp-max-i=1-t-Rank-T1-I-Gamma-i-Gamma-i} \\
&\leq \max\limits_{C\in\Lambda(\mathcal{N})} \frac{~\textup{Rank}\big(T_1[I_{C}]\big)~}{~|C|~}, \label{max-Gamma-rank-Gamma-T1-Gamma-leq-max-Gamma-Rank-T1-I-Gamma}
\end{align}
where the equality \eqref{max-Gamma-rank-Gamma-T1-Gamma-eq-max-Gamma-all-sp-i=1-t-Rank-T1-I-Gamma-i-Gamma-i} follows from $|C|=\sum_{i=1}^t |C_i|$ for a strong partition $\mathcal{P}_{C}=\{C_1,C_2,\cdots,C_t\}$ of~$C$; the inequality \eqref{max-Gamma-rank-Gamma-T1-Gamma-leq-max-Gamma-all-sp-max-i=1-t-Rank-T1-I-Gamma-i-Gamma-i} follows from a generalization of the mediant inequality, more precisely,
\begin{align*}
\frac{~\sum_{i=1}^t \textup{Rank}\big(T_1[I_{C_i}]\big)~}{~\sum_{i=1}^t |C_i|~} =\frac{~\sum_{i=1}^t \omega_i |C_i|~}{~\sum_{i=1}^t |C_i|~}\leq \max_{1\leq i\leq t}\, \omega_i,\quad \text{where $\omega_i\triangleq\frac{~\textup{Rank}\big(T_1[I_{C_i}]\big)~}{~|C_i|~}$, $1\leq i\leq t$};
\end{align*}
and the inequality~\eqref{max-Gamma-rank-Gamma-T1-Gamma-leq-max-Gamma-Rank-T1-I-Gamma} follows because for each strong partition $\mathcal{P}_{C}=\{C_1,C_2,\cdots,C_t\}$ of $C$, the following inequality holds:
\[
\max_{1\leq i\leq t} \frac{~\textup{Rank}\big(T_1[I_{C_i}]\big)~}{~|C_i|~} \leq  \max\limits_{C\in\Lambda(\mathcal{N})} \frac{~\textup{Rank}\big(T_1[I_{C}]\big)~}{~|C|~}.
\]
Combining \eqref{equ2} and \eqref{max-Gamma-rank-Gamma-T1-Gamma-leq-max-Gamma-Rank-T1-I-Gamma}, we have specified the lower bound \eqref{3-m-Omega-T-geq-max-best-known-upper-bound} for the model $\big(3,m,\Omega,T_1\big)$ as follows:
\begin{align*}
\max\limits_{C\in\Lambda(\mathcal{N})}\, \max_{ \substack{\textup{all strong partitions}\\ \mathcal{P}_C\,\textup{of}\,C} }  \frac{~\textbf{\textup{rank}}_{\mathcal{P}_C}(T_1)~}{~|C|~}=\max\limits_{C\in\Lambda(\mathcal{N})} \frac{~\textup{Rank}\big(T_1[I_{C}]\big)~}{~|C|~}.
\end{align*}

Furthermore, we claim that
\begin{align}\label{C-Lambda-N-Rank-T1-I-C-eq-C-In-rho-Rank-T1-I-C-C}
\max\limits_{C\in\Lambda(\mathcal{N})} \frac{~\textup{Rank}\big(T_1[I_{C}]\big)~}{~|C|~}
=\max\limits_{C\in\Lambda(\mathcal{N})\,\textup{s.t.}\,C\subseteq\textup{In}(\rho)} \frac{~\textup{Rank}\big(T_1[I_{C}]\big)~}{~|C|~}.
\end{align}
This claim can be proved by the same argument for proving \eqref{id-max-I-C-C-eq-max-In-rho-I-C-C}. Now, we rewrite the RHS of \eqref{C-Lambda-N-Rank-T1-I-C-eq-C-In-rho-Rank-T1-I-C-C} as
\begin{align}
\max\limits_{C\in\Lambda(\mathcal{N})\,\textup{s.t.}\,C\subseteq\textup{In}(\rho)} \frac{~\textup{Rank}\big(T_1[I_{C}]\big)~}{~|C|~}&=
\max\limits_{C\in\Lambda(\mathcal{N})\,\textup{s.t.}\,C\subseteq\textup{In}(\rho)} \frac{~\textup{Rank}\big(T_1[I_{C}]\big)~}{~\big|\{\textup{tail}(e):e\in C\}\big|~} \notag \\
&=\max\limits_{\Gamma\subseteq V} \frac{~\textup{Rank}\big(T_1[I_{\Gamma}]\big)~}{~|\Gamma|~}, \label{C-In-rho-Rank-T1-I-C-C-eq-Gamma-V-Rank-T1-I-Gamma-Gamma}
\end{align}
where the equality \eqref{C-In-rho-Rank-T1-I-C-C-eq-Gamma-V-Rank-T1-I-Gamma-Gamma} holds because by the definition of $I_{\Gamma}$ in \eqref{T1-def-I-Gamma}, we have $I_{\Gamma}=I_C$ if $C=\{(v,\rho):v\in\Gamma\}$.
We have thus proved the converse part, i.e.,
\begin{equation}\label{cap-3-m-T1-geq-rank-T1-I-Gamma-Gamma}
\mathcal{C}\big(3,m,\Omega,T_1\big)
\geq \max\limits_{\Gamma\subseteq V} \frac{~\textup{Rank}\big(T_1[I_{\Gamma}]\big)~}{~|\Gamma|~}.
\end{equation}

\medskip

For the notational simplicity in the rest of the proof, we denote by $R_1$ the lower bound in \eqref{cap-3-m-T1-geq-rank-T1-I-Gamma-Gamma}, i.e.,
\begin{align}\label{def-R1}
R_1\triangleq\max\limits_{\Gamma\subseteq V} \frac{~\textup{Rank}\big(T_1[I_{\Gamma}]\big)~}{~|\Gamma|~}.
\end{align}
It now remains to prove the achievability of $R_1$, namely that for each $1\leq m\leq 3$ and an arbitrary connectivity state $\Omega$,
$\mathcal{C}\big(3,m,\Omega,T_1\big)\leq R_1$.
Toward this end, we will construct an admissible $k$-shot source code $\mathbf{C}$ for $\big(3,m,\Omega,T_1\big)$ such that the coding rate $R(\mathbf{C})\leq R_1$, and thus $\mathcal{C}\big(3,m,\Omega,T_1\big)\leq R(\mathbf{C})\leq R_1$.

\vspace{-0.5em}

\subsection{The Case of $m=1$}

For the case of $m=1$, we have $V=\{v_1\}$ and the unique connectivity state $\Omega=\big(\Gamma_{\sigma_i}=\{v_1\}:~ i=1,2,3\big)$ as depicted in Fig.\,\ref{T1-m=1}. By \eqref{def-R1} we can compute
\[
R_1=\frac{~\textup{Rank}\big(T_1\big)~}{~|\{v_1\}|~}=2,
\]
and so $\mathcal{C}\big(3,m,\Omega,T_1\big)\geq R_1=2$. On the other hand, in Fig.\,\ref{T1-m=1} we  present an admissible $1$-shot source code with $n(\mathbf{C})=2$ and hence the coding rate $R(\mathbf{C})=n(\mathbf{C})/k=2$. This implies that $\mathcal{C}(3,1,\Omega,T_1)\leq2$. We have thus proved that $\mathcal{C}(3,1,\Omega,T_1)=R_1=2$.

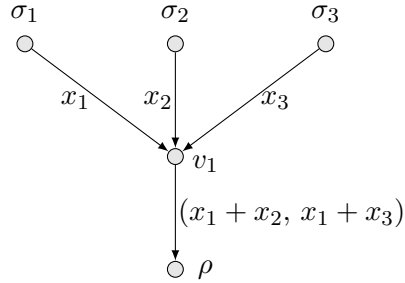
\begin{figure}[htbp]
\vspace{-1em}
	\centering

\tikzstyle{format}=[draw,circle,fill=gray!20, minimum size=6pt, inner sep=0pt]

	\begin{tikzpicture}

		\node[format](a1)at(6,4){};
        \node at (6,4.4) {$\sigma_1$};
		\node[format](a2)at(8,4){};
        \node at (8,4.4) {$\sigma_2$};
		\node[format](a3)at(10,4){};
        \node at (10,4.4) {$\sigma_3$};

		\node[format](r1)at(8,2.5){};
        \node at (8.4,2.4) {$v_1$};

		\node[format](r)at(8,1){};
        \node at (8.4,1) {$\rho$};

 \draw[->,>=latex](a1)--(r1) node[midway, auto,swap, left=0mm] {$x_1$};
  \draw[->,>=latex](a2)--(r1) node[midway, auto,swap, left=-1mm] {$x_2$};
 \draw[->,>=latex](a3)--(r1) node[midway, auto,swap, right=0mm] {$x_3$};
 \draw[->,>=latex](r1)--(r) node[midway, auto,swap, right=-1mm] {$(x_1+x_2,\,x_1+x_3)$};
	\end{tikzpicture}
 \vspace{-1em}
\caption{The model $(3,1,\Omega,T_1)$.}
\label{T1-m=1}
 \vspace{-3em}
\end{figure}

\subsection{The Case of $m=2$}\label{subsec-T1-m=2}

For the case of $m=2$, we write $V=\{v_1,v_2\}$. We divide all connectivity states into two classes below and consider them respectively:
\begin{itemize}
\item \textbf{Class 1:} $\Omega=\big(\Gamma_{\sigma_1},\Gamma_{\sigma_2},\Gamma_{\sigma_3}\big)$ such that $\big|\Gamma_{\{\sigma_i,\sigma_j\}}\big|=1$ for some two-index set $\{i,j\}\subseteq[3]$;\footnote{Here, we let $\Gamma_I\triangleq\cup_{\sigma_i\in I}\, \Gamma_{\sigma_i}$ for a source subset $I\subseteq S$.}
\item \textbf{Class 2:} $\Omega=\big(\Gamma_{\sigma_1},\Gamma_{\sigma_2},\Gamma_{\sigma_3}\big)$ such that $\Gamma_{\{\sigma_i,\sigma_j\}}=V$ for all two-index sets $\{i,j\}\subseteq[3]$.
\end{itemize}

\smallskip

We first consider connectivity states in \textbf{Class 1}. By \eqref{def-R1}, we compute the lower bound $R_1=2$.
To see the achievability of the lower bound $R_1=2$ for all connectivity states in \textbf{Class 1}, we consider the following three subcases.


\textbf{Case 1A:} The connectivity state $\Omega$ with $\big|\Gamma_{\{\sigma_1,\sigma_2\}}\big|=1$.

By the isomorphism of models, we assume without loss of generality that $\Gamma_{\{\sigma_1,\sigma_2\}}=\{v_1\}$. To be specific, if $\Gamma_{\{\sigma_1,\sigma_2\}}=\{v_2\}$ in $\Omega$, we take the permutations $\pi$ to be the identity permutation for $S$ and $\tau=\setlength{\arraycolsep}{3pt}\renewcommand{\arraystretch}{0.7}
\begin{pmatrix}
1 & 2 \\
2 & 1
\end{pmatrix}$ for $V$. By performing $(\pi,\tau)$ to $\Omega$, we immediately obtain that $\Gamma_{\{\sigma_1,\sigma_2\}}=\{v_1\}$ in the connectivity state $\Omega\circ(\pi,\tau)$. Together with Corollary \ref{T-eq-cap=permutation} and the observation that $T_1\circ\pi=T_1$, the models $(3,2,\Omega,T_1)$ and $\big(3,\,2,\,\Omega\circ(\pi,\tau),\,T_1\big)$ are isomorphic, and so both models have the identical capacity. It follows from $\Gamma_S=V$ that $\sigma_3\rightarrow v_2$, i.e., $v_2\in\Gamma_{\sigma_3}$. With $\Gamma_{\{\sigma_1,\sigma_2\}}=\{v_1\}$, we let
\[
\Omega^*=\big(\Gamma_{\sigma_1}=\{v_1\},\Gamma_{\sigma_2}=\{v_1\},\Gamma_{\sigma_3}=\{v_2\}\big).
\]
See Fig.\,\ref{T1-m=2-xi=2-fig1} for $\Omega^*$. It is not difficult to see that $\Omega^*$ is a connectivity state in this subcase and satisfies $\Omega^*\preceq\Omega$ for any connectivity state $\Omega$ in this subcase. In other words, $\Omega^*$ is the ``minimum'' connectivity state in this subcase. By Lemma \ref{lem-cap-omega'-geq-cap-omega}, this immediately implies that $\mathcal{C}(3,2,\Omega^*,T_1)\geq\mathcal{C}(3,2,\Omega,T_1)$ for any connectivity state $\Omega$ in this subcase. Together with the admissible $1$-shot source code $\mathbf{C}$ for $(3,2,\Omega^*,T_1)$ shown in Fig.\,\ref{T1-m=2-xi=2-fig1} of which the coding rate is $2$, we have thus proved that $\mathcal{C}(3,2,\Omega^*,T_1)\leq2$. Together with the lower bound $R_1=2$, we have $\mathcal{C}(3,2,\Omega,T_1)=2$ for any connectivity state $\Omega$ in this subcase.

\medskip

\textbf{Case 1B:} The connectivity state $\Omega$ with $\big|\Gamma_{\{\sigma_1,\sigma_3\}}\big|=1$.

By the isomorphism of models, we also assume without loss of generality that $\Gamma_{\{\sigma_1,\sigma_3\}}=\{v_1\}$. By $\Gamma_S=V$, we have $v_2\in\Gamma_{\sigma_2}$ and let
$\Omega^*=\big(\Gamma_{\sigma_1}=\{v_1\},\Gamma_{\sigma_2}=\{v_2\},
\Gamma_{\sigma_3}=\{v_1\}\big)$. See Fig.\,\ref{T1-m=2-xi=2-fig2} for~$\Omega^*$. In addition, we see that $\Omega^*$ is a connectivity state in this subcase and $\Omega^*\preceq\Omega$ for any connectivity state $\Omega$ in this subcase, i.e., $\Omega^*$ is the ``minimum'' connectivity state in this subcase. By Lemma \ref{lem-cap-omega'-geq-cap-omega}, we obtain that $\mathcal{C}(3,2,\Omega^*,T_1)\geq\mathcal{C}(3,2,\Omega,T_1)$ for any connectivity state $\Omega$ in this subcase. We also present an admissible $1$-shot source code $\mathbf{C}$ for $(3,2,\Omega^*,T_1)$ in Fig.\,\ref{T1-m=2-xi=2-fig2} of which the coding rate is $2$, implying that $\mathcal{C}(3,2,\Omega^*,T_1)\leq2$. Together with the lower bound $R_1=2$, we have proved that $\mathcal{C}(3,2,\Omega,T_1)=2$ for any connectivity state $\Omega$ in this subcase.

\medskip

\textbf{Case 1C:} The connectivity state $\Omega$ with $\big|\Gamma_{\{\sigma_2,\sigma_3\}}\big|=1$.

Consider an arbitrary connectivity state $\Omega$ with $\big|\Gamma_{\{\sigma_2,\sigma_3\}}\big|=1$. We take the permutation $\pi=\setlength{\arraycolsep}{3pt}\renewcommand{\arraystretch}{0.7}
\begin{pmatrix}
1 & 2 & 3 \\
2 & 1 & 3
\end{pmatrix}$ for $S$ and the identity permutation $\tau$ for $V$. Performing $(\pi,\tau)$ to $\Omega$, we have $1=\big|\Gamma_{\{\sigma_{\pi(2)},\sigma_{\pi(3)}\}}\big|=\big|\Gamma_{\{\sigma_1,\sigma_3\}}\big|$ in the connectivity state $\Omega\circ(\pi,\tau)$. Then, $\Omega\circ(\pi,\tau)$ is a connectivity state in \textbf{Case 1B}, and so $\mathcal{C}\big(3,\,2,\,\Omega\circ(\pi,\tau),\,T_1\big)=2$ (See \textbf{Case 1B}).

Further, we note that
\[
x_S\cdot T_1=(x_1+x_3,x_2+x_3)\quad \textup{and}\quad x_S\cdot T_1\circ\pi=x_{\pi(S)}\cdot T_1=(x_2+x_3,x_1+x_3),
\]
both of which are identical. This shows that $\big(3,\,2,\,\Omega\circ(\pi,\tau),\,T_1\big)$ and $\big(3,\,2,\,\Omega\circ(\pi,\tau),\,T_1\circ\pi\big)$ are identical, and then
\[
\mathcal{C}\big(3,\,2,\,\Omega\circ(\pi,\tau),\,T_1\big)
=\mathcal{C}\big(3,\,2,\,\Omega\circ(\pi,\tau),\,T_1\circ\pi\big).
\]
By Corollary \ref{T-eq-cap=permutation}, we obtain that
\[
\mathcal{C}(3,2,\Omega,T_1)=\mathcal{C}\big(3,\,2,\,\Omega\circ(\pi,\tau),\,T_1\circ\pi\big)
=\mathcal{C}\big(3,\,2,\,\Omega\circ(\pi,\tau),\,T_1\big)=2.
\]

\begin{figure*}[t]
\tikzstyle{vertex}=[draw,circle,fill=gray!30,minimum size=6pt, inner sep=0pt]
\tikzstyle{vertex1}=[draw,circle,fill=gray!80,minimum size=6pt, inner sep=0pt]
\centering
\begin{minipage}[b]{0.5\textwidth}
\centering
{
 \begin{tikzpicture}[x=0.6cm]
    \node[draw,circle,fill=gray!30,minimum size=6pt, inner sep=0pt](a1)at(1,4){};
        \node at (1,4.4) {$\sigma_1$};
        \node[draw,circle,fill=gray!30,minimum size=6pt, inner sep=0pt](a2)at(4,4){};
        \node at (4,4.4) {$\sigma_2$};
        \node[draw,circle,fill=gray!30,minimum size=6pt, inner sep=0pt](a3)at(7,4){};
        \node at (7,4.4) {$\sigma_3$};

        \node[draw,circle,fill=gray!30,minimum size=6pt, inner sep=0pt](r1)at(2.5,2.5){};
        \node[draw,circle,fill=gray!30,minimum size=6pt, inner sep=0pt](r2)at(5.5,2.5){};
        \node at (2,2.4) {$v_1$};
        \node at (6,2.4) {$v_2$};

        \node[draw,circle,fill=gray!30,minimum size=6pt, inner sep=0pt](r)at(4,1){};
        \node at (4,0.6) {$\rho$};

        \draw[->,>=latex](a1)--(r1) node[midway, auto,swap, left=0mm] {$x_1$};
        \draw[->,>=latex](a2)--(r1) node[midway, auto,swap, right=0mm] {$x_2$};
        \draw[->,>=latex](a3)--(r2) node[midway, auto,swap, right=0mm] {$x_3$};
        \draw[->,>=latex](r1)--(r) node[midway, auto,swap, left=0mm] {$(x_1,x_2)$};
        \draw[->,>=latex](r2)--(r) node[midway, auto,swap, right=0mm] {$x_3$};
    \end{tikzpicture}
}
 \vspace{-1em}
 \caption{The model $(3,2,\Omega^*,T_1)$\\ \qquad with $\Omega^*$ in \textbf{Case 1A}.}
    \label{T1-m=2-xi=2-fig1}
\end{minipage}%
\centering
\begin{minipage}[b]{0.5\textwidth}
\centering
 \begin{tikzpicture}[x=0.6cm]
   \node[draw,circle,fill=gray!30,minimum size=6pt, inner sep=0pt](a1)at(1,4){};
        \node at (1,4.4) {$\sigma_1$};
        \node[draw,circle,fill=gray!30,minimum size=6pt, inner sep=0pt](a2)at(4,4){};
        \node at (4,4.4) {$\sigma_2$};
        \node[draw,circle,fill=gray!30,minimum size=6pt, inner sep=0pt](a3)at(7,4){};
        \node at (7,4.4) {$\sigma_3$};

        \node[draw,circle,fill=gray!30,minimum size=6pt, inner sep=0pt](r1)at(2.5,2.5){};
        \node[draw,circle,fill=gray!30,minimum size=6pt, inner sep=0pt](r2)at(5.5,2.5){};
        \node at (2,2.4) {$v_1$};
        \node at (6,2.4) {$v_2$};

        \node[draw,circle,fill=gray!30,minimum size=6pt, inner sep=0pt](r)at(4,1){};
        \node at (4,0.6) {$\rho$};

        \draw[->,>=latex](a1)--(r1) node[midway, auto,swap, left=0mm] {$x_1$};
        \draw[->,>=latex](a2)--(r2); \node at (4.1,3.5) {$x_2$};
        \draw[->,>=latex](a3)--(r1); \node at (6.2,3.5) {$x_3$};
        \draw[->,>=latex](r1)--(r) node[midway, auto,swap, left=0mm] {$(x_1,x_3)$};
        \draw[->,>=latex](r2)--(r) node[midway, auto,swap, right=0mm] {$x_2$};
    \end{tikzpicture}
     \vspace{-1em}
    \caption{The model $(3,2,\Omega^*,T_1)$\\ \qquad with $\Omega^*$ in \textbf{Case 1B}.}
    \label{T1-m=2-xi=2-fig2}
\end{minipage}
  \vspace{-4em}
\end{figure*}


Next, we consider connectivity states in \textbf{Class 2}, i.e., $\Gamma_{\{\sigma_i,\sigma_j\}}=V=\{v_1,v_2\}$ for all two-index sets $\{i,j\}\subseteq[3]$. By \eqref{def-R1}, we first compute the lower bound
\begin{align*}
R_1=
\frac{~\textup{Rank}\big(T_1\big)~}{~|V|~}=1.
\end{align*}

Consider an arbitrary connectivity state $\Omega$ in \textbf{Class 2}. We claim that in $\Omega$, $\big|\Gamma_{\sigma_i}\big|=2$, or equivalently, $\Gamma_{\sigma_i}=V$ for some $\sigma_i\in S$. Otherwise, assume that $\big|\Gamma_{\sigma_1}\big|\mspace{-4mu}=\mspace{-4mu}\big|\Gamma_{\sigma_2}\big|\mspace{-4mu}
=\mspace{-4mu}\big|\Gamma_{\sigma_3}\big|\mspace{-4mu}=\mspace{-3mu}1$. Since $|V|\mspace{-4mu}=\mspace{-2mu}m\mspace{-2mu}=\mspace{-2mu}2$, there exist two sources, say $\sigma_i$ and $\sigma_j$, such that $\big|\Gamma_{\{\sigma_i,\sigma_j\}}\big|=\mspace{-4mu}1$, a contradiction to the condition of $\big|\Gamma_{\{\sigma_i,\sigma_j\}}\big|=2$ for all two-index sets $\{i,j\}\subseteq[3]$ for \textbf{Class 2}. Now, we consider three subcases of $\Gamma_{\sigma_1}=V$, $\Gamma_{\sigma_2}=V$ and $\Gamma_{\sigma_3}=V$ as follows.

\medskip

\textbf{Case 2A:} $\Gamma_{\sigma_1}=V$ for the connectivity state $\Omega$.

For this subcase, we can readily see that $\Gamma_{\sigma_1}\mspace{-4mu}=\mspace{-3mu}\Gamma_{\{\sigma_2,\sigma_3\}}\mspace{-4mu}
=\mspace{-3.5mu}V\mspace{-3.5mu}=\mspace{-4mu}\{v_1,v_2\}$. By $\Gamma_{\{\sigma_2,\sigma_3\}}=V$, either $\sigma_2\rightarrow v_1$ and $\sigma_3\rightarrow v_2$ or $\sigma_2\rightarrow v_2$ and $\sigma_3\rightarrow v_1$. By the isomorphism of models, we assume without loss of generality that $\sigma_2\rightarrow v_1$ and $\sigma_3\rightarrow v_2$, i.e., $v_1\in\Gamma_{\sigma_2}$ and $v_2\in\Gamma_{\sigma_3}$.\footnote{If $\sigma_2\rightarrow v_2$ and $\sigma_3\rightarrow v_1$ in $\Omega$, we can take the permutations $\pi$ to be the identity permutation for $S$ and $\tau=\setlength{\arraycolsep}{3pt}\renewcommand{\arraystretch}{0.7}
\begin{pmatrix}
1 & 2 \\
2 & 1
\end{pmatrix}$ for~$V$. By performing $(\pi,\tau)$ to $\Omega$, we immediately obtain that $\sigma_2\rightarrow v_1$ and $\sigma_3\rightarrow v_2$ in the connectivity state $\Omega\circ(\pi,\tau)$. Together with $T_1\circ\pi=T_1$, we obtain that the models $(3,2,\Omega,T_1)$ and $\big(3,\,2,\,\Omega\circ(\pi,\tau),\,T_1\big)$ are isomorphic.} Let
\[
\Omega^*=\big(\Gamma_{\sigma_1}=V,\,\Gamma_{\sigma_2}=\{v_1\},\Gamma_{\sigma_3}=\{v_2\}\big).
\]
See Fig.\,\ref{T1-m=2-xi=1-fig1}. It is not difficult to see that $\Omega^*$ is the ``minimum'' connectivity state in this subcase. By Lemma~\ref{lem-cap-omega'-geq-cap-omega}, we obtain that for any connectivity state $\Omega$ in this subcase,
$\mathcal{C}(3,2,\Omega^*,T_1)\geq\mathcal{C}(3,2,\Omega,T_1)$.
We further present in Fig.\,\ref{T1-m=2-xi=1-fig1} an admissible $1$-shot source code $\mathbf{C}$ for $(3,2,\Omega^*,T_1)$  of which the coding rate is $1$. Thus, we have proved that $\mathcal{C}(3,2,\Omega^*,T_1)\leq1$. Together with the lower bound $R_1=1$, this implies that $\mathcal{C}(3,2,\Omega,T_1)=1$ for any connectivity state $\Omega$ in this subcase.

\medskip

\textbf{Case 2B:} $\Gamma_{\sigma_2}=V$ for the connectivity state $\Omega$.

For this subcase, we have $\Gamma_{\sigma_2}=\Gamma_{\{\sigma_1,\sigma_3\}}
=V=\{v_1,v_2\}$. Consider the permutation $\pi=\setlength{\arraycolsep}{3pt}\renewcommand{\arraystretch}{0.7}
\begin{pmatrix}
1 & 2 & 3 \\
2 & 1 & 3
\end{pmatrix}$ for $S$ and the identity permutation $\tau$ for $V$. In the connectivity state $\Omega\circ(\pi,\tau)$, we have
\[
\Gamma_{\sigma_{\pi(2)}}=\Gamma_{\{\sigma_{\pi(1)},\sigma_{\pi(3)}\}}=V,\quad \textup{i.e.}, \quad \Gamma_{\sigma_1}=\Gamma_{\{\sigma_2,\sigma_3\}}=V.
\]
Then $\Omega\circ(\pi,\tau)$ is a connectivity state in \textbf{Case 2A} and thus $\mathcal{C}\big(3,\,2,\,\Omega\circ(\pi,\tau),\,T_1\big)=1$.

Further, we see that
\[
x_S\cdot T_1=(x_1+x_3,x_2+x_3) \quad \textup{and} \quad x_S\cdot T_1\circ\pi=x_{\pi(S)}\cdot T_1=(x_2+x_3,x_1+x_3),
\]
both of which are identical. So the models $\big(3,\,2,\,\Omega\circ(\pi,\tau),\,T_1\big)$ and $\big(3,\,2,\,\Omega\circ(\pi,\tau),\,T_1\circ\pi\big)$ are identical, and
$
\mathcal{C}\big(3,\,2,\,\Omega\circ(\pi,\tau),\,T_1\big)
=\mathcal{C}\big(3,\,2,\,\Omega\circ(\pi,\tau),\,T_1\circ\pi\big)$.
Thus, we obtain that
\[
\mathcal{C}(3,2,\Omega,T_1)=\mathcal{C}\big(3,\,2,\,\Omega\circ(\pi,\tau),\,T_1\circ\pi\big)
=\mathcal{C}\big(3,\,2,\,\Omega\circ(\pi,\tau),\,T_1\big)=1.
\]

\textbf{Case 2C:} $\Gamma_{\sigma_3}=V$ for the connectivity state $\Omega$.

Similar to \textbf{Case 2A}, we have $\Gamma_{\sigma_3}\mspace{-4mu}=\mspace{-3mu}\Gamma_{\{\sigma_1,\sigma_2\}}\mspace{-4mu}
=\mspace{-3.5mu}V\mspace{-3.5mu}=\mspace{-4mu}\{v_1,v_2\}$. Hence, either $\sigma_1\rightarrow v_1$ and $\sigma_2\rightarrow v_2$ or $\sigma_1\rightarrow v_2$ and $\sigma_2\rightarrow v_1$. By the isomorphism of models, we assume without loss of generality that $\sigma_1\rightarrow v_1$ and $\sigma_2\rightarrow v_2$, i.e., $v_1\in\Gamma_{\sigma_1}$ and $v_2\in\Gamma_{\sigma_2}$. Let
$
\Omega^*=\big(\Gamma_{\sigma_1}=\{v_1\},\Gamma_{\sigma_2}=\{v_2\},
\Gamma_{\sigma_3}=V\big).
$
See Fig.\,\ref{T1-m=2-xi=1-fig2}. Here, $\Omega^*$ is a connectivity state in this subcase and $\Omega^*\preceq\Omega$ for any connectivity state $\Omega$ in this subcase. By Lemma \ref{lem-cap-omega'-geq-cap-omega}, we further have
\[
\mathcal{C}(3,2,\Omega^*,T_1)\geq\mathcal{C}(3,2,\Omega,T_1),\quad \forall~ \Omega~ \textup{in this subcase}.
\]
Together with an admissible $1$-shot source code $\mathbf{C}$ for $(3,2,\Omega^*,T_1)$ with coding rate $1$ in Fig.\,\ref{T1-m=2-xi=1-fig2}, we obtain that $\mathcal{C}(3,2,\Omega^*,T_1)\leq1$. Together with the lower bound $R_1=1$, we have proved that $\mathcal{C}(3,2,\Omega,T_1)=1$ for any connectivity state $\Omega$ in this subcase.

\begin{figure*}[t]
\tikzstyle{vertex}=[draw,circle,fill=gray!30,minimum size=6pt, inner sep=0pt]
\tikzstyle{vertex1}=[draw,circle,fill=gray!80,minimum size=6pt, inner sep=0pt]
\centering
\begin{minipage}[b]{0.5\textwidth}
\centering
{
 \begin{tikzpicture}[x=0.6cm]
     \node[draw,circle,fill=gray!30,minimum size=6pt, inner sep=0pt](a1)at(1,4){};
        \node at (1,4.4) {$\sigma_1$};
        \node[draw,circle,fill=gray!30,minimum size=6pt, inner sep=0pt](a2)at(4,4){};
        \node at (4,4.4) {$\sigma_2$};
        \node[draw,circle,fill=gray!30,minimum size=6pt, inner sep=0pt](a3)at(7,4){};
        \node at (7,4.4) {$\sigma_3$};

        \node[draw,circle,fill=gray!30,minimum size=6pt, inner sep=0pt](r1)at(2.5,2.5){};
        \node[draw,circle,fill=gray!30,minimum size=6pt, inner sep=0pt](r2)at(5.5,2.5){};
        \node at (2,2.4) {$v_1$};
        \node at (6,2.4) {$v_2$};

        \node[draw,circle,fill=gray!30,minimum size=6pt, inner sep=0pt](r)at(4,1){};
        \node at (4,0.6) {$\rho$};

        \draw[->,>=latex](a1)--(r1) node[midway, auto,swap, left=0mm] {$x_1$};
        \draw[->,>=latex](a1)--(r2); \node at (2.2,3.8) {$x_1$};
        \draw[->,>=latex](a2)--(r1); \node at (3.9,3.5) {$x_2$};
        \draw[->,>=latex](a3)--(r2) node[midway, auto,swap, right=0mm] {$x_3$};
        \draw[->,>=latex](r1)--(r) node[midway, auto,swap, left=-1mm] {$x_2-x_1$};
        \draw[->,>=latex](r2)--(r) node[midway, auto,swap, right=0mm] {$x_1+x_3$};
    \end{tikzpicture}
}
 \vspace{-1em}
 \caption{The model $(3,2,\Omega^*,T_1)$\\ \qquad with $\Omega^*$ in \textbf{Case 2A}.}
    \label{T1-m=2-xi=1-fig1}
\end{minipage}%
\centering
\begin{minipage}[b]{0.5\textwidth}
\centering
 \begin{tikzpicture}[x=0.6cm]
   \node[draw,circle,fill=gray!30,minimum size=6pt, inner sep=0pt](a1)at(1,4){};
        \node at (1,4.4) {$\sigma_1$};
        \node[draw,circle,fill=gray!30,minimum size=6pt, inner sep=0pt](a2)at(4,4){};
        \node at (4,4.4) {$\sigma_2$};
        \node[draw,circle,fill=gray!30,minimum size=6pt, inner sep=0pt](a3)at(7,4){};
        \node at (7,4.4) {$\sigma_3$};

        \node[draw,circle,fill=gray!30,minimum size=6pt, inner sep=0pt](r1)at(2.5,2.5){};
        \node[draw,circle,fill=gray!30,minimum size=6pt, inner sep=0pt](r2)at(5.5,2.5){};
        \node at (2,2.4) {$v_1$};
        \node at (6,2.4) {$v_2$};

        \node[draw,circle,fill=gray!30,minimum size=6pt, inner sep=0pt](r)at(4,1){};
        \node at (4,0.6) {$\rho$};

        \draw[->,>=latex](a1)--(r1) node[midway, auto,swap, left=0mm] {$x_1$};
        \draw[->,>=latex](a2)--(r2); \node at (4,3.5) {$x_2$};
        \draw[->,>=latex](a3)--(r1); \node at (5.8,3.8) {$x_3$};
        \draw[->,>=latex](a3)--(r2) node[midway, auto,swap, right=0mm] {$x_3$};
        \draw[->,>=latex](r1)--(r) node[midway, auto,swap, left=0mm] {$x_1+x_3$};
        \draw[->,>=latex](r2)--(r) node[midway, auto,swap, right=0mm] {$x_2+x_3$};
    \end{tikzpicture}
     \vspace{-1em}
    \caption{The model $(3,2,\Omega^*,T_1)$\\ \qquad with $\Omega^*$ in \textbf{Case 2C}.}
    \label{T1-m=2-xi=1-fig2}
\end{minipage}
 \vspace{-4em}
\end{figure*}

\vspace{-0.5em}

\subsection{The case of $m=3$}\label{subsec-T1-m=3}

For the case of $m=3$, we write $V=\{v_1,v_2,v_3\}$. We divide all connectivity states into two classes:
\begin{itemize}
\item \textbf{Class 1:} $\Omega=\big(\Gamma_{\sigma_1},\Gamma_{\sigma_2},\Gamma_{\sigma_3}\big)$ such that $\big|\Gamma_{\{\sigma_i,\sigma_j\}}\big|=1$ for some two-index set $\{i,j\}\subseteq[3]$;
\item \textbf{Class 2:} $\Omega=\big(\Gamma_{\sigma_1},\Gamma_{\sigma_2},\Gamma_{\sigma_3}\big)$ such that $\big|\Gamma_{\{\sigma_i,\sigma_j\}}\big|\geq2$ for all two-index sets $\{i,j\}\subseteq[3]$.
\end{itemize}

\smallskip

We first consider connectivity states $\Omega$  in \textbf{Class 1}. By \eqref{def-R1}, we compute the lower bound $R_1=2$.
Let~$\Omega$ be an arbitrary connectivity state in \textbf{Class 1}. Without loss of generality, we assume that $\Gamma_{\{\sigma_i,\sigma_j\}}=\{v_1\}$ in $\Omega$ by the isomorphism of models, and then we have $\sigma_h\rightarrow v_2$ and $\sigma_h\rightarrow v_3$, i.e., $\{v_2,v_3\}\subseteq\Gamma_{\sigma_h}$ by $\Gamma_S=V$, where $\{i,j,h\}=[3]$. Recall \textbf{Class 1} in the case of $m=2$ in Section \ref{subsec-T1-m=2}, in which for the connectivity state $\widehat{\Omega}=\big( \Gamma_{\sigma_i}=\Gamma_{\sigma_j}=\{v_1\},\,\Gamma_{\sigma_h}=\{v_2\} \big)$, we have $\mathcal{C}(3,2,\widehat{\Omega},T_1)=2$. Comparing~$\widehat{\Omega}$ with~$\Omega$, $\widehat{\Omega}$ is in fact a subgraph of $\Omega$. This implies that $\mathcal{C}(3,3,\Omega,T_1)\leq\mathcal{C}(3,2,\widehat{\Omega},T_1)=2$. Together with the lower bound $R_1=2$, 
i.e., $\mathcal{C}(3,3,\Omega,T_1)\geq2$, we have proved that $\mathcal{C}(3,3,\Omega,T_1)=2$ for any connectivity state $\Omega$ in \textbf{Class 1}.

\smallskip

Next, we consider connectivity states in \textbf{Class 2}, in which we consider the following three cases.

\textbf{Case 2A:} $\Gamma_{\{\sigma_i,\sigma_j\}}=V$ for all two-index sets $\{i,j\}\subseteq[3]$ and $\big|\Gamma_{\sigma_i}\big|\geq2$ for all $i\in[3]$.

For this subcase, by \eqref{def-R1} we compute the lower bound
\begin{align}\label{T1-m=3-R1=2/3}
R_1=
\frac{~\textup{Rank}\big(T_1\big)~}{~|V|~}=\frac{2}{3}.
\end{align}
Consider an arbitrary connectivity state $\Omega$ in \textbf{Case 2A}. Then, there always exists a connectivity state~$\Omega^*$ with $\Gamma_{\{\sigma_i,\sigma_j\}}=V$ for all two-index sets $\{i,j\}\subseteq[3]$ and $\big|\Gamma_{\sigma_i}\big|=2$ for all $i\in[3]$ such that $\Omega^*\preceq\Omega$. By Lemma \ref{lem-cap-omega'-geq-cap-omega} which implies $\mathcal{C}(3,3,\Omega^*,T_1)\geq\mathcal{C}(3,3,\Omega,T_1)$, we only need to prove that $\mathcal{C}(3,3,\Omega^*,T_1)\leq R_1=2/3$ for any connectivity state $\Omega^*$ with $\Gamma_{\{\sigma_i,\sigma_j\}}=V$ for all two-index sets $\{i,j\}\subseteq[3]$ and $\big|\Gamma_{\sigma_i}\big|=2$ for all $i\in[3]$.

Consider such a connectivity state $\Omega^*=\big(\Gamma_{\sigma_1},\Gamma_{\sigma_2},\Gamma_{\sigma_3} \big)$. Since $\big|\Gamma_{\sigma_1}\big|=2$, we assume without loss of generality that $\Gamma_{\sigma_1}=\{v_1,v_2\}$ by the isomorphism of models. Together with $\Gamma_{\{\sigma_1,\sigma_2\}}=\Gamma_{\{\sigma_1,\sigma_3\}}=V$, we have $\sigma_2\rightarrow v_3$ and $\sigma_3\rightarrow v_3$, i.e., $v_3\in\Gamma_{\sigma_2}$ and $v_3\in\Gamma_{\sigma_3}$. Further, it follows from $\big|\Gamma_{\sigma_2}\big|=\big|\Gamma_{\sigma_3}\big|=2$ and $\Gamma_{\{\sigma_2,\sigma_3\}}=V$ that either $\sigma_2\rightarrow v_1$ and $\sigma_3\rightarrow v_2$ or $\sigma_2\rightarrow v_2$ and $\sigma_3\rightarrow v_1$, or equivalently, either $\Gamma_{\sigma_2}=\{v_1,v_3\}$ and $\Gamma_{\sigma_3}=\{v_2,v_3\}$ or $\Gamma_{\sigma_2}=\{v_2,v_3\}$ and $\Gamma_{\sigma_3}=\{v_1,v_3\}$. By the isomorphism of models, we only need to consider $\Gamma_{\sigma_2}=\{v_1,v_3\}$ and $\Gamma_{\sigma_3}=\{v_2,v_3\}$. As such, we consider the connectivity state~$\Omega^*$ as follows (see Fig.\,\ref{graph-T1-m=3-xi=2/3}):
\[
\Omega^*=\big(\Gamma_{\sigma_1}=\{v_1,v_2\},\Gamma_{\sigma_2}=\{v_1,v_3\},
\Gamma_{\sigma_3}=\{v_2,v_3\}\big).
\]

We construct an admissible $3$-shot source code $\mathbf{C}=\{\varphi_1,\varphi_2,\varphi_3;\,\psi\}$ for $(3,3,\Omega^*,T_1)$, where the coding rate $R(\mathbf{C})=n(\mathbf{C})/k=2/3$. Here, we let $\boldsymbol{x}_i=(x_{i,1},\,x_{i,2},\,x_{i,3})^\top$ for $1\leq i\leq 3$, and
\begin{align*}
&\varphi_1(\boldsymbol{x}_1,\boldsymbol{x}_2)=(x_{2,1}-x_{1,1},~ x_{1,2}-x_{2,2}),\\ &\varphi_2(\boldsymbol{x}_1,\boldsymbol{x}_3)=(x_{1,1}+x_{3,1},~ x_{1,3}+x_{3,3}), \\ &\varphi_3(\boldsymbol{x}_2,\boldsymbol{x}_3)=(x_{2,2}+x_{3,2},~ x_{2,3}+x_{3,3}).
\end{align*}
With the received messages $\varphi_1(\boldsymbol{x}_1,\boldsymbol{x}_2)$, $\varphi_2(\boldsymbol{x}_1,\boldsymbol{x}_3)$ and $\varphi_3(\boldsymbol{x}_2,\boldsymbol{x}_3)$ as above, we can compute at the decoder $\rho$
\[
\boldsymbol{x}_S\cdot T_1=
\begin{bmatrix}
\hspace{-0.1mm}x_{1,1}+x_{3,1} &x_{2,1}+x_{3,1}\\
\hspace{-0.1mm}x_{1,2}+x_{3,2} &x_{2,2}+x_{3,2}\\
\hspace{-0.1mm}x_{1,3}+x_{3,3} &x_{2,3}+x_{3,3}
\end{bmatrix}.
\]
We thus have proved that $\mathcal{C}(3,3,\Omega^*,T_1)\leq2/3$ and so $\mathcal{C}(3,3,\Omega,T_1)\leq2/3$ for any connectivity state $\Omega$ in \textbf{Case 2A}. Together with the lower bound $R_1=2/3$ in \eqref{T1-m=3-R1=2/3}, we have obtained that $\mathcal{C}(3,3,\Omega,T_1)=2/3$ for any $\Omega$ in \textbf{Case 2A}.


\begin{figure*}[t]
\tikzstyle{vertex}=[draw,circle,fill=gray!30,minimum size=6pt, inner sep=0pt]
\tikzstyle{vertex1}=[draw,circle,fill=gray!80,minimum size=6pt, inner sep=0pt]
\centering
\begin{minipage}[b]{0.5\textwidth}
\centering
{
 \begin{tikzpicture}[x=0.6cm]
     \node[draw,circle,fill=gray!30,minimum size=6pt, inner sep=0pt](a1)at(1,4){};
        \node at (1,4.4) {$\sigma_1$};
        \node[draw,circle,fill=gray!30,minimum size=6pt, inner sep=0pt](a2)at(4,4){};
        \node at (4,4.4) {$\sigma_2$};
        \node[draw,circle,fill=gray!30,minimum size=6pt, inner sep=0pt](a3)at(7,4){};
        \node at (7,4.4) {$\sigma_3$};

        \node[draw,circle,fill=gray!30,minimum size=6pt, inner sep=0pt](r1)at(1,2.5){};
        \node[draw,circle,fill=gray!30,minimum size=6pt, inner sep=0pt](r2)at(4,2.5){};
         \node[draw,circle,fill=gray!30,minimum size=6pt, inner sep=0pt](r3)at(7,2.5){};
        \node at (0.4,2.4) {$v_1$};
        \node at (3.4,2.4) {$v_2$};
        \node at (7.6,2.4) {$v_3$};

        \node[draw,circle,fill=gray!30,minimum size=6pt, inner sep=0pt](r)at(4,1){};
        \node at (4,0.6) {$\rho$};

        \draw[->,>=latex](a1)--(r1);
        \draw[->,>=latex](a1)--(r2);
        \draw[->,>=latex](a2)--(r1);
        \draw[->,>=latex](a2)--(r3);
        \draw[->,>=latex](a3)--(r2);
        \draw[->,>=latex](a3)--(r3);
        \draw[->,>=latex](r1)--(r) node[midway, auto,swap, left=0mm] {$\varphi_1$};
        \draw[->,>=latex](r2)--(r) node[midway, auto,swap, left=0mm] {$\varphi_2$};
        \draw[->,>=latex](r3)--(r) node[midway, auto,swap, right=0mm] {$\varphi_3$};
    \end{tikzpicture}
}
\vspace{-1em}
  \caption{The model $(3,3,\Omega^*,T_1)$\\ \qquad with $\Omega^*$ in \textbf{Case 2A}.}
        \label{graph-T1-m=3-xi=2/3}
\end{minipage}%
\centering
\begin{minipage}[b]{0.5\textwidth}
\centering
 \begin{tikzpicture}[x=0.6cm]
   \node[draw,circle,fill=gray!30,minimum size=6pt, inner sep=0pt](a1)at(1,4){};
        \node at (1,4.4) {$\sigma_1$};
        \node[draw,circle,fill=gray!30,minimum size=6pt, inner sep=0pt](a2)at(4,4){};
        \node at (4,4.4) {$\sigma_2$};
        \node[draw,circle,fill=gray!30,minimum size=6pt, inner sep=0pt](a3)at(7,4){};
        \node at (7,4.4) {$\sigma_3$};

        \node[draw,circle,fill=gray!30,minimum size=6pt, inner sep=0pt](r1)at(1,2.5){};
        \node[draw,circle,fill=gray!30,minimum size=6pt, inner sep=0pt](r2)at(4,2.5){};
         \node[draw,circle,fill=gray!30,minimum size=6pt, inner sep=0pt](r3)at(7,2.5){};
        \node at (0.4,2.4) {$v_1$};
        \node at (3.4,2.4) {$v_2$};
        \node at (7.6,2.4) {$v_3$};

        \node[draw,circle,fill=gray!30,minimum size=6pt, inner sep=0pt](r)at(4,1){};
        \node at (4,0.6) {$\rho$};

        \draw[->,>=latex](a1)--(r1) node[midway, auto,swap, left=0mm] {$x_1$};
        \draw[->,>=latex](a2)--(r2) node[midway, auto,swap, left=0mm] {$x_2$};
        \draw[->,>=latex](a3)--(r3) node[midway, auto,swap, right=0mm] {$x_3$};
        \draw[->,>=latex](r1)--(r) node[midway, auto,swap, left=0mm] {$x_1$};
        \draw[->,>=latex](r2)--(r) node[midway, auto,swap, left=0mm] {$x_2$};
        \draw[->,>=latex](r3)--(r) node[midway, auto,swap, right=0mm] {$x_3$};
    \end{tikzpicture}
    \vspace{-1em}
      \caption{The model $(3,3,\Omega^*,T_1)$ with\\ \qquad~~ $\Omega^*$ in \textbf{Case 2B} or \textbf{Case 2C}.}
   \label{graph-T1-m=3-xi=1}
\end{minipage}
 \vspace{-4em}
\end{figure*}

\medskip

\textbf{Case 2B:} $\big|\Gamma_{\{\sigma_i,\sigma_j\}}\big|\geq2$ for all two-index sets $\{i,j\}\subseteq[3]$ and $\big|\Gamma_{\{\sigma_i,\sigma_j\}}\big|=2$ for some two-index set $\{i,j\}\subseteq[3]$.

\textbf{Case 2C:} $\big|\Gamma_{\{\sigma_i,\sigma_j\}}\big|\geq2$ for all two-index sets $\{i,j\}\subseteq[3]$ and $\big|\Gamma_{\sigma_i}\big|=1$ for some $i\in[3]$.

We consider \textbf{Case 2B} and \textbf{Case 2C} together. First, for the two subcases, by \eqref{def-R1} we can compute the lower bound $R_1=1$. Consider an arbitrary connectivity state $\Omega$ in \textbf{Case 2B} or \textbf{Case 2C}. We see that there always exists a connectivity state $\Omega^*$ satisfying $\big|\Gamma_{\sigma_i}\big|=1$ for all $i\in[3]$ and $\big|\Gamma_{\{\sigma_i,\sigma_j\}}\big|=2$ for all two-index sets $\{i,j\}\subseteq[3]$ such that $\Omega^*\preceq\Omega$.

By Lemma \ref{lem-cap-omega'-geq-cap-omega}, we have $\mathcal{C}(3,3,\Omega^*,T_1)\geq\mathcal{C}(3,3,\Omega,T_1)$. So it suffices to prove that $\mathcal{C}(3,3,\Omega^*,T_1)\leq1$ for such a connectivity state $\Omega^*$. Let
$
\Omega^*=\big(\Gamma_{\sigma_1}=\{v_1\},
\Gamma_{\sigma_2}=\{v_2\},\Gamma_{\sigma_3}=\{v_3\}\big)
$
be such a connectivity state. See Fig.\,\ref{graph-T1-m=3-xi=1} for $\Omega^*$. This connectivity state $\Omega^*$ is the ``minimum'' connectivity state in \textbf{Case 2B} and \textbf{Case 2C} under the isomorphism of models.

In Fig.\,\ref{graph-T1-m=3-xi=1}, we present an admissible $1$-shot source code $\mathbf{C}$ for $(3,3,\Omega^*,T_1)$ of which the coding rate is~$1$. This implies that $\mathcal{C}(3,3,\Omega^*,T_1)\leq1$. Together with the lower bound $R_1=1$, we have proved that $\mathcal{C}(3,3,\Omega,T_1)=1$ for any connectivity state $\Omega$ in \textbf{Case 2B} or \textbf{Case 2C}.

\medskip

To end this section, we rewrite Theorem \ref{cap-T1} in a more precise way according to the above discussions.

\begin{theorem}\label{cap-T1-rewrite}
Consider a model $\big(3,m,\Omega,T_1\big)$, where $1\leq m\leq 3$ and $\Omega=\big(\Gamma_{\sigma_1},\Gamma_{\sigma_2},\Gamma_{\sigma_3}\big)$ is an arbitrary connectivity state.
\begin{itemize}
  \item For $m=1$, $\mathcal{C}(3,1,\Omega,T_1)=2$.
  \item For $m=2$,
  \begin{enumerate}
\item if in $\Omega$, there exists a two-index set $\{i,j\}\subseteq[3]$ such that $\big|\Gamma_{\{\sigma_i,\sigma_j\}}\big|=1$ \textup{(cf.~\textbf{Cases 1A}, \textbf{1B} and \textbf{1C} in Section \ref{subsec-T1-m=2})}, then
 $\mathcal{C}\big(3,2,\Omega,T_1\big)=2$;
\item otherwise, namely that in $\Omega$, $\Gamma_{\{\sigma_i,\sigma_j\}}=V$ for all two-index sets $\{i,j\}\subseteq[3]$ \textup{(cf.~\textbf{Cases 2A}, \textbf{2B} and \textbf{2C} in Section~\ref{subsec-T1-m=2})}, then
$\mathcal{C}\big(3,2,\Omega,T_1\big)=1$.
\end{enumerate}
  \item For $m=3$,
\begin{enumerate}
\item if in $\Omega$, there exists a two-index set $\{i,j\}\subseteq[3]$ such that $\big|\Gamma_{\{\sigma_i,\sigma_j\}}\big|=1$ \textup{(cf. \textbf{Class 1} in Section~\ref{subsec-T1-m=3})}, then
$\mathcal{C}\big(3,3,\Omega,T_1\big)=2$;
\item if in $\Omega$, $\Gamma_{\{\sigma_i,\sigma_j\}}=V$ for all two-index sets $\{i,j\}\subseteq[3]$ with $\big|\Gamma_{\sigma_i}\big|\geq2$ for all $i\in[3]$ \textup{(cf. \textbf{Case~2A} in Section \ref{subsec-T1-m=3})}, then
$\mathcal{C}\big(3,3,\Omega,T_1\big)=2/3$;
\item otherwise, namely that in $\Omega$, $\big|\Gamma_{\{\sigma_i,\sigma_j\}}\big|\geq2$ for all two-index sets $\{i,j\}\subseteq[3]$ and either $\big|\Gamma_{\{\sigma_i,\sigma_j\}}\big|=2$ for some two-index set $\{i,j\}\subseteq[3]$ or $\big|\Gamma_{\sigma_i}\big|=1$ for some $i\in[3]$ \textup{(cf. \textbf{Cases 2B} and \textbf{2C} in Section \ref{subsec-T1-m=3})}, then
$\mathcal{C}\big(3,3,\Omega,T_1\big)=1$.
\end{enumerate}
\end{itemize}
\end{theorem}

\section{Capacity Characterization for the Model $(3,m,\Omega,T_2)$}\label{section-3-m-T2}

In this section, we will fully characterize the function-compression capacities for all the models $(3,m,\Omega,T_2)$ with $1\leq m\leq3$ and arbitrary connectivity states~$\Omega$. We first specify the lower bound~\eqref{3-m-Omega-T-geq-max-best-known-upper-bound}  for the model $(3,m,\Omega,T_2)$.

\begin{lemma}\label{prop-lower-bound-for-T2}
Consider a model $\big(3,m,\Omega,T_2\big)$, where $1\leq m\leq 3$ and $\Omega$ is an arbitrary connectivity state.
\begin{itemize}
\item If $m=2$ and $\Omega$ satisfies $\Gamma_{\sigma_3}=V=\{v_1,v_2\}$ and $\Gamma_{\sigma_1}=\{v_i\}$, $\Gamma_{\sigma_2}=\{v_j\}$ for $\{i,j\}=\{1,2\}$, then
    \begin{align}\label{def-R2-3/2}
    \mathcal{C}\big(3,2,\Omega,T_2\big)\geq\frac{3}{2};
    \end{align}
\item Otherwise,
\begin{align}\label{T2-geq-def-R2}
\mathcal{C}\big(3,m,\Omega,T_2\big)\geq \max\limits_{\Gamma\subseteq V} \frac{~\textup{Rank}\big(T_2[I_{\Gamma}]\big)~}{~|\Gamma|~}.
\end{align}
\end{itemize}
\end{lemma}
\begin{proof}
Consider an arbitrary model $\big(3,m,\Omega,T_2\big)$ with $1\leq m\leq 3$, and let $(\mathcal{N},T_2)$ be the corresponding model of network function computation, where we recall \eqref{T1-T2} that $T_2=\left[\begin{smallmatrix} 1&0\\ 1&0\\ 0&1  \end{smallmatrix}\right]$.
In the following, we will specify the lower bound \eqref{3-m-Omega-T-geq-max-best-known-upper-bound} for the model $\big(3,m,\Omega,T_2\big)$ according to the number~$m$ of encoders.

We first consider $m=1$, i.e., $V=\{v_1\}$. By the equivalence of the distributed source coding model for function compression and the model of network function computation discussed in Section \ref{NFC}, the corresponding model of network function computation is depicted in Fig.\,\ref{T2-m=1-NFC}, where by $r\triangleq\textup{Rank}(T_2)=2$ and $m=1$, we have $\ell=\lceil m/r \rceil=1$ (cf. \eqref{def-ell}) and thus there is only one edge from each source node $\sigma_i$ to $v_1$ for $i=1,2,3$.
\begin{figure}[htbp]
\vspace{-1.5em}
	\centering

\tikzstyle{format}=[draw,circle,fill=gray!20, minimum size=6pt, inner sep=0pt]

	\begin{tikzpicture}

		\node[format](a1)at(6,4){};
        \node at (6,4.4) {$\sigma_1$};
		\node[format](a2)at(8,4){};
        \node at (8,4.4) {$\sigma_2$};
		\node[format](a3)at(10,4){};
        \node at (10,4.4) {$\sigma_3$};

		\node[format](r1)at(8,2.5){};
        \node at (8.4,2.4) {$v_1$};

		\node[format](r)at(8,1){};
        \node at (8.4,1) {$\rho$};

 \draw[->,>=latex](a1)--(r1) node[midway, auto,swap, left=0mm] {};
  \draw[->,>=latex](a2)--(r1) node[midway, auto,swap, left=-1mm] {};
 \draw[->,>=latex](a3)--(r1) node[midway, auto,swap, right=0mm] {};
 \draw[->,>=latex](r1)--(r) node[midway, auto,swap, right=0mm] {$e_1$};
	\end{tikzpicture}
 \vspace{-1em}
\caption{The model $(\mathcal{N},T_2)$ corresponding to $(3,1,\Omega,T_2)$.}
\label{T2-m=1-NFC}
 \vspace{-1.5em}
\end{figure}
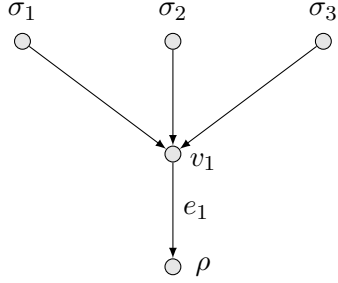

Consider an arbitrary cut set $C\in\Lambda(\mathcal{N})$ and a strong partition $\mathcal{P}_C=\{C_1,C_2,\cdots,C_t\}$ of~$C$, where $1\leq t\leq 3$. If $\mathcal{P}_C$ is the trivial strong partition $\mathcal{P}_C=\{C\}$ (i.e., $t=1$), we have
\begin{align}\label{3-1-T2-rank-PC-trivial-T2-leq-2}
\textbf{\textup{rank}}_{\mathcal{P}_C}(T_2)&=\sum_{i=1}^t \textup{Rank}\big(T_2[I_{C_i}]\big)+\textup{Rank}\big(T_2[I_{C}]\big)
-\textup{Rank}\big(T_2[\cup_{i=1}^t I_{C_i}]\big) \\
&=\textup{Rank}(T_2[I_C])\leq2, \notag
\end{align}
which implies that
\begin{align}\label{3-1-Omega-T2-trivial-sp}
\max\limits_{C\in\Lambda(\mathcal{N})} \frac{~\textbf{\textup{rank}}_{\{C\}}(T_2)~}{~|C|~}=\max\limits_{C\in\Lambda(\mathcal{N})} \frac{~\textup{Rank}(T_2[I_C])~}{~|C|~}\leq2.
\end{align}
If $\mathcal{P}_C$ is a nontrivial strong partition of $C$ (i.e., $t=2$ or $3$), we have $|C|\geq2$ and
\[
\sum_{i=1}^t \textup{Rank}\big(T_2[I_{C_i}]\big)\leq3,~~ \textup{Rank}\big(T_2[I_{C}]\big)\leq2~~ \textup{and}~~ \textup{Rank}\big(T_2[\cup_{i=1}^t I_{C_i}]\big)\geq1,
\]
immediately implying that
\[
\textbf{\textup{rank}}_{\mathcal{P}_C}(T_2)=\sum_{i=1}^t \textup{Rank}\big(T_2[I_{C_i}]\big)+\textup{Rank}\big(T_2[I_{C}]\big)
-\textup{Rank}\big(T_2[\cup_{i=1}^t I_{C_i}]\big)\leq4.
\]
Hence, we have
\begin{align}\label{3-1-Omega-T2-nontrivial-sp}
\max\limits_{C\in\Lambda(\mathcal{N})}\, \max_{ \substack{\textup{all nontrivial strong}\\  \textup{partitions}\, \mathcal{P}_C\,\textup{of}\,C} }  \frac{~\textbf{\textup{rank}}_{\mathcal{P}_C}(T_2)~}{~|C|~}
\leq\frac{~4~}{~2~}=2.
\end{align}
Combining \eqref{3-1-Omega-T2-trivial-sp} and \eqref{3-1-Omega-T2-nontrivial-sp}, we obtain that
\begin{align}
&\max\limits_{C\in\Lambda(\mathcal{N})}\, \max_{ \substack{\textup{all strong partitions}\\ \mathcal{P}_C\,\textup{of}\,C} }  \frac{~\textbf{\textup{rank}}_{\mathcal{P}_C}(T_2)~}{~|C|~} \notag \\
&=\max\left\{ \max\limits_{C\in\Lambda(\mathcal{N})} \frac{~\textbf{\textup{rank}}_{\{C\}}(T_2)~}{~|C|~},~ \max\limits_{C\in\Lambda(\mathcal{N})}\, \max_{ \substack{\textup{all nontrivial strong}\\  \textup{partitions}\, \mathcal{P}_C\,\textup{of}\,C} }  \frac{~\textbf{\textup{rank}}_{\mathcal{P}_C}(T_2)~}{~|C|~} \right\} \leq2.   \label{3-1-Omega-T2-max-trivial--nontrivail-sp}
\end{align}
On the other hand, we take the cut set $C=\{e_1=(v_1,\rho)\}$ to yield
\[
\frac{~\textbf{\textup{rank}}_{\{C\}}(T_2)~}{~|C|~}=\frac{~\textup{Rank}(T_2)~}{~|\{e_1\}|~}=2.
\]
Together with \eqref{3-1-Omega-T2-max-trivial--nontrivail-sp}, the lower bound \eqref{3-m-Omega-T-geq-max-best-known-upper-bound} is specified to $2$ for the model $(3,1,\Omega,T_2)$, alternatively, written as 
\[
\max\limits_{\Gamma\subseteq V} \frac{~\textup{Rank}\big(T_2[I_{\Gamma}]\big)~}{~|\Gamma|~}~~ \textup{($=2$)}
\]
for unifying the form of the specified lower bounds.

Next, we consider $m=2$ or $3$ in the following. Let $C\in\Lambda(\mathcal{N})$ be an arbitrary cut set and $\mathcal{P}_C=\{C_1,C_2,\cdots,C_t\}$ be an arbitrary strong partition of~$C$, where $1\leq t\leq 3$. We let
\[
\textbf{\textup{1}}_C(\mathcal{P}_C)\triangleq \textup{Rank}\big(T_2[I_{C}]\big)
-\textup{Rank}\big(T_2[\cup_{i=1}^t I_{C_i}]\big),
\]
and then
\begin{align}\label{rank-Gamma-T2-eq-i=1-t-Rank-T1-I-Gamma-i}
\textbf{\textup{rank}}_{\mathcal{P}_C}(T_2)=\sum_{i=1}^t \textup{Rank}\big(T_2[I_{C_i}]\big)+\textbf{\textup{1}}_C(\mathcal{P}_C)
\end{align}
(cf. \eqref{3-1-T2-rank-PC-trivial-T2-leq-2}). First, since $I_{C}\supseteq \cup_{i=1}^t I_{C_i}$, we have $\textbf{\textup{1}}_C(\mathcal{P}_C)\geq0$. 
Further, since $\textup{Rank}\big(T_2[I_{C}]\big)\leq\textup{Rank}\big(T_2\big)=2$ and $\textup{Rank}\big(T_2[\cup_{i=1}^t I_{C_i}]\big)\geq1$, we further have $\textbf{\textup{1}}_C(\mathcal{P}_C)\leq1$.
This implies that $\textbf{\textup{1}}_C(\mathcal{P}_C)$ only takes values~$0$ and $1$, and $\textbf{\textup{1}}_C(\mathcal{P}_C)=1$ if and only if
\begin{align}\label{T2-m=2-def-I-C-PC=1}
\textup{Rank}\big(T_2[I_{C}]\big)=2 \quad\textup{and}\quad \textup{Rank}\big(T_2[\cup_{i=1}^t I_{C_i}]\big)=1.
\end{align}
We further see that the only case that satisfies \eqref{T2-m=2-def-I-C-PC=1} is that $\mathcal{P}_C=\{C_1,C_2\}$ (i.e., $t=2$) is a nontrivial strong partition with $I_C=S=\{\sigma_1,\sigma_2,\sigma_3\}$, and either $I_{C_1}=\{\sigma_1\}$ and $I_{C_2}=\{\sigma_2\}$ or $I_{C_1}=\{\sigma_2\}$ and $I_{C_2}=\{\sigma_1\}$. This further implies that there exists at least an edge $e'_1\in C_1$ and an edge $e'_2\in C_2$ such that there must exist two paths from $\sigma_3$ to $\rho$ passing through $e'_1$ and $e'_2$, respectively. More precisely, for $i=1,2$, either $e'_i\in\textup{Out}(\sigma_3)$ or there exists an edge connecting $\sigma_3$ to $\textup{tail}(e'_i)$ (which implies that $e'_i\in\textup{In}(\rho)$).

We denote by $\textbf{\textup{P}}_C$ the collection of all the strong partitions of a cut set $C\in\Lambda(\mathcal{N})$ and rewrite the lower bound~\eqref{3-m-Omega-T-geq-max-best-known-upper-bound} as follows:
\begin{align}\label{rank-C-T2-C-eq-max-C-P-C}
&\max\limits_{C\in\Lambda(\mathcal{N})}\, \max\limits_{ \substack{\textup{all strong partitions}\\ \mathcal{P}_{C}\, \textup{of}\,C} } \frac{~\textbf{\textup{rank}}_{\mathcal{P}_C}(T_2)~}{~|C|~}
=\max_{\textup{all}\,(C,\mathcal{P}_C)\in\Lambda(\mathcal{N})\times\textbf{\textup{P}}_C} \frac{~\textbf{\textup{rank}}_{\mathcal{P}_C}(T_2)~}{~|C|~}  \notag \\
&=\max\left\{ \max_{ \substack{\textup{all}\,(C,\mathcal{P}_C)\in\Lambda(\mathcal{N})\times\textbf{\textup{P}}_C\\ \textup{with}\,\textbf{\textup{1}}_C(\mathcal{P}_C)=0} } \frac{~\textbf{\textup{rank}}_{\mathcal{P}_C}(T_2)~}{~|C|~},~ \max_{ \substack{\textup{all}\,(C,\mathcal{P}_C)\in\Lambda(\mathcal{N})\times\textbf{\textup{P}}_C\\ \textup{with}\,\textbf{\textup{1}}_C(\mathcal{P}_C)=1} } \frac{~\textbf{\textup{rank}}_{\mathcal{P}_C}(T_2)~}{~|C|~} \right\}.
\end{align}
Here, we remark that if there does not exist such a pair $(C,\mathcal{P}_C)\in\Lambda(\mathcal{N})\times\textbf{\textup{P}}_C$ with  $\textbf{\textup{1}}_C(\mathcal{P}_C)=1$, we set
\[
\max_{ \substack{\textup{all}\,(C,\mathcal{P}_C)\in\Lambda(\mathcal{N})\times\textbf{\textup{P}}_C\\ \textup{s.t.}\,\textbf{\textup{1}}_C(\mathcal{P}_C)=1} } \frac{~\textbf{\textup{rank}}_{\mathcal{P}_C}(T_2)~}{~|C|~}=0.
\]
For the other case, we can always find a strong partition $\mathcal{P}_C$ for a cut set $C\in\Lambda(\mathcal{N})$ such that $\textbf{\textup{1}}_C(\mathcal{P}_C)=0$, e.g., the pair of a cut set $C$ and $\mathcal{P}_C=\{C\}$ implying $\textbf{\textup{1}}_C(\mathcal{P}_C)=0$.

We first focus on
\[
\max_{ \substack{\textup{all}\,(C,\mathcal{P}_C)\in\Lambda(\mathcal{N})\times\textbf{\textup{P}}_C\\ \textup{s.t.}\,\textbf{\textup{1}}_C(\mathcal{P}_C)=0} } \frac{~\textbf{\textup{rank}}_{\mathcal{P}_C}(T_2)~}{~|C|~}.
\]
We note that for each pair $(C,\mathcal{P}_C)\in\Lambda(\mathcal{N})\times\textbf{\textup{P}}_C$ with $\textbf{\textup{1}}_C(\mathcal{P}_C)=0$, by \eqref{rank-Gamma-T2-eq-i=1-t-Rank-T1-I-Gamma-i} we have
\[
\textbf{\textup{rank}}_{\mathcal{P}_C}(T_2)=\sum_{i=1}^t \textup{Rank}\big(T_2[I_{C_i}]\big),
\]
where we let $\mathcal{P}_C=\{C_1,C_2,\cdots,C_t\}$. This implies that
\begin{align}
\max_{ \substack{\textup{all}\,(C,\mathcal{P}_C)\in\Lambda(\mathcal{N})\times\textbf{\textup{P}}_C\\ \textup{with}\,\textbf{\textup{1}}_C(\mathcal{P}_C)=0} } \frac{~\textbf{\textup{rank}}_{\mathcal{P}_C}(T_2)~}{~|C|~}&=\max_{ \substack{\textup{all}\,(C,\mathcal{P}_C)\in\Lambda(\mathcal{N})\times\textbf{\textup{P}}_C\\ \textup{with}\,\textbf{\textup{1}}_C(\mathcal{P}_C)=0} } \frac{~\sum_{i=1}^t \textup{Rank}\big(T_2[I_{C_i}]\big)~}{~|C|~} \notag \\ 
&=\max_{ \substack{\textup{all}\,(C,\mathcal{P}_C)\in\Lambda(\mathcal{N})\times\textbf{\textup{P}}_C\\ \textup{with}\,\textbf{\textup{1}}_C(\mathcal{P}_C)=0} } \frac{~\sum_{i=1}^t \textup{Rank}\big(T_2[I_{C_i}]\big)~}{~\sum_{i=1}^t |C_i|~} \label{max-Gamma-rank-Gamma-T2-Gamma-eq-max-Gamma-all-sp-i=1-t-Rank-T2-I-Gamma-i-Gamma-i}  \\
&\leq \max_{ \substack{\textup{all}\,(C,\mathcal{P}_C)\in\Lambda(\mathcal{N})\times\textbf{\textup{P}}_C\\ \textup{with}\,\textbf{\textup{1}}_C(\mathcal{P}_C)=0} }\, \max_{1\leq i\leq t} \frac{~\textup{Rank}\big(T_2[I_{C_i}]\big)~}{~|C_i|~} \label{max-Gamma-rank-Gamma-T2-Gamma-leq-max-Gamma-all-sp-max-i=1-t-Rank-T2-I-Gamma-i-Gamma-i} \\
&\leq  \max_{C\in\Lambda(\mathcal{N})}\frac{~\textup{Rank}\big(T_2[I_{C}]\big)~}{~|C|~}, \label{max-Gamma-rank-Gamma-T2-Gamma-leq-max-Gamma-Rank-T2-I-Gamma}
\end{align}
where the equations \eqref{max-Gamma-rank-Gamma-T2-Gamma-eq-max-Gamma-all-sp-i=1-t-Rank-T2-I-Gamma-i-Gamma-i}, \eqref{max-Gamma-rank-Gamma-T2-Gamma-leq-max-Gamma-all-sp-max-i=1-t-Rank-T2-I-Gamma-i-Gamma-i} and \eqref{max-Gamma-rank-Gamma-T2-Gamma-leq-max-Gamma-Rank-T2-I-Gamma} hold in the same way to prove \eqref{max-Gamma-rank-Gamma-T1-Gamma-eq-max-Gamma-all-sp-i=1-t-Rank-T1-I-Gamma-i-Gamma-i}, \eqref{max-Gamma-rank-Gamma-T1-Gamma-leq-max-Gamma-all-sp-max-i=1-t-Rank-T1-I-Gamma-i-Gamma-i} and \eqref{max-Gamma-rank-Gamma-T1-Gamma-leq-max-Gamma-Rank-T1-I-Gamma}.

On the other hand, by considering the trivial strong partition $\mathcal{P}_C=\{C\}$ for each cut set $C\in\Lambda(\mathcal{N})$, for which we have $\textbf{\textup{1}}_C(\mathcal{P}_C)=0$, we thus obtain that
\begin{align}
\max_{ \substack{\textup{all}\,(C,\mathcal{P}_C)\in\Lambda(\mathcal{N})\times\textbf{\textup{P}}_C\\ \textup{with}\,\textbf{\textup{1}}_C(\mathcal{P}_C)=0} } \frac{~\sum_{i=1}^t \textup{Rank}\big(T_2[I_{C_i}]\big)~}{~|C|~}
&\geq \max_{ \textup{all}\,(C,\{C\})\in\Lambda(\mathcal{N})\times\textbf{\textup{P}}_C } \frac{~\textup{Rank}\big(T_2[I_{C}]\big)~}{~|C|~}   \notag \\
&=\max_{C\in\Lambda(\mathcal{N})}\frac{~\textup{Rank}\big(T_2[I_{C}]\big)~}{~|C|~}. \label{max-Gamma-rank-Gamma-T2-Gamma-geq-max-Gamma-Rank-T2-I-Gamma}
\end{align}
By combining  \eqref{max-Gamma-rank-Gamma-T2-Gamma-leq-max-Gamma-Rank-T2-I-Gamma} and  \eqref{max-Gamma-rank-Gamma-T2-Gamma-geq-max-Gamma-Rank-T2-I-Gamma}, we have proved that
\begin{align}\label{1=0=rank-P-C-T2-C-eq-Rank-T2-I-C-C}
\max_{ \substack{\textup{all}\,(C,\mathcal{P}_C)\in\Lambda(\mathcal{N})\times\textbf{\textup{P}}_C\\ \textup{with}\,\textbf{\textup{1}}_C(\mathcal{P}_C)=0} } \frac{~\textbf{\textup{rank}}_{\mathcal{P}_C}(T_2)~}{~|C|~} =\max_{ C\in\Lambda(\mathcal{N}) } \frac{~\textup{Rank}(T_2[I_C])~}{~|C|~}.
\end{align}

Next, we focus on
\[
\max_{ \substack{\textup{all}\,(C,\mathcal{P}_C)\in\Lambda(\mathcal{N})\times\textbf{\textup{P}}_C\\ \textup{with}\,\textbf{\textup{1}}_C(\mathcal{P}_C)=1} } \frac{~\textbf{\textup{rank}}_{\mathcal{P}_C}(T_2)~}{~|C|~}.
\]
Recalling the discussion immediately below \eqref{T2-m=2-def-I-C-PC=1}, for each pair $(C,\mathcal{P}_C)\in\Lambda(\mathcal{N})\times\textbf{\textup{P}}_C$ with $\textbf{\textup{1}}_C(\mathcal{P}_C)=1$, we have $t=2$ and $|I_{C_1}|=|I_{C_2}|=1$, which implies that
\[
\textbf{\textup{rank}}_{\mathcal{P}_C}(T_2)=\sum_{i=1}^t \textup{Rank}\big(T_2[I_{C_i}]\big)+\textbf{\textup{1}}_C(\mathcal{P}_C)
=\textup{Rank}\big(T_2[I_{C_1}]\big)+\textup{Rank}\big(T_2[I_{C_2}]\big)
+\textbf{\textup{1}}_C(\mathcal{P}_C)=3.
\]
Thus,
\begin{align}\label{1=1=rank-P-C-T2-C-eq-Rank-T2-I-C-C}
\max_{ \substack{\textup{all}\,(C,\mathcal{P}_C)\in\Lambda(\mathcal{N})\times\textbf{\textup{P}}_C\\ \textup{with}\,\textbf{\textup{1}}_C(\mathcal{P}_C)=1} } \frac{~\textbf{\textup{rank}}_{\mathcal{P}_C}(T_2)~}{~|C|~}=\max_{ \substack{\textup{all}\,(C,\mathcal{P}_C)\in\Lambda(\mathcal{N})\times\textbf{\textup{P}}_C\\ \textup{with}\,\textbf{\textup{1}}_C(\mathcal{P}_C)=1} } \frac{~3~}{~|C|~}.
\end{align}
Combining \eqref{rank-C-T2-C-eq-max-C-P-C}, \eqref{1=0=rank-P-C-T2-C-eq-Rank-T2-I-C-C} and \eqref{1=1=rank-P-C-T2-C-eq-Rank-T2-I-C-C}, we thus specify the lower bound \eqref{3-m-Omega-T-geq-max-best-known-upper-bound} to be
\begin{align}
\max\limits_{C\in\Lambda(\mathcal{N})} \max\limits_{ \substack{\textup{all strong partitions}\\ \mathcal{P}_{C}\, \textup{of}\,C} } \frac{~\textbf{\textup{rank}}_{\mathcal{P}_C}(T_2)~}{~|C|~}
=\max\left\{ \max_{ C\in\Lambda(\mathcal{N}) } \frac{~\textup{Rank}(T_2[I_C])~}{~|C|~}, \max_{ \substack{\textup{all}\,(C,\mathcal{P}_C)\in\Lambda(\mathcal{N})\times\textbf{\textup{P}}_C\\ \textup{with}\,\textbf{\textup{1}}_C(\mathcal{P}_C)=1} } \frac{~3~}{~|C|~} \right\}. \label{max-max-1-max-0-eq-max-max-Rank-max-3}
\end{align}

In the following, we compute the lower bound \eqref{max-max-1-max-0-eq-max-max-Rank-max-3} for the cases $m=2$ and $m=3$, respectively.

\smallskip

\noindent \textbf{Case 1:} $m=2$, i.e., $V=\{v_1,v_2\}$.

By \eqref{max-max-1-max-0-eq-max-max-Rank-max-3}, we first consider
\[
\max_{ \substack{\textup{all}\,(C,\mathcal{P}_C)\in\Lambda(\mathcal{N})\times\textbf{\textup{P}}_C\\ \textup{with}\,\textbf{\textup{1}}_C(\mathcal{P}_C)=1} } \frac{~3~}{~|C|~}.
\]
Let $(C,\mathcal{P}_C)\in\Lambda(\mathcal{N})\times\textbf{\textup{P}}_C$ be such a pair with $\textbf{\textup{1}}_C(\mathcal{P}_C)=1$. Recalling the discussion immediately below~\eqref{T2-m=2-def-I-C-PC=1}, we can see that if $|C|=2$, then $C=\big\{e_1\triangleq(v_1,\rho),\,e_2\triangleq(v_2,\rho)\big\}=\textup{In}(\rho)$.
Thus, either $C_1=\{e_1\}$ and $C_2=\{e_2\}$ or $C_1=\{e_2\}$ and $C_2=\{e_1\}$. Together with the other requirements that $I_C=S$ and either $I_{C_1}=\{\sigma_1\}$ and $I_{C_2}=\{\sigma_2\}$ or $I_{C_1}=\{\sigma_2\}$ and $I_{C_2}=\{\sigma_1\}$, we obtain that if in $(\mathcal{N},T_2)$ there exists a pair $(C,\mathcal{P}_C)\in\Lambda(\mathcal{N})\times\textbf{\textup{P}}_C$ such that $\textbf{\textup{1}}_C(\mathcal{P}_C)=1$ and $|C|=2$, $(\mathcal{N},T_2)$ must be one of the two models depicted in
Fig.\,\ref{T2-m=2-NFC}, where by
$m=2$ and $r\triangleq\textup{Rank}(T_2)=2$, we have $\ell=\lceil m/r \rceil=1$ (cf. \eqref{def-ell}), i.e., there is only one edge from each source node $\sigma_i$ to $v$ in $\Gamma_{\sigma_i}$ for $i=1,2,3$.
Immediately, for the two models $(\mathcal{N},T_2)$, we have
\begin{align}\label{m=2-first-equality-eq-3/2}
\max_{ \substack{\textup{all}\,(C,\mathcal{P}_C)\in\Lambda(\mathcal{N})\times\textbf{\textup{P}}_C\\ \textup{with}\,\textbf{\textup{1}}_C(\mathcal{P}_C)=1} } \frac{~3~}{~|C|~}=\frac{~3~}{~\big|\{e_1,e_2\}\big|~}=\frac{~3~}{~2~}.
\end{align}
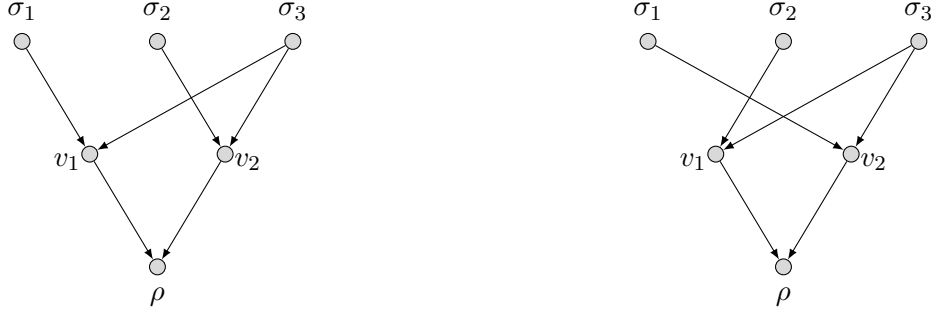
\begin{figure*}[t]
\tikzstyle{vertex}=[draw,circle,fill=gray!30,minimum size=6pt, inner sep=0pt]
\tikzstyle{vertex1}=[draw,circle,fill=gray!80,minimum size=6pt, inner sep=0pt]
\centering
\begin{minipage}[b]{0.5\textwidth}
\centering
{
 \begin{tikzpicture}[x=0.6cm]
    \node[draw,circle,fill=gray!30,minimum size=6pt, inner sep=0pt](a1)at(1,4){};
        \node at (1,4.4) {$\sigma_1$};
        \node[draw,circle,fill=gray!30,minimum size=6pt, inner sep=0pt](a2)at(4,4){};
        \node at (4,4.4) {$\sigma_2$};
        \node[draw,circle,fill=gray!30,minimum size=6pt, inner sep=0pt](a3)at(7,4){};
        \node at (7,4.4) {$\sigma_3$};

        \node[draw,circle,fill=gray!30,minimum size=6pt, inner sep=0pt](r1)at(2.5,2.5){};
        \node[draw,circle,fill=gray!30,minimum size=6pt, inner sep=0pt](r2)at(5.5,2.5){};
        \node at (2,2.4) {$v_1$};
        \node at (6,2.4) {$v_2$};

        \node[draw,circle,fill=gray!30,minimum size=6pt, inner sep=0pt](r)at(4,1){};
        \node at (4,0.6) {$\rho$};

        \draw[->,>=latex](a1)--(r1);
        \draw[->,>=latex](a2)--(r2);
        \draw[->,>=latex](a3)--(r1);
        \draw[->,>=latex](a3)--(r2);
        \draw[->,>=latex](r1)--(r) node[midway, auto,swap, left=0mm] {};
        \draw[->,>=latex](r2)--(r) node[midway, auto,swap, right=0mm] {};
    \end{tikzpicture}
}
\end{minipage}%
\centering
\begin{minipage}[b]{0.5\textwidth}
\centering
 \begin{tikzpicture}[x=0.6cm]
   \node[draw,circle,fill=gray!30,minimum size=6pt, inner sep=0pt](a1)at(1,4){};
        \node at (1,4.4) {$\sigma_1$};
        \node[draw,circle,fill=gray!30,minimum size=6pt, inner sep=0pt](a2)at(4,4){};
        \node at (4,4.4) {$\sigma_2$};
        \node[draw,circle,fill=gray!30,minimum size=6pt, inner sep=0pt](a3)at(7,4){};
        \node at (7,4.4) {$\sigma_3$};

        \node[draw,circle,fill=gray!30,minimum size=6pt, inner sep=0pt](r1)at(2.5,2.5){};
        \node[draw,circle,fill=gray!30,minimum size=6pt, inner sep=0pt](r2)at(5.5,2.5){};
        \node at (2,2.4) {$v_1$};
        \node at (6,2.4) {$v_2$};

        \node[draw,circle,fill=gray!30,minimum size=6pt, inner sep=0pt](r)at(4,1){};
        \node at (4,0.6) {$\rho$};

        \draw[->,>=latex](a1)--(r2);
        \draw[->,>=latex](a2)--(r1);
        \draw[->,>=latex](a3)--(r1);
        \draw[->,>=latex](a3)--(r2);
        \draw[->,>=latex](r1)--(r) node[midway, auto,swap, left=0mm] {};
        \draw[->,>=latex](r2)--(r) node[midway, auto,swap, right=0mm] {};
    \end{tikzpicture}
\end{minipage}
\vspace{-2.5em}
\caption{The only two models $(\mathcal{N},T_2)$ where there is a pair $(C,\mathcal{P}_C)$ with $\textbf{\textup{1}}_C(\mathcal{P}_C)=1$ and $|C|=2$.}
    \label{T2-m=2-NFC}
    \vspace{-2.5em}
\end{figure*}
For the two models, we further compute
\begin{align}\label{m=2-second-equality-eq-1}
\max_{ C\in\Lambda(\mathcal{N}) } \frac{~\textup{Rank}(T_2[I_C])~}{~|C|~}=1.
\end{align}
Combining \eqref{m=2-first-equality-eq-3/2} and \eqref{m=2-second-equality-eq-1} with the lower bound \eqref{max-max-1-max-0-eq-max-max-Rank-max-3}, we have obtained that
\[
\max\limits_{C\in\Lambda(\mathcal{N})}\, \max\limits_{ \substack{\textup{all strong partitions}\\ \mathcal{P}_{C}\, \textup{of}\,C} } \frac{~\textbf{\textup{rank}}_{\mathcal{P}_C}(T_2)~}{~|C|~}=\frac{~3~}{~2~}.
\]
We have thus proved that for the model $(3,2,\Omega,T_2)$, when $\Omega$ satisfies that $\Gamma_{\sigma_3}=V=\{v_1,v_2\}$, $\Gamma_{\sigma_1}=\{v_i\}$ and $\Gamma_{\sigma_2}=\{v_j\}$ in which $\{i,j\}=\{1,2\}$, then
\[
\mathcal{C}\big(3,2,\Omega,T_2\big)\geq \frac{3}{2}.
\]

Next, we consider other cases of the model $(\mathcal{N},T_2)$ where if there exists a pair $(C,\mathcal{P}_C)\in\Lambda(\mathcal{N})\times\textbf{\textup{P}}_C$ such that $\textbf{\textup{1}}_C(\mathcal{P}_C)=1$, then $|C|\geq3$. With this, we easily see that
\begin{align}\label{m=2-second-leq-1}
\max_{ \substack{\textup{all}\,(C,\mathcal{P}_C)\in\Lambda(\mathcal{N})\times\textbf{\textup{P}}_C\\ \textup{with}\,\textbf{\textup{1}}_C(\mathcal{P}_C)=1} } \frac{~3~}{~|C|~} \leq\frac{3}{3}=1.
\end{align}
Furthermore, we take the cut set $\textup{In}(\rho)=\big\{(v_1,\rho),(v_2,\rho)\big\}$, we have
\[
\max_{ C\in\Lambda(\mathcal{N}) } \frac{~\textup{Rank}(T_2[I_C])~}{~|C|~}\geq \frac{~\textup{Rank}(T_2)~}{~|\textup{In}(\rho)|~}=1,
\]
where we note that $I_{\textup{In}(\rho)}=S$ and thus $T_2[I_{\textup{In}(\rho)}]=T_2$. Together with \eqref{m=2-second-leq-1}, we have
\[
\max_{ C\in\Lambda(\mathcal{N}) } \frac{~\textup{Rank}(T_2[I_C])~}{~|C|~}\geq \max_{ \substack{\textup{all}\,(C,\mathcal{P}_C)\in\Lambda(\mathcal{N})\times\textbf{\textup{P}}_C\\ \textup{with}\,\textbf{\textup{1}}_C(\mathcal{P}_C)=1} } \frac{~3~}{~|C|~},
\]
and so
\[
\max\limits_{C\in\Lambda(\mathcal{N})}\, \max\limits_{ \substack{\textup{all strong partitions}\\ \mathcal{P}_{C}\, \textup{of}\,C} } \frac{~\textbf{\textup{rank}}_{\mathcal{P}_C}(T_2)~}{~|C|~}=\max_{ C\in\Lambda(\mathcal{N}) } \frac{~\textup{Rank}(T_2[I_C])~}{~|C|~}.
\]

\medskip

\noindent \textbf{Case 2:} $m=3$, i.e., $V=\{v_1,v_2,v_3\}$.

For the case, we will prove that
\begin{align}\label{m=3-first-leq-second}
\max_{ \substack{\textup{all}\,(C,\mathcal{P}_C)\in\Lambda(\mathcal{N})\times\textbf{\textup{P}}_C\\ \textup{with}\,\textbf{\textup{1}}_C(\mathcal{P}_C)=1} } \frac{~3~}{~|C|~} \leq \max_{ C\in\Lambda(\mathcal{N}) } \frac{~\textup{Rank}(T_2[I_C])~}{~|C|~},
\end{align}
which, together with \eqref{max-max-1-max-0-eq-max-max-Rank-max-3}, shows that
\[
\mathcal{C}\big(3,3,\Omega,T_2\big)\geq \max_{ C\in\Lambda(\mathcal{N}) } \frac{~\textup{Rank}(T_2[I_C])~}{~|C|~}.
\]
To prove \eqref{m=3-first-leq-second}, it suffices to prove that for any pair $(C,\mathcal{P}_C)\in\Lambda(\mathcal{N})\times\textbf{\textup{P}}_C$ such that $\textbf{\textup{1}}_C(\mathcal{P}_C)=1$,
\begin{equation}\label{m=3-3/C-leq-max-Rank-T2-I-C-C}
\frac{~3~}{~|C|~} \leq \max_{ C\in\Lambda(\mathcal{N}) } \frac{~\textup{Rank}(T_2[I_C])~}{~|C|~}.
\end{equation}
Consider such a pair $(C,\mathcal{P}_C)\in\Lambda(\mathcal{N})\times\textbf{\textup{P}}_C$ with $\textbf{\textup{1}}_C(\mathcal{P}_C)=1$. By the discussion immediately below~\eqref{T2-m=2-def-I-C-PC=1}, we have known that $I_C=S$. Together with $m=3$, we have $|C|\geq3$ because the minimum cut capacity separating $\rho$ from all the source nodes in $S$ is greater than or equal to $3$.
In the following, we continue to prove \eqref{m=3-3/C-leq-max-Rank-T2-I-C-C} according to the size of $C$.

\begin{itemize}
\item  If $|C|=3$, which implies that $3/|C|=1$, by the discussion immediately below \eqref{T2-m=2-def-I-C-PC=1}, we have known that $\mathcal{P}_C=\{C_1,C_2\}$, and $I_{C_1}=\{\sigma_i\}$, $I_{C_2}=\{\sigma_j\}$ for $\{i,j\}=\{1,2\}$. By $|C|=|C_1|+|C_2|$, we obtain that either $|C_1|=1$ and $|C_2|=2$ or $|C_1|=2$ and $|C_2|=1$, which implies that
\begin{align}\label{m=3-C=3}
\max_{ C\in\Lambda(\mathcal{N}) } \frac{~\textup{Rank}(T_2[I_C])~}{~|C|~}&\geq \max\left\{ \frac{~\textup{Rank}(T_2[I_{C_1}])~}{~|C_1|~},~ \frac{~\textup{Rank}(T_2[I_{C_2}])~}{~|C_2|~}\right\} \notag \\
&=\max\left\{1,\frac{1}{2}\right\}
\mspace{-3mu}=\mspace{-4mu}1\mspace{-4mu}=\mspace{-4mu}\frac{~3~}{~|C|~}.
\end{align}
\item  If $|C|=4$, we have $3/|C|=3/4$. Similarly, we have known that $\mathcal{P}_C=\{C_1,C_2\}$. With $|C|=|C_1|+|C_2|$, we have \textup{\rmnum{1})} $|C_1|=1$ and $|C_2|=3$ or \textup{\rmnum{2})} $|C_1|=3$ and $|C_2|=1$ or \textup{\rmnum{3})} $|C_1|=|C_2|=2$. For \textup{\rmnum{1})} and \textup{\rmnum{2})},
by the same argument for proving \eqref{m=3-C=3}, we have
\begin{align*}
\max_{ C\in\Lambda(\mathcal{N}) } \frac{~\textup{Rank}(T_2[I_C])~}{~|C|~}&\geq \max\left\{ \frac{~\textup{Rank}(T_2[I_{C_1}])~}{~|C_1|~},~ \frac{~\textup{Rank}(T_2[I_{C_2}])~}{~|C_2|~}\right\} \\
&=\max\left\{1,\frac{1}{3}\right\}=1\geq\frac{3}{4}=\frac{~3~}{~|C|~}.
\end{align*}
For \textup{\rmnum{3})} $|C_1|=|C_2|=2$, we assume without loss of generality that $I_{C_1}=\{\sigma_1\}$ and $I_{C_2}=\{\sigma_2\}$ by the discussion immediately below \eqref{T2-m=2-def-I-C-PC=1}. Further, by the equivalence of $(3,3,\Omega,T_2)$ and $(\mathcal{N},T_2)$ in Section~\ref{NFC}, it follows from $m=3$ and $r\triangleq\textup{Rank}(T_2)=2$ that $\ell=\lceil m/r \rceil=2$ (cf. \eqref{def-ell}), i.e., there are two edges from each source node $\sigma_i$ to $v$ in $\Gamma_{\sigma_i}$ for $i=1,2,3$. We now claim $|\Gamma_{\sigma_1}|=1$ or $|\Gamma_{\sigma_2}|=1$. To see this, we first note that $|\Gamma_{\sigma_1}|\leq2$ and $|\Gamma_{\sigma_2}|\leq2$ because $|C_1|=|C_2|=2$. We further assume the contrary of the claim that $|\Gamma_{\sigma_1}|=2$ and $|\Gamma_{\sigma_2}|=2$. Since $|C_1|=|C_2|=2$, it must be
\begin{align*}
C_1=\big\{(v,\rho):v\in\Gamma_{\sigma_1}\big\}\subseteq\textup{In}(\rho)\quad \textup{and} \quad C_2=\big\{(v,\rho):v\in\Gamma_{\sigma_2}\big\}\subseteq\textup{In}(\rho).
\end{align*}
This shows that $C_1\cup C_2\subseteq\textup{In}(\rho)$, implying that
$
4=|C_1|+|C_2|=|C_1\cup C_2|\leq |\textup{In}(\rho)|=3,
$
a contradiction. We have thus proved the claim. Without loss of generality, we let $|\Gamma_{\sigma_1}|=1$, say $\Gamma_{\sigma_1}=\{v\}$. This implies that the singleton edge subset $\{(v,\rho)\}$ satisfies $\{\sigma_1\}\subseteq I_{\{(v,\rho)\}}$. Thus, we have
\[
\max_{ C\in\Lambda(\mathcal{N}) } \frac{~\textup{Rank}(T_2[I_C])~}{~|C|~}\geq \frac{~\textup{Rank}(T_2[I_{\{(v,\rho)\}}])~}{~|\{(v,\rho)\}|~}\geq1\geq \frac{~3~}{~4~}=\frac{~3~}{~|C|~}.
\]
\item  If $|C|\geq5$, we have $3/|C|\leq3/5$. We take a cut set $\textup{In}(\rho)=\{(v_j,\rho):j=1,2,3\}$ and consider
\[
\max_{ C\in\Lambda(\mathcal{N}) } \frac{~\textup{Rank}(T_2[I_C])~}{~|C|~}\geq \frac{~\textup{Rank}(T_2)~}{~|\textup{In}(\rho)|~}=\frac{~2~}{~3~}\geq\frac{~3~}{~5~}\geq\frac{~3~}{~|C|~}.
\]
\end{itemize}

Combining the above discussions, we have proved the inequality \eqref{m=3-3/C-leq-max-Rank-T2-I-C-C} and thus \eqref{m=3-first-leq-second}. Together with the lower bound~\eqref{max-max-1-max-0-eq-max-max-Rank-max-3}, we have specified the lower bound to be
\begin{equation*}\label{2-eq-Rank}
\max\limits_{C\in\Lambda(\mathcal{N})}\, \max\limits_{ \substack{\textup{all strong partitions}\\ \mathcal{P}_{C}\, \textup{of}\,C} } \frac{~\textbf{\textup{rank}}_{\mathcal{P}_C}(T_2)~}{~|C|~}=\max_{ C\in\Lambda(\mathcal{N}) } \frac{~\textup{Rank}(T_2[I_C])~}{~|C|~}.
\end{equation*}
Further, we consider
\begin{align*}
\max_{ C\in\Lambda(\mathcal{N}) } \frac{~\textup{Rank}(T_2[I_C])~}{~|C|~}
=\max\limits_{C\in\Lambda(\mathcal{N})\,\textup{s.t.}\,C\subseteq\textup{In}(\rho)} \frac{~\textup{Rank}\big(T_2[I_{C}]\big)~}{~|C|~}=\max\limits_{\Gamma\subseteq V} \frac{~\textup{Rank}\big(T_2[I_{\Gamma}]\big)~}{~|\Gamma|~},
\end{align*}
where the equalities follow from the same way to prove \eqref{C-Lambda-N-Rank-T1-I-C-eq-C-In-rho-Rank-T1-I-C-C} and \eqref{C-In-rho-Rank-T1-I-C-C-eq-Gamma-V-Rank-T1-I-Gamma-Gamma}. The lemma is thus proved.
\end{proof}

\medskip

For notational simplicity in the rest of the paper, we denote by $R_2$ the lower bound in \eqref{T2-geq-def-R2}, i.e.,
\begin{align}\label{def-R2}
R_2\triangleq\max\limits_{\Gamma\subseteq V} \frac{~\textup{Rank}\big(T_2[I_{\Gamma}]\big)~}{~|\Gamma|~}.
\end{align}
We now start to characterize the capacities for the model $\big(3,m,\Omega,T_2\big)$.

\vspace{-0.5em}

\subsection{The Case of $m=1$}

For this case, we have $V=\{v_1\}$ and the unique connectivity state
$\Omega=\big(\Gamma_{\sigma_i}=\{v_1\}:i=1,2,3\big)$
as depicted in Fig.\,\ref{T2-m=1}. By \eqref{def-R2} we can compute
\[
R_2=\frac{~\textup{Rank}\big(T_2\big)~}{~|\{v_1\}|~}=2,
\]
and thus obtain $\mathcal{C}\big(3,m,\Omega,T_2\big)\geq2$ by Lemma \ref{prop-lower-bound-for-T2}. On the other hand, we present in Fig.\,\ref{T2-m=1} an admissible $1$-shot source code with $n(\mathbf{C})=2$ and hence the coding rate $R(\mathbf{C})=n(\mathbf{C})/k=2$. This implies that $\mathcal{C}(3,1,\Omega,T_2)\leq2$. We have thus proved that $\mathcal{C}(3,1,\Omega,T_2)=R_2=2$.

\begin{figure}[htbp]
\vspace{-1em}
	\centering

\tikzstyle{format}=[draw,circle,fill=gray!20, minimum size=6pt, inner sep=0pt]

	\begin{tikzpicture}

		\node[format](a1)at(6,4){};
        \node at (6,4.4) {$\sigma_1$};
		\node[format](a2)at(8,4){};
        \node at (8,4.4) {$\sigma_2$};
		\node[format](a3)at(10,4){};
        \node at (10,4.4) {$\sigma_3$};

		\node[format](r1)at(8,2.5){};
        \node at (8.4,2.4) {$v_1$};

		\node[format](r)at(8,1){};
        \node at (8.4,1) {$\rho$};

 \draw[->,>=latex](a1)--(r1) node[midway, auto,swap, left=0mm] {$x_1$};
  \draw[->,>=latex](a2)--(r1) node[midway, auto,swap, left=-1mm] {$x_2$};
 \draw[->,>=latex](a3)--(r1) node[midway, auto,swap, right=0mm] {$x_3$};
 \draw[->,>=latex](r1)--(r) node[midway, auto,swap, right=-1mm] {$(x_1+x_2,\,x_3)$};
	\end{tikzpicture}
\vspace{-1em}
\caption{The model $(3,1,\Omega,T_2)$.}
\label{T2-m=1}
\vspace{-2em}
\end{figure}
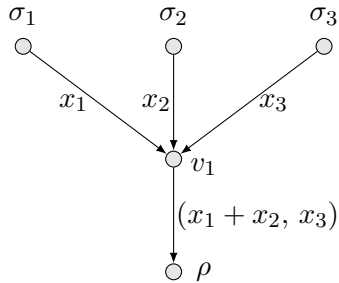

\subsection{The Case of $m=2$}\label{subsec-T2-m=2}

For the case of $m=2$, we write $V=\{v_1,v_2\}$. We divide all connectivity states into two classes below and consider them respectively:
\begin{itemize}
\item \textbf{Class 1:} $\Omega=\big(\Gamma_{\sigma_1},\Gamma_{\sigma_2},\Gamma_{\sigma_3}\big)$ such that $\big|\Gamma_{\{\sigma_i,\sigma_j\}}\big|=1$ for some two-index set $\{i,j\}\subseteq[3]$;
\item \textbf{Class 2:} $\Omega=\big(\Gamma_{\sigma_1},\Gamma_{\sigma_2},\Gamma_{\sigma_3}\big)$ such that $\Gamma_{\{\sigma_i,\sigma_j\}}=V$ for all two-index sets $\{i,j\}\subseteq[3]$.
\end{itemize}

\smallskip

We first consider connectivity states in \textbf{Class 1} by the following three cases of $\big|\Gamma_{\{\sigma_1,\sigma_2\}}\big|=1$, $\big|\Gamma_{\{\sigma_1,\sigma_3\}}\big|=1$ and $\big|\Gamma_{\{\sigma_2,\sigma_3\}}\big|=1$.


\textbf{Case 1A:} The connectivity state $\Omega$ with $\big|\Gamma_{\{\sigma_1,\sigma_2\}}\big|=1$.

By the isomorphism of models, we assume without loss of generality that $\Gamma_{\{\sigma_1,\sigma_2\}}=\{v_1\}$. It follows from $\Gamma_S=V$ that $\sigma_3\rightarrow v_2$, i.e., $v_2\in\Gamma_{\sigma_3}$. This implies that $\Gamma_{\{\sigma_1,\sigma_3\}}=\Gamma_{\{\sigma_2,\sigma_3\}}=V$. Now, by \eqref{def-R2} we can compute the lower bound $R_2=1$.
Let
$\Omega^*\!=\!\big(\Gamma_{\sigma_1}=\{v_1\},\Gamma_{\sigma_2}=\{v_1\},
\Gamma_{\sigma_3}=\{v_2\}\big)$.
See~Fig.\,\ref{T2-m=2-xi=1-fig} for~$\Omega^*$.
Further, we can readily see that $\Omega^*$ is the ``minimum'' connectivity state in this subcase. Combining  Lemma~\ref{lem-cap-omega'-geq-cap-omega} and the lower bound $R_2=1$, we obtain that for any connectivity state $\Omega$ in this subcase,
\begin{align}\label{cap-m=2-Omega*-T2-geq-1}
\mathcal{C}(3,2,\Omega^*,T_2)\geq\mathcal{C}(3,2,\Omega,T_2)\geq1.
\end{align}
We further present in Fig.\,\ref{T2-m=2-xi=1-fig} an admissible $1$-shot source code $\mathbf{C}$ for $(3,2,\Omega^*,T_2)$  of which the coding rate is $1$. This implies that $\mathcal{C}(3,2,\Omega^*,T_2)\leq1$. Together with \eqref{cap-m=2-Omega*-T2-geq-1}, we have proved that $\mathcal{C}(3,2,\Omega,T_2)=1$ for any connectivity state $\Omega$ in this subcase.

\medskip

\textbf{Case 1B:} The connectivity state $\Omega$ with $\big|\Gamma_{\{\sigma_1,\sigma_3\}}\big|=1$.

Let $\Omega$ be an arbitrary connectivity state in this subcase. Without loss of generality, we assume that $\Gamma_{\{\sigma_1,\sigma_3\}}=\{v_1\}$ by the isomorphism of models. Together with $\Gamma_S=V$, we have $\sigma_2\rightarrow v_2$ for all connectivity states in this subcase. With this, by \eqref{def-R2} we can compute the lower bound $R_2=2$.
Further, we see that
$
\Omega^*\triangleq\big(\Gamma_{\sigma_1}
=\{v_1\},\Gamma_{\sigma_2}=\{v_2\},\Gamma_{\sigma_3}=\{v_1\}\big)
$
is the ``minimum'' connectivity state in this subcase. By Lemma~\ref{lem-cap-omega'-geq-cap-omega} and the lower bound $R_2=2$, we obtain that for any connectivity state $\Omega$ in this subcase,
$\mathcal{C}(3,2,\Omega^*,T_2)\geq\mathcal{C}(3,2,\Omega,T_2)\geq2$.
Together with the admissible $1$-shot source code $\mathbf{C}$ for $(3,2,\Omega^*,T_2)$ depicted in Fig.\,\ref{T2-m=2-xi=2-fig} of which the coding rate is $2$, we have proved that $\mathcal{C}(3,2,\Omega^*,T_2)\leq2$. Hence, we have proved that $\mathcal{C}(3,2,\Omega,T_2)=2$ for any connectivity state $\Omega$ in this subcase.

\medskip

\textbf{Case 1C:} The connectivity state $\Omega$ with $\big|\Gamma_{\{\sigma_2,\sigma_3\}}\big|=1$.

Let $\Omega$ be an arbitrary connectivity state in \textbf{Case 1C}. Similarly, we assume without loss of generality that $\Gamma_{\{\sigma_2,\sigma_3\}}=\{v_1\}$. By $\Gamma_S=V$, we have $v_2\in\Gamma_{\sigma_1}$. Further, we take the permutation $\pi=\setlength{\arraycolsep}{3pt}\renewcommand{\arraystretch}{0.7}
\begin{pmatrix}
1 & 2 & 3 \\
2 & 1 & 3
\end{pmatrix}$ for $S$ and the identity permutation $\tau$ for $V$. By performing the permutation pair $(\pi,\tau)$ to an arbitrary connectivity state $\Omega$ with $\big|\Gamma_{\{\sigma_2,\sigma_3\}}\big|=1$ in this subcase, we have $\Gamma_{\{\sigma_1,\sigma_3\}}=\Gamma_{\{\sigma_{\pi(2)},\sigma_{\pi(3)}\}}=\{v_1\}$ in the connectivity state $\Omega\circ(\pi,\tau)$ because $\Gamma_{\{\sigma_2,\sigma_3\}}=\{v_1\}$ in the original connectivity state $\Omega$. Clearly, we see that $\Omega\circ(\pi,\tau)$ is a connectivity state in \textbf{Case 1B}, and so $\mathcal{C}\big(3,\,2,\,\Omega\circ(\pi,\tau),\,T_2\big)=2$ (See \textbf{Case 1B}). Together with $T_2\circ\pi=T_2$, we obtain that
\[
\mathcal{C}(3,2,\Omega,T_2)=\mathcal{C}\big(3,\,2,\,\Omega\circ(\pi,\tau),\,T_2\circ\pi\big)
=\mathcal{C}\big(3,\,2,\,\Omega\circ(\pi,\tau),\,T_2\big)=2.
\]

\begin{figure*}[t]
\tikzstyle{vertex}=[draw,circle,fill=gray!30,minimum size=6pt, inner sep=0pt]
\tikzstyle{vertex1}=[draw,circle,fill=gray!80,minimum size=6pt, inner sep=0pt]
\centering
\begin{minipage}[b]{0.32\textwidth}
\centering
{
 \begin{tikzpicture}[x=0.6cm]
    \node[draw,circle,fill=gray!30,minimum size=6pt, inner sep=0pt](a1)at(1,4){};
        \node at (1,4.4) {$\sigma_1$};
        \node[draw,circle,fill=gray!30,minimum size=6pt, inner sep=0pt](a2)at(4,4){};
        \node at (4,4.4) {$\sigma_2$};
        \node[draw,circle,fill=gray!30,minimum size=6pt, inner sep=0pt](a3)at(7,4){};
        \node at (7,4.4) {$\sigma_3$};

        \node[draw,circle,fill=gray!30,minimum size=6pt, inner sep=0pt](r1)at(2.5,2.5){};
        \node[draw,circle,fill=gray!30,minimum size=6pt, inner sep=0pt](r2)at(5.5,2.5){};
        \node at (2,2.4) {$v_1$};
        \node at (6,2.4) {$v_2$};

        \node[draw,circle,fill=gray!30,minimum size=6pt, inner sep=0pt](r)at(4,1){};
        \node at (4,0.6) {$\rho$};

        \draw[->,>=latex](a1)--(r1) node[midway, auto,swap, left=0mm] {$x_1$};
        \draw[->,>=latex](a2)--(r1) node[midway, auto,swap, right=0mm] {$x_2$};
        \draw[->,>=latex](a3)--(r2) node[midway, auto,swap, right=0mm] {$x_3$};
        \draw[->,>=latex](r1)--(r) node[midway, auto,swap, left=0mm] {$x_1+x_2$};
        \draw[->,>=latex](r2)--(r) node[midway, auto,swap, right=0mm] {$x_3$};
    \end{tikzpicture}
}
\vspace{-1em}
 \caption{\small{The model $(3,2,\Omega^*,T_2)$\\ \qquad with $\Omega^*$ in \textbf{Case 1A}.}}
    \label{T2-m=2-xi=1-fig}
\end{minipage}%
\centering
\begin{minipage}[b]{0.32\textwidth}
\centering
 \begin{tikzpicture}[x=0.6cm]
   \node[draw,circle,fill=gray!30,minimum size=6pt, inner sep=0pt](a1)at(1,4){};
        \node at (1,4.4) {$\sigma_1$};
        \node[draw,circle,fill=gray!30,minimum size=6pt, inner sep=0pt](a2)at(4,4){};
        \node at (4,4.4) {$\sigma_2$};
        \node[draw,circle,fill=gray!30,minimum size=6pt, inner sep=0pt](a3)at(7,4){};
        \node at (7,4.4) {$\sigma_3$};

        \node[draw,circle,fill=gray!30,minimum size=6pt, inner sep=0pt](r1)at(2.5,2.5){};
        \node[draw,circle,fill=gray!30,minimum size=6pt, inner sep=0pt](r2)at(5.5,2.5){};
        \node at (2,2.4) {$v_1$};
        \node at (6,2.4) {$v_2$};

        \node[draw,circle,fill=gray!30,minimum size=6pt, inner sep=0pt](r)at(4,1){};
        \node at (4,0.6) {$\rho$};

        \draw[->,>=latex](a1)--(r1) node[midway, auto,swap, left=0mm] {$x_1$};
        \draw[->,>=latex](a2)--(r2); \node at (4.1,3.5) {$x_2$};
        \draw[->,>=latex](a3)--(r1); \node at (6.2,3.5) {$x_3$};
        \draw[->,>=latex](r1)--(r) node[midway, auto,swap, left=0mm] {$(x_1,x_3)$};
        \draw[->,>=latex](r2)--(r) node[midway, auto,swap, right=0mm] {$x_2$};
    \end{tikzpicture}
    \vspace{-1em}
    \caption{\small{The model $(3,2,\Omega^*,T_2)$\\ \qquad~~ with $\Omega^*$ in \textbf{Case 1B}.}}
    \label{T2-m=2-xi=2-fig}
\end{minipage}
\begin{minipage}[b]{0.32\textwidth}
\centering
 \begin{tikzpicture}[x=0.6cm]
  \node[draw,circle,fill=gray!30,minimum size=6pt, inner sep=0pt](a1)at(1,4){};
    \node at (1,4.4) {$\sigma_1$};
        \node[draw,circle,fill=gray!30,minimum size=6pt, inner sep=0pt](a2)at(4,4){};
        \node at (4,4.4) {$\sigma_2$};
        \node[draw,circle,fill=gray!30,minimum size=6pt, inner sep=0pt](a3)at(7,4){};
        \node at (7,4.4) {$\sigma_3$};

        \node[draw,circle,fill=gray!30,minimum size=6pt, inner sep=0pt](r1)at(2.5,2.5){};
        \node[draw,circle,fill=gray!30,minimum size=6pt, inner sep=0pt](r2)at(5.5,2.5){};
        \node at (2,2.4) {$v_1$};
        \node at (6,2.4) {$v_2$};

        \node[draw,circle,fill=gray!30,minimum size=6pt, inner sep=0pt](r)at(4,1){};
        \node at (4,0.6) {$\rho$};

    \draw[->,>=latex](a1)--(r1);
    \draw[->,>=latex](a2)--(r2);
    \draw[->,>=latex](a3)--(r1);
    \draw[->,>=latex](a3)--(r2);
    \draw[->,>=latex](r1)--(r) node[midway, auto,swap, left=0mm] {$\varphi_1$};
        \draw[->,>=latex](r2)--(r) node[midway, auto,swap, right=0mm] {$\varphi_2$};
    \end{tikzpicture}
    \vspace{-1em}
    \caption{\small{The model $(3,2,\Omega,T_2)$\\ \qquad~~ with $\Omega$ in \textbf{Case 2A}.}}
    \label{T2-m=2-xi=3/2-fig}
\end{minipage}
\vspace{-2em}
\end{figure*}


Next, we consider connectivity states in \textbf{Class 2} according to the following two cases of $\Gamma_{\sigma_1}\cap\Gamma_{\sigma_2}=\emptyset$ and $\Gamma_{\sigma_1}\cap\Gamma_{\sigma_2}\neq\emptyset$.


\textbf{Case 2A:} $\Gamma_{\{\sigma_i,\sigma_j\}}=V$ for all two-index sets $\{i,j\}\subseteq[3]$ and $\Gamma_{\sigma_1}\cap\Gamma_{\sigma_2}=\emptyset$.

Let $\Omega$ be an arbitrary connectivity state in \textbf{Case 2A}. Since $m=2$ and $\Gamma_{\sigma_1}\cap\Gamma_{\sigma_2}=\emptyset$, we have $\big|\Gamma_{\sigma_1}\big|=\big|\Gamma_{\sigma_2}\big|=1$, and by the isomorphism of models, assume that $\Gamma_{\sigma_1}=\{v_1\}$ and $\Gamma_{\sigma_2}=\{v_2\}$. Further, it follows from $\Gamma_{\{\sigma_1,\sigma_3\}}=\Gamma_{\{\sigma_2,\sigma_3\}}=V$ that $\Gamma_{\sigma_3}=V$. Thus, we determine the connectivity state
\[
\Omega=\big(\Gamma_{\sigma_1}=\{v_1\},\,\Gamma_{\sigma_2}=\{v_2\},\Gamma_{\sigma_3}=V\big).
\]
See Fig.\,\ref{T2-m=2-xi=3/2-fig}. By \eqref{def-R2-3/2} in Lemma \ref{prop-lower-bound-for-T2}, we have $\mathcal{C}\big(3,2,\Omega,T_2\big)\geq3/2$.

Next, we construct an admissible $2$-shot source code $\mathbf{C}=\{\varphi_1,\varphi_2;\,\psi\}$ for $(3,2,\Omega,T_2)$ as follows. Let $\boldsymbol{x}_i=(x_{i,1},\,x_{i,2})^\top,\,i=1,2,3$ be the source messages. Define the encoding functions for $v_1$ and $v_2$ as
\begin{align*}
\varphi_1(\boldsymbol{x}_1,\boldsymbol{x}_3)=(x_{1,1},\, x_{1,2},\, x_{3,1})\quad \textup{and}\quad \varphi_2(\boldsymbol{x}_2,\boldsymbol{x}_3)=(x_{2,1},\, x_{2,2},\, x_{3,2}).
\end{align*}
With the received messages $\varphi_1(\boldsymbol{x}_1,\boldsymbol{x}_3)$ and $\varphi_2(\boldsymbol{x}_2,\boldsymbol{x}_3)$, we can compute at the decoder $\rho$
\[
\boldsymbol{x}_S\cdot T_2=
\begin{bmatrix}
\hspace{-0.1mm}x_{1,1}+x_{2,1} &x_{3,1}\\
\hspace{-0.1mm}x_{1,2}+x_{2,2} &x_{3,2}\\
\hspace{-0.1mm}x_{1,3}+x_{2,3} &x_{3,3}
\end{bmatrix}.
\]
We note that $n(\mathbf{C})\mspace{-5mu}=\mspace{-5mu}3$ and thus the coding rate $R(\mathbf{C})\mspace{-5mu}=\mspace{-5mu}n(\mathbf{C})/k\mspace{-5mu}=\mspace{-5mu}3/2$. This implies that $\mathcal{C}(3,2,\Omega,T_2)\mspace{-5mu}\leq\mspace{-4mu}3/2$. Together with the lower bound $3/2$, we have proved that $\mathcal{C}(3,2,\Omega,T_2)\mspace{-5mu}=\mspace{-5mu}3/2$ for any connectivity state~$\Omega$ in \textbf{Case 2A}.

\medskip

\textbf{Case 2B:} $\Gamma_{\{\sigma_i,\sigma_j\}}=V$ for all two-index sets $\{i,j\}\subseteq[3]$ and $\Gamma_{\sigma_1}\cap\Gamma_{\sigma_2}\neq\emptyset$.

For this subcase, by \eqref{def-R2} we can compute the lower bound
\begin{align}\label{equ1}
R_2=\frac{~\textup{Rank}\big(T_2\big)~}{~|V|~}=1.
\end{align}
Let $\Omega$ be an arbitrary connectivity state in this subcase. Since $\Gamma_{\sigma_1}\cap\Gamma_{\sigma_2}\neq\emptyset$, we assume without loss of generality that $v_1\in\Gamma_{\sigma_1}\cap\Gamma_{\sigma_2}$ by the isomorphism of models. We further consider the other encoder $v_2$.
\begin{itemize}
\item If $v_2\notin\Gamma_{\sigma_1}\cap\Gamma_{\sigma_2}$, it follows from $\Gamma_{\{\sigma_1,\sigma_3\}}=\Gamma_{\{\sigma_2,\sigma_3\}}=V$ that $v_2\in\Gamma_{\sigma_3}$.
\item If $v_2\in\Gamma_{\sigma_1}\cap\Gamma_{\sigma_2}$, we have $\Gamma_{\sigma_1}=\Gamma_{\sigma_2}=V$ and either $v_1\in\Gamma_{\sigma_3}$ or $v_2\in\Gamma_{\sigma_3}$. Without loss of generality, we assume that $v_2\in\Gamma_{\sigma_3}$ by the isomorphism of models.\footnote{More precisely, we take $\pi$ to be the identity permutation for $S$ and $\tau=\setlength{\arraycolsep}{3pt}\renewcommand{\arraystretch}{0.7}
\begin{pmatrix}
1 & 2 \\
2 & 1
\end{pmatrix}$ for $V$. Performing $(\pi,\tau)$ to $\Omega$, we obtain that $v_1=v_{\tau(2)}\in\Gamma_{\sigma_3}$ in the connectivity state $\Omega\circ(\pi,\tau)$, which, together with $T_2\circ\pi=T_2$, the models $(3,2,\Omega,T_2)$ and $\big(3,\,2,\,\Omega\circ(\pi,\tau),\,T_2\big)$ are isomorphic.}
\end{itemize}
Accordingly, we see that $v_1\in\Gamma_{\sigma_1}\cap\Gamma_{\sigma_2}$ and $v_2\in\Gamma_{\sigma_3}$ regardless of which case mentioned above the connectivity state $\Omega$ is in. We now consider the connectivity state $\widehat{\Omega}=\big( \Gamma_{\sigma_1}=\Gamma_{\sigma_2}=\{v_1\},\,\Gamma_{\sigma_3}=\{v_2\} \big)$. By recalling \textbf{Case 1A} in \textbf{Class 1} for the case of $m=2$ (see Section \ref{subsec-T2-m=2}), we have $\mathcal{C}(3,2,\widehat{\Omega},T_2)=1$. Together with $\widehat{\Omega}\preceq\Omega$, we have $\mathcal{C}(3,2,\Omega,T_2)\leq\mathcal{C}(3,2,\widehat{\Omega},T_2)=1$ by Lemma \ref{lem-cap-omega'-geq-cap-omega}. We recall the lower bound $R_2=1$ in \eqref{equ1} and have proved that $\mathcal{C}(3,2,\Omega,T_2)=1$ for any connectivity state $\Omega$ in \textbf{Case~2B}.

\vspace{-0.5em}

\subsection{The case of $m=3$}\label{subsec-T2-m=3}

For the case of $m=3$, we write $V=\{v_1,v_2,v_3\}$. We divide all connectivity states into two classes below.
\begin{itemize}
\item \textbf{Class 1:} $\Omega=\big(\Gamma_{\sigma_1},\Gamma_{\sigma_2},\Gamma_{\sigma_3}\big)$ such that $\big|\Gamma_{\{\sigma_i,\sigma_j\}}\big|=1$ for some two-index set $\{i,j\}\subseteq[3]$;
\item \textbf{Class 2:} $\Omega=\big(\Gamma_{\sigma_1},\Gamma_{\sigma_2},\Gamma_{\sigma_3}\big)$ such that $\big|\Gamma_{\{\sigma_i,\sigma_j\}}\big|\geq2$ for all two-index sets $\{i,j\}\subseteq[3]$.
\end{itemize}

\smallskip

We first consider connectivity states in \textbf{Class 1} according to the case of $\big|\Gamma_{\{\sigma_1,\sigma_2\}}\big|=1$ and the case of $\big|\Gamma_{\{\sigma_1,\sigma_3\}}\big|=1$ or $\big|\Gamma_{\{\sigma_2,\sigma_3\}}\big|=1$.


\textbf{Case 1A:} The connectivity state $\Omega$ with $\big|\Gamma_{\{\sigma_1,\sigma_2\}}\big|=1$.

By the isomorphism of models, we assume without loss of generality that $\Gamma_{\{\sigma_1,\sigma_2\}}=\{v_1\}$, i.e., $\Gamma_{\sigma_1}=\Gamma_{\sigma_2}=\{v_1\}$. It follows from $\Gamma_S=V$ that $\sigma_3\rightarrow v_2$ and $\sigma_3\rightarrow v_3$, i.e., $\{v_2,v_3\}\subseteq\Gamma_{\sigma_3}$. Thus, we have $\Gamma_{\{\sigma_1,\sigma_3\}}=\Gamma_{\{\sigma_2,\sigma_3\}}=V$. Now, by \eqref{def-R2} we can compute the lower bound $R_2=1$.
We now recall \textbf{Case 1A} for the case of $m=2$ in Section \ref{subsec-T2-m=2} that for the connectivity state $\widehat{\Omega}=\big( \Gamma_{\sigma_1}=\Gamma_{\sigma_2}=\{v_1\},\,\Gamma_{\sigma_3}=\{v_2\} \big)$, we have $\mathcal{C}(3,2,\widehat{\Omega},T_2)=1$. In fact, $\widehat{\Omega}$ is a subgraph of $\Omega$ considered here and so $\mathcal{C}(3,2,\Omega,T_2)\leq\mathcal{C}(3,2,\widehat{\Omega},T_2)=1$. Together with the lower bound $R_2=1$, we have proved that $\mathcal{C}(3,2,\Omega,T_2)=1$ for any connectivity state $\Omega$ in \textbf{Case 1A}.

\medskip

\textbf{Case 1B:} The connectivity state $\Omega$ with $\big|\Gamma_{\{\sigma_1,\sigma_3\}}\big|=1$ or $\big|\Gamma_{\{\sigma_2,\sigma_3\}}\big|=1$.

We first consider an arbitrary connectivity state $\Omega$ in \textbf{Case 1B} with $\big|\Gamma_{\{\sigma_1,\sigma_3\}}\big|=1$. Without loss of generality, we assume that $\Gamma_{\{\sigma_1,\sigma_3\}}=\{v_1\}$ by the isomorphism of models. Together with $\Gamma_S=V$, we have $\sigma_2\rightarrow v_2$ and $\sigma_2\rightarrow v_3$, i.e., $\{v_2,v_3\}\subseteq\Gamma_{\sigma_2}$. With this, by \eqref{def-R2} we compute the lower bound $R_2=2$.
We recall \textbf{Case 1B} for the case of $m=2$ in Section~\ref{subsec-T2-m=2}, namely that for the connectivity state $\widehat{\Omega}=\big( \Gamma_{\sigma_1}=\Gamma_{\sigma_3}=\{v_1\},\,\Gamma_{\sigma_2}=\{v_2\} \big)$, we have $\mathcal{C}(3,2,\widehat{\Omega},T_2)=2$. Similarly, since $\widehat{\Omega}$ is a subgraph of~$\Omega$, we have $\mathcal{C}(3,3,\Omega,T_2)\leq\mathcal{C}(3,2,\widehat{\Omega},T_2)=2$. Together with the lower bound $R_2=2$, 
we have proved that $\mathcal{C}(3,3,\Omega,T_2)=2$.

For any connectivity state $\Omega$ in \textbf{Case 1B} with $\big|\Gamma_{\{\sigma_2,\sigma_3\}}\big|=1$, we similarly assume $\Gamma_{\{\sigma_2,\sigma_3\}}=\{v_1\}$, and further take the permutation $\pi=\setlength{\arraycolsep}{3pt}\renewcommand{\arraystretch}{0.7}
\begin{pmatrix}
1 & 2 & 3 \\
2 & 1 & 3
\end{pmatrix}$ for $S$ and the identity permutation $\tau$ for $V$. We perform $(\pi,\tau)$ to $\Omega$ and then see that $\Omega\circ(\pi,\tau)$ is a connectivity state in \textbf{Case 1B} with $\big|\Gamma_{\{\sigma_1,\sigma_3\}}\big|=1$. By the discussion in the last paragraph, we prove that
\[
\mathcal{C}\big(3,\,3,\,\Omega,\,T_2\big)
=\mathcal{C}\big(3,\,3,\,\Omega\circ(\pi,\tau),\,T_2\circ\pi\big)
=\mathcal{C}\big(3,\,3,\,\Omega\circ(\pi,\tau),\,T_2\big)=2,
\]
where the second equality follows from $T_2\circ\pi=T_2$.

\medskip

Next, we consider connectivity states in \textbf{Class 2} (cf. the beginning of Section \ref{subsec-T2-m=3}).

\textbf{Case 2A:} Connectivity states $\Omega$ in \textbf{Class 2} such that one of the three conditions below is satisfied:\newline
\rmnum{1}) $\big|\Gamma_{\{\sigma_1,\sigma_3\}}\big|=2$, \rmnum{2}) $\big|\Gamma_{\{\sigma_2,\sigma_3\}}\big|=2$, and \rmnum{3}) $\big|\Gamma_{\sigma_i}\big|=1$ for some $i\in[3]$.

First, by \eqref{def-R2} we can compute the lower bound $R_2=1$. Consider the connectivity state
\[
\Omega^*=\big(\Gamma_{\sigma_1}=\{v_1\},\Gamma_{\sigma_2}=\{v_2\},\Gamma_{\sigma_3}=\{v_3\}\big)
\]
as depicted in Fig.\,\ref{graph-T2-m=3-xi=1}. Clearly, $\Omega^*$ is a connectivity state in \textbf{Case 2A}. For any connectivity state in \textbf{Case~2A}, we can see that there always exists a connectivity state $\Omega'$ with $\big|\Gamma_{\sigma_i}\big|=1$ for all $i\in[3]$ such that $\Omega'\preceq\Omega$. Together with the fact that $\Gamma_S=V$ in $\Omega'$ by the definition of connectivity states, $\Omega'$ is isomorphic to $\Omega^*$. To be specific, we can take the identity permutation $\pi$ for $S$ and a permutation $\tau$ for~$V$ such that $\Omega'=\Omega^*\circ(\pi,\tau)$ and $T_2\circ\pi=T_2$, which implies
\[
\mathcal{C}(3,3,\Omega',T_2)=\mathcal{C}\big(3,3,\Omega^*\circ(\pi,\tau),T_2\circ\pi\big)
=\mathcal{C}(3,3,\Omega^*,T_2).
\]
We further note that $\mathcal{C}(3,3,\Omega',T_2)\geq\mathcal{C}(3,3,\Omega,T_2)$ by Lemma \ref{lem-cap-omega'-geq-cap-omega}, which immediately implies that
\begin{align}\label{m=3-Case-2A}
\mathcal{C}(3,3,\Omega,T_2)\leq\mathcal{C}(3,3,\Omega^*,T_2).
\end{align}
In fact, $\Omega^*$ is the ``minimum'' connectivity state in \textbf{Case 2A} under the isomorphism of models.

In Fig.\,\ref{graph-T2-m=3-xi=1}, we also present an admissible $1$-shot source code $\mathbf{C}$ for $(3,3,\Omega^*,T_2)$ where the coding rate is~$1$. This implies that $\mathcal{C}(3,3,\Omega^*,T_2)\leq1$ and thus $\mathcal{C}(3,3,\Omega,T_2)\leq1$ by \eqref{m=3-Case-2A}. Together with the lower bound $R_2=1$ for this subcase, we have thus proved that $\mathcal{C}(3,3,\Omega,T_2)=1$ for any connectivity state $\Omega$ in \textbf{Case 2A}.

\medskip

\textbf{Case 2B:} Connectivity states $\Omega$ in \textbf{Class 2} such that $\Gamma_{\{\sigma_1,\sigma_3\}}=\Gamma_{\{\sigma_2,\sigma_3\}}=V$ and $\big|\Gamma_{\sigma_i}\big|\geq2,~\forall~i\in[3]$.

We first note that all the connectivity states in \textbf{Class 2} but not in \textbf{Case 2A} are contained in \textbf{Case 2B}.
By \eqref{def-R2} we compute the lower bound for this subcase as
\begin{align}\label{T2-m=3-R2=2/3}
R_2=\frac{~\textup{Rank}\big(T_2\big)~}{~|V|~}=\frac{~2~}{~3~}.
\end{align}
We now consider an arbitrary connectivity state $\Omega$ that satisfies the condition of \textbf{Case 2B}. Clearly, $\Omega$ is in \textbf{Class 2} and in $\Omega$, we have $\big|\Gamma_{\sigma_1}\cap\Gamma_{\sigma_2}\big|\geq1$ as $\big|\Gamma_{\sigma_i}\big|\geq2,\,i=1,2,3$.

\begin{itemize}
\item  \textbf{Case 2B-1:} $\big|\Gamma_{\sigma_1}\cap\Gamma_{\sigma_2}\big|\geq2$ for $\Omega$.

 \quad For the connectivity state $\Omega$, by the conditions of \textbf{Case 2B} and \textbf{Case 2B-1}, there always exists a connectivity state $\Omega^*$ such that $\Omega^*\preceq\Omega$ in which $\Omega^*$ satisfies that 
 $\Gamma_{\{\sigma_1,\sigma_3\}}=\Gamma_{\{\sigma_2,\sigma_3\}}=V$, $\big|\Gamma_{\sigma_i}\big|=2$ for all $i\in[3]$ and $\big|\Gamma_{\sigma_1}\cap\Gamma_{\sigma_2}\big|=2$. By Lemma \ref{lem-cap-omega'-geq-cap-omega}, we have $\mathcal{C}(3,3,\Omega^*,T_2)\geq\mathcal{C}(3,3,\Omega,T_2)$ and thus we only need to prove that $\mathcal{C}(3,3,\Omega^*,T_2)\leq R_2=2/3$.

    \quad For the connectivity state $\Omega^*$, since $\big|\Gamma_{\sigma_1}\cap\Gamma_{\sigma_2}\big|=2$, we assume without loss of generality that $\Gamma_{\sigma_1}\cap\Gamma_{\sigma_2}=\{v_1,v_2\}$ by the isomorphism of models. Together with $\big|\Gamma_{\sigma_1}\big|=\big|\Gamma_{\sigma_2}\big|=2$ for $\Omega^*$, we have $\Gamma_{\sigma_1}=\Gamma_{\sigma_2}=\{v_1,v_2\}$. Further, it follows from $\Gamma_{\{\sigma_1,\sigma_3\}}=V$ that $v_3\in\Gamma_{\sigma_3}$. Together with $\big|\Gamma_{\sigma_3}\big|=2$, we obtain that either $\Gamma_{\sigma_3}=\{v_1,v_3\}$ or $\Gamma_{\sigma_3}=\{v_2,v_3\}$. By the isomorphism of models, we only consider one of both cases, say $\Gamma_{\sigma_3}=\{v_2,v_3\}$. Now, we can write the connectivity state~$\Omega^*$ as follows (see Fig.\,\ref{graph-T2-m=3-xi=2/3}):
    \[
    \Omega^*=\big(\Gamma_{\sigma_1}=\{v_1,v_2\},\,\Gamma_{\sigma_2}=\{v_1,v_2\},
    \Gamma_{\sigma_3}=\{v_2,v_3\}\big).
    \]
    We construct the following admissible $3$-shot source code $\mathbf{C}=\{\varphi_1,\varphi_2,\varphi_3;\,\psi\}$ for $(3,3,\Omega^*,T_2)$ of $n(\mathbf{C})=2$, so that coding rate $2/3$. Let $\boldsymbol{x}_i=(x_{i,1},\,x_{i,2},\,x_{i,3})^\top,\,i=1,2,3$, and
\begin{align*}
&\varphi_1(\boldsymbol{x}_1,\boldsymbol{x}_2)=(x_{1,1}+x_{2,1},~ x_{1,2}+x_{2,2}),\\ &\varphi_2(\boldsymbol{x}_1,\boldsymbol{x}_2,\boldsymbol{x}_3)=(x_{1,3}+x_{2,3},~ x_{3,1}), \\ &\varphi_3(\boldsymbol{x}_3)=(x_{3,2},~ x_{3,3}).
\end{align*}
With the received messages $\varphi_1(\boldsymbol{x}_1,\boldsymbol{x}_2)$, $\varphi_2(\boldsymbol{x}_1,\boldsymbol{x}_2,\boldsymbol{x}_3)$ and $\varphi_3(\boldsymbol{x}_3)$ as above, we can compute at the decoder $\rho$
\[
\boldsymbol{x}_S\cdot T_2=
\begin{bmatrix}
\hspace{-0.1mm}x_{1,1}+x_{2,1} &x_{3,1}\\
\hspace{-0.1mm}x_{1,2}+x_{2,2} &x_{3,2}\\
\hspace{-0.1mm}x_{1,3}+x_{2,3} &x_{3,3}
\end{bmatrix}.
\]
This shows that $\mathcal{C}(3,3,\Omega^*,T_2)\leq2/3$ and so $\mathcal{C}(3,3,\Omega,T_2)\leq2/3$ for any connectivity state $\Omega$ in \textbf{Case 2B-1}. Together with the lower bound $R_2=2/3$ in \eqref{T2-m=3-R2=2/3}, we have obtained that $\mathcal{C}(3,3,\Omega,T_2)=2/3$ for any $\Omega$ in the subcase.

\item \textbf{Case 2B-2:} $\big|\Gamma_{\sigma_1}\cap\Gamma_{\sigma_2}\big|=1$ for $\Omega$.

 \quad For the connectivity state $\Omega$, we assume without loss of generality that $\Gamma_{\sigma_1}\cap\Gamma_{\sigma_2}=\{v_1\}$ by the isomorphism of models. First, we claim that $\big|\Gamma_{\sigma_1}\big|=\big|\Gamma_{\sigma_2}\big|=2$. To see this, we first note that $\big|\Gamma_{\sigma_1}\big|\geq2$ and $\big|\Gamma_{\sigma_2}\big|\geq2$ by the condition of \textbf{Case 2B}. If $\big|\Gamma_{\sigma_1}\big|=3$, i.e., $\Gamma_{\sigma_1}=V$, we have $\big|\Gamma_{\sigma_1}\cap\Gamma_{\sigma_2}\big|\geq2$ by $\big|\Gamma_{\sigma_2}\big|\geq2$, a contradiction to the condition $\big|\Gamma_{\sigma_1}\cap\Gamma_{\sigma_2}\big|=1$. Similarly, we also have $\big|\Gamma_{\sigma_2}\big|=2$. Further, with $\big|\Gamma_{\sigma_1}\cap\Gamma_{\sigma_2}\big|=1$ from the condition of \textbf{Case 2B-2}, either $\sigma_1\rightarrow v_2$ and $\sigma_2\rightarrow v_3$ or $\sigma_1\rightarrow v_3$ and $\sigma_2\rightarrow v_2$. By the isomorphism of models, we only need to consider one of both cases, say $\sigma_1\rightarrow v_2$ and $\sigma_2\rightarrow v_3$, which thus implies that $\Gamma_{\sigma_1}=\{v_1,v_2\}$ and $\Gamma_{\sigma_2}=\{v_1,v_3\}$. Furthermore, since $\Gamma_{\{\sigma_1,\sigma_3\}}=\Gamma_{\{\sigma_2,\sigma_3\}}=V$ by the condition of \textbf{Case 2B}, we obtain that $\{v_2,v_3\}\subseteq\Gamma_{\sigma_3}$. With $\Gamma_{\sigma_1}=\{v_1,v_2\}$ and $\Gamma_{\sigma_2}=\{v_1,v_3\}$ we obtained above, we have either
\begin{align}\label{Omega1}
\Omega=\big(\Gamma_{\sigma_1}=\{v_1,v_2\},~ \Gamma_{\sigma_2}=\{v_1,v_3\},~
\Gamma_{\sigma_3}=\{v_2,v_3\}\big)
\end{align}
or
\begin{align}\label{Omega2}
\Omega=\big(\Gamma_{\sigma_1}=\{v_1,v_2\},~ \Gamma_{\sigma_2}=\{v_1,v_3\},~
\Gamma_{\sigma_3}=V\big).
\end{align}
We recall in Section \ref{section-not-tight} that for the connectivity state $\Omega$ regardless of \eqref{Omega1} or \eqref{Omega2}, we have $\mathcal{C}(3,3,\Omega,T_2)=3/4$. This thus implies that the lower bound $R_2=2/3$ is not tight for \textbf{Case~2B-2}.
\end{itemize}

\begin{figure*}[t]
\tikzstyle{vertex}=[draw,circle,fill=gray!30,minimum size=6pt, inner sep=0pt]
\tikzstyle{vertex1}=[draw,circle,fill=gray!80,minimum size=6pt, inner sep=0pt]
\centering
\begin{minipage}[b]{0.5\textwidth}
\centering
 \begin{tikzpicture}[x=0.6cm]
   \node[draw,circle,fill=gray!30,minimum size=6pt, inner sep=0pt](a1)at(1,4){};
        \node at (1,4.4) {$\sigma_1$};
        \node[draw,circle,fill=gray!30,minimum size=6pt, inner sep=0pt](a2)at(4,4){};
        \node at (4,4.4) {$\sigma_2$};
        \node[draw,circle,fill=gray!30,minimum size=6pt, inner sep=0pt](a3)at(7,4){};
        \node at (7,4.4) {$\sigma_3$};

        \node[draw,circle,fill=gray!30,minimum size=6pt, inner sep=0pt](r1)at(1,2.5){};
        \node[draw,circle,fill=gray!30,minimum size=6pt, inner sep=0pt](r2)at(4,2.5){};
         \node[draw,circle,fill=gray!30,minimum size=6pt, inner sep=0pt](r3)at(7,2.5){};
        \node at (0.4,2.4) {$v_1$};
        \node at (3.4,2.4) {$v_2$};
        \node at (7.6,2.4) {$v_3$};

        \node[draw,circle,fill=gray!30,minimum size=6pt, inner sep=0pt](r)at(4,1){};
        \node at (4,0.6) {$\rho$};

        \draw[->,>=latex](a1)--(r1) node[midway, auto,swap, left=0mm] {$x_1$};
        \draw[->,>=latex](a2)--(r2) node[midway, auto,swap, left=0mm] {$x_2$};
        \draw[->,>=latex](a3)--(r3) node[midway, auto,swap, right=0mm] {$x_3$};
        \draw[->,>=latex](r1)--(r) node[midway, auto,swap, left=0mm] {$x_1$};
        \draw[->,>=latex](r2)--(r) node[midway, auto,swap, left=0mm] {$x_2$};
        \draw[->,>=latex](r3)--(r) node[midway, auto,swap, right=0mm] {$x_3$};
    \end{tikzpicture}
    \vspace{-1em}
      \caption{The model $(3,3,\Omega^*,T_2)$\\ \qquad~~  with $\Omega^*$ in \textbf{Case 2A}.}
   \label{graph-T2-m=3-xi=1}
\end{minipage}%
\centering
\begin{minipage}[b]{0.5\textwidth}
\centering
{
 \begin{tikzpicture}[x=0.6cm]
     \node[draw,circle,fill=gray!30,minimum size=6pt, inner sep=0pt](a1)at(1,4){};
        \node at (1,4.4) {$\sigma_1$};
        \node[draw,circle,fill=gray!30,minimum size=6pt, inner sep=0pt](a2)at(4,4){};
        \node at (4,4.4) {$\sigma_2$};
        \node[draw,circle,fill=gray!30,minimum size=6pt, inner sep=0pt](a3)at(7,4){};
        \node at (7,4.4) {$\sigma_3$};

        \node[draw,circle,fill=gray!30,minimum size=6pt, inner sep=0pt](r1)at(1,2.5){};
        \node[draw,circle,fill=gray!30,minimum size=6pt, inner sep=0pt](r2)at(4,2.5){};
         \node[draw,circle,fill=gray!30,minimum size=6pt, inner sep=0pt](r3)at(7,2.5){};
        \node at (0.4,2.4) {$v_1$};
        \node at (3.4,2.4) {$v_2$};
        \node at (7.6,2.4) {$v_3$};

        \node[draw,circle,fill=gray!30,minimum size=6pt, inner sep=0pt](r)at(4,1){};
        \node at (4,0.6) {$\rho$};

        \draw[->,>=latex](a1)--(r1);
        \draw[->,>=latex](a1)--(r2);
        \draw[->,>=latex](a2)--(r1);
        \draw[->,>=latex](a2)--(r2);
        \draw[->,>=latex](a3)--(r2);
        \draw[->,>=latex](a3)--(r3);
        \draw[->,>=latex](r1)--(r) node[midway, auto,swap, left=0mm] {$\varphi_1$};
        \draw[->,>=latex](r2)--(r) node[midway, auto,swap, left=0mm] {$\varphi_2$};
        \draw[->,>=latex](r3)--(r) node[midway, auto,swap, right=0mm] {$\varphi_3$};
    \end{tikzpicture}
}
\vspace{-1em}
  \caption{The model $(3,3,\Omega^*,T_2)$\\ \qquad~~ with $\Omega^*$ in \textbf{Case 2B-1}.}
        \label{graph-T2-m=3-xi=2/3}
\end{minipage}
\vspace{-4em}
\end{figure*}

\medskip

To end this section, we summarize all the capacities for the model $(3,m,\Omega,T_2)$ as follows.

\begin{theorem}\label{cap-T2-rewrite}
Consider a model $\big(3,m,\Omega,T_2\big)$, where $1\leq m\leq 3$ and $\Omega=\big(\Gamma_{\sigma_1},\Gamma_{\sigma_2},\Gamma_{\sigma_3}\big)$ is an arbitrary connectivity state.

\begin{itemize}
  \item For $m=1$, $\mathcal{C}(3,1,\Omega,T_2)=2$.
  \item For $m=2$,
  \begin{enumerate}
\item if $\big|\Gamma_{\{\sigma_1,\sigma_3\}}\big|=1$ or $\big|\Gamma_{\{\sigma_2,\sigma_3\}}\big|=1$ in $\Omega$ \textup{(cf.~\textbf{Cases 1B} and \textbf{1C} in Section~\ref{subsec-T2-m=2})}, then
$\mathcal{C}\big(3,2,\Omega,T_2\big)=2$;
\item if $\Gamma_{\{\sigma_i,\sigma_j\}}=V$ for all two-index sets $\{i,j\}\subseteq[3]$ and $\Gamma_{\sigma_1}\cap\Gamma_{\sigma_2}=\emptyset$ in $\Omega$  \textup{(cf. \textbf{Case~2A} in Section \ref{subsec-T2-m=2})}, then
$\mathcal{C}\big(3,2,\Omega,T_2\big)=3/2$;
\item otherwise, namely that in $\Omega$, either $\big|\Gamma_{\{\sigma_1,\sigma_2\}}\big|=1$ or $\Gamma_{\{\sigma_i,\sigma_j\}}=V$ for all two-index sets $\{i,j\}\subseteq[3]$ with $\Gamma_{\sigma_1}\cap\Gamma_{\sigma_2}\neq\emptyset$ \textup{(cf. \textbf{Cases 1A} and \textbf{2B} in Section \ref{subsec-T2-m=2})}, then
$\mathcal{C}\big(3,2,\Omega,T_2\big)=1$.
\end{enumerate}
  \item For $m=3$,
\begin{enumerate}
\item if $\big|\Gamma_{\{\sigma_1,\sigma_3\}}\big|=1$ or $\big|\Gamma_{\{\sigma_2,\sigma_3\}}\big|=1$ in $\Omega$ \textup{(cf.~\textbf{Case 1B} in Section~\ref{subsec-T2-m=3})}, then
 $\mathcal{C}\big(3,3,\Omega,T_2\big)=2$;
\item if in $\Omega$, either $\big|\Gamma_{\{\sigma_1,\sigma_2\}}\big|=1$ or $\big|\Gamma_{\{\sigma_i,\sigma_j\}}\big|\geq2$ for all two-index sets $\{i,j\}\subseteq[3]$ such that one of the following three conditions is satisfied: \textup{\rmnum{1})} $\big|\Gamma_{\{\sigma_1,\sigma_3\}}\big|=2$, \textup{\rmnum{2})}  $\big|\Gamma_{\{\sigma_2,\sigma_3\}}\big|=2$, and \textup{\rmnum{3})} $\big|\Gamma_{\sigma_i}\big|=1$ for some $i\in[3]$ \textup{(cf.~\textbf{Cases 1A} and \textbf{2A} in Section \ref{subsec-T2-m=3})}, then
$\mathcal{C}\big(3,3,\Omega,T_2\big)=1$;
\item if in $\Omega$, $\big|\Gamma_{\sigma_1}\cap\Gamma_{\sigma_2}\big|\geq2$, $\Gamma_{\{\sigma_1,\sigma_3\}}\!=\!\Gamma_{\{\sigma_2,\sigma_3\}}\!=\!V$ and $\big|\Gamma_{\sigma_3}\big|\geq2$ \textup{(cf.~\textbf{Case 2B-1} in Section~\ref{subsec-T2-m=3})}, then
    $\mathcal{C}\big(3,3,\Omega,T_2\big)=2/3$;
\item otherwise, namely that in $\Omega$,  $\big|\Gamma_{\sigma_1}\cap\Gamma_{\sigma_2}\big|=1$, $\big|\Gamma_{\sigma_1}\big|=\big|\Gamma_{\sigma_2}\big|=2$ and $\Gamma_{\{\sigma_1,\sigma_3\}}=\Gamma_{\{\sigma_2,\sigma_3\}}=V$ \textup{(cf.~\textbf{Case 2B-2} in Section \ref{subsec-T2-m=3})}, then
    $\mathcal{C}\big(3,3,\Omega,T_2\big)=3/4$.
\end{enumerate}
\end{itemize}
\end{theorem}

\section{Conclusion}\label{conclusion}

We put forward the distributed source coding problem for function compression in this paper. We explicitly characterized the function-compression capacities of all the distributed source coding models $(s,m,\Omega,T)$ for compressing vector-linear functions with three sources and no more than three encoders (i.e., $s=3$ and $1\leq m\leq 3$), where the connectivity state $\Omega$ and the vector-linear function $T$ are arbitrary. In particular, in characterizing the function-compression capacities for the two most nontrivial models, we developed a novel approach by not only upper bounding but also lower bounding the size of image sets of encoding functions, which thus implies an improved converse proof. The results thus obtained can be applied in network function computation to show that the best known general upper bound proved by Guang \emph{et. al.}~\cite{Guang_Improved_upper_bound} on computing capacity is in general not tight for computing vector-linear functions. This thus answers the open problem that whether this bound is always tight or not.

For the model considered in the current paper, several interesting problems still remain open, such as how to generalize the developed approach for the two most nontrivial models to improve the general upper bounds previously obtained on the computing capacity in network function computation; and whether the results obtained and the techniques developed can be applied to characterize the function-compression capacities for more general models with $s\geq3$ and $m\geq4$.


\begin{thebibliography}{99}

\bibitem{Yeung_distributed_source_coding}
R. W. Yeung and Z. Zhang, ``Distributed source coding for satellite communications,''
\textit{IEEE Trans. Inf. Theory}, vol.~45, no.~4, pp.~3100--3117, May~1999.

%
%
%
%
%




%
%
%
%







\bibitem{Roche_symmetrical}
J. R. Roche, R. W. Yeung, and K. P. Hau, ``Symmetrical multilevel
diversity coding,''
\textit{IEEE Trans. Inf. Theory}, vol.~43, no.~3, pp.~1059-1064, May~1997.

\bibitem{Yeung_distortion}
R. W. Yeung, ``Multilevel diversity coding with distortion,''
\textit{IEEE Trans. Inf. Theory}, vol.~41, no.~2, pp.~412-422, May~1995.

\bibitem{Yeung_symmetrical}
R. W. Yeung and Z. Zhang, ``On symmetrical multilevel diversity coding,''
\textit{IEEE Trans. Inf. Theory}, vol.~45, no.~2, pp.~609-621, Mar.~1999.




\bibitem{Network_information_flow}
R. Ahlswede, N. Cai, S.-Y. R. Li, and R. W. Yeung, ``Network information flow,''
\textit{IEEE Trans. Inf. Theory}, vol. 46, no. 4, pp. 1204-1216, July 2000.


\bibitem{linear}
S.-Y.~R. Li, R.~W. Yeung, and N.~Cai, ``Linear network coding,'' \textit{IEEE
Trans. Inf. Theory}, vol.~49, no.~2, pp. 371--381, Feb. 2003.

\bibitem{alg}
R.~Koetter and M.~M\'{e}dard, ``An algebraic approach to network coding,'' \textit{IEEE/ACM Trans. Networking}, vol.~11, no.~5, pp. 782--795, Oct. 2003.

\bibitem{Zhang-book}
R. W. Yeung, S.-Y. R. Li, N. Cai, and Z. Zhang, ``Network coding theory,''
\textit{Foundations and Trends in Communications and Information Theory}, vol. 2, nos.4 and 5, pp. 241-381, 2005.

\bibitem{yeung08b}
R.~W.~Yeung, \emph{Information Theory and Network Coding}.\hskip 1em plus 0.5em
  minus 0.4em\relax Springer, 2008.





%
%
%
%
%
%




%
%
%
\bibitem{Korner-Marton-IT73}
J.~K\"{o}rner and K.~Marton, ``How to encode the modulo-two sum of binary
sources (Corresp.),'' \textit{IEEE Trans. Inf. Theory}, vol.~25, no.~2, pp.~219--221,
Mar.~1979.
%
%
\bibitem{Doshi_fun_comp_graph_color_sch}
V. Doshi, D. Shah, M. M{\'e}dard, and M. Effros, ``Functional compression through graph coloring,''
\textit{IEEE Trans. Inf. Theory}, vol.~56, no.~8, pp.~3901-3917, Aug.~2010.
%
%
\bibitem{Feizi-Medard}
S. Feizi and M. M{\'e}dard, ``On network functional compression,''
\textit{IEEE Trans. Inf. Theory}, vol.~60, no.~9, pp.~5387-5401, Sept.~2014.
%
%
\bibitem{Orlitsky-Roche_general_side_inf_model_rat_reg}
A. Orlitsky and J. R. Roche, ``Coding for computing,''
     \textit{IEEE Trans. Inf. Theory}, vol.~47, no.~3, pp.~903-917, Mar.~2001.
%
%
\bibitem{Witsenhausen-IT-76}
H.~Witsenhausen, ``The zero-error side information problem and chromatic numbers (corresp.),'' \textit{IEEE Trans. Inf. Theory}, vol.~22, no.~5, pp.~592-593, Sept. 1976.
%
%
\bibitem{Alon-Orlitsky_source_cod_graph_entropies}
N. Alon and A. Orlitsky, ``Source coding and graphs entropies,''
    \textit{IEEE Trans. Inf. Theory}, vol.~42, no.~5, pp.~1329-1339, Sept.~1996.
%
%
%
%
%


%


\bibitem{Guang_Zhang_Arithmetic_sum_TIT}
X. Guang and R. Zhang, ``Zero-error distributed compression of binary arithmetic sum,''
\textit{IEEE Trans. Inf. Theory}, vol.~70, no.~5, pp.~1111-1120, May~2024.




%
%
%
%
%
%
%
%
%
%





\bibitem{Appuswamy11}
R.~Appuswamy, M.~Franceschetti, N.~Karamchandani, and K.~Zeger, ``Network coding for computing: cut-set bounds,'' \textit{IEEE Trans. Inf. Theory}, vol.~57, no.~2, pp.~1015--1030, Feb.~2011.

%



\bibitem{Huang_Comment_cut_set_bound}
C. Huang, Z. Tan, S. Yang, and X. Guang, ``Comments on cut-set bounds on network function computation,''
 \textit{IEEE Trans. Inf. Theory}, vol.~64, no.~9, pp.~6454--6459, Sept.~2018.




\bibitem{Guang_Improved_upper_bound}
X. Guang, R. W. Yeung, S. Yang and C. Li, ``Improved upper bound on the network function computing capacity,''
\textit{IEEE Trans. Inf. Theory}, vol.~65, no.~6, pp.~3790--3811, June~2019.




\bibitem{Appuswamy-lin-func-lin-code}
R.~Appuswamy and M.~Franceschetti, ``Computing linear functions by linear coding over networks,'' \textit{IEEE Trans. Inf. Theory}, vol.~60, no.~1, pp.~422--431, Jan.~2014.


\bibitem{Li_Xu_vector_linear_diamond}
D. Li and Y. Xu, ``Computing vector-linear functions on diamond network,''
\textit{IEEE Commun. Lett.}, vol.~26, no.~7, pp.~1519-1523, July~2022.






\bibitem{Guang_Zhang_Arithmetic_sum_Sel_Areas}
R. Zhang, X. Guang, S. Yang, X. Niu, and B. Bai, ``Computation of binary arithmetic sum over an asymmetric diamond network,''
\textit{IEEE J. Sel. Areas Inf. Theory}, vol.~5, pp.~585--596, 2024.






\bibitem{Slepian-Wolf-IT73}
D.~Slepian and J.~K.~Wolf, ``Noiseless coding of correlated information sources,'' \textit{IEEE Trans. Inf. Theory}, vol.~19, no.~4, pp.~471--480, July~1973.




\bibitem{Koulgi_zero-error-cod_cor_inf_sour}
P. Koulgi, E. Tuncel, S. Regunathan, and K. Rose, ``On zero-error coding of correlated sources,''
    \textit{IEEE Trans. Inf. Theory}, vol.~49, no.~11, pp.~2856-2873, Nove.~2003.






\bibitem{Wyner-Ziv_rat_distortion_fun_sid_inf}
A. Wyner and J. Ziv, ``The rate-distortion function for source coding
with side information at the decoder,''
\textit{IEEE Trans. Inf. Theory}, vol.~22, no.~1, pp.~1-10, Jan.~1976.


\bibitem{Yamamoto_rat_distortion_gener_fun_sid_inf}
H. Yamamoto, ``Wyner-Ziv theory for a general function of the
correlated sources (corresp.),''
 \textit{IEEE Trans. Inf. Theory}, vol.~28, no.~5,
pp.~803-807, Sept.~1982.


\bibitem{Berger-Yeung_multi_source_cod_dist_cri}
T. Berger and R. W. Yeung, ``Multiterminal source encoding with
one distortion criterion,''
 \textit{IEEE Trans. Inf. Theory}, vol.~35, no.~2, pp.~228-236, Mar.~1989.





\end{thebibliography}
\end{document}